\documentclass[a4paper]{article}

\usepackage[latin9]{inputenc}
\usepackage{mathrsfs}
\usepackage{bm}
\usepackage{amsmath}
\usepackage{amsthm}
\usepackage{amssymb}
\usepackage{graphicx}
\usepackage{esint}

\makeatletter
\theoremstyle{plain}

\theoremstyle{plain}

\ifx\proof\undefined
\newenvironment{proof}[1][\protect\proofname]{\par
	\normalfont\topsep6\p@\@plus6\p@\relax
	\trivlist
	\itemindent\parindent
	\item[\hskip\labelsep\scshape #1]\ignorespaces
}{%
	\endtrivlist\@endpefalse
}
\providecommand{\proofname}{Proof}
\fi

\usepackage{jhepmod}
\pdfoutput=1

\usepackage[table ]{ xcolor}

\usepackage{amssymb}
\usepackage{graphicx,color}

\usepackage{amsmath,amsthm}

\usepackage{mathrsfs}

\usepackage{bm}

\def\beq{\begin{equation}}
\def\eeq{\end{equation}}
\def\bi{\begin{itemize}}
\def\ei{\end{itemize}}
	\def\ba{\begin{array}}
	\def\ea{\end{array}}
	\def\bfig{\begin{figure}}
	\def\efig{\end{figure}}

	\def\C{\mathbb{C}}
	\def\R{\mathbb{R}}
	\def\Z{\mathbb{Z}}

	\newtheorem{theorem}{Theorem}[section]

	\newtheorem{lemma}[theorem]{Lemma}

	\newcommand{\Slc}{\mathrm{SL}(2,\mathbb{C})}
	\newcommand{\PSlc}{\mathrm{PSL}(2,\mathbb{C})}

	\def\be{\begin{eqnarray}}
	\def\ee{\end{eqnarray}}

	\newcommand{\ca}{\mathcal A}

	\newcommand{\cc}{\mathcal C}
	\newcommand{\cd}{\mathcal D}
	
	\newcommand{\cf}{\mathcal F}
	
	\newcommand{\ch}{\mathcal H}

	\newcommand{\cl}{\mathcal L}

	\newcommand{\co}{\mathcal O}

	\newcommand{\cs}{\mathcal S}
	\newcommand{\ct}{\mathcal T}
	\newcommand{\cu}{\mathcal U}
	\newcommand{\cv}{\mathcal V}

	  \newcommand{\Fd}{\mathfrak{D}}

	\newcommand{\fp}{\mathfrak{p}}

	  \newcommand{\Fu}{\mathfrak{U}}
	  
	  \newcommand{\Fw}{\mathfrak{W}}

	\renewcommand{\a}{\alpha}
	\renewcommand{\b}{\beta}
	\newcommand{\g}{\gamma}

	\newcommand{\eps}{\varepsilon}

	\newcommand{\sig}{\sigma}
	
	\renewcommand{\l}{\lambda}
	\renewcommand{\L }{\Lambda}
	\renewcommand{\o}{\omega}

	\newcommand{\rmd}{\mathrm d}

	\newcommand{\lt}{\left}
	\newcommand{\rt}{\right}

	\newcommand{\lag}{\left\langle}
	\newcommand{\rag}{\right\rangle}

	\newcommand{\sw}{\mathscr{W}}

	\newcommand{\css}{\mathscr{S}}
	\newcommand{\re}{\mathrm{Re}}
	\newcommand{\im}{\mathrm{Im}}

	\newcommand{\tr}{\mathrm{Tr}}

	\newcommand{\bmY}{\bm{Y}}
	
	\newcommand{\bmy}{\bm{y}}
	
	\newcommand{\bmL}{\bm{L}}
	\newcommand{\bmU}{\bm{U}}
\newcommand{\bfq}{\bm{q}}
\newcommand{\rmD}{\mathrm{D}}

\newcommand{\suquqt}{\mathscr{U}_{\bfq}(sl_2)\otimes \mathscr{U}_{\tilde{\bfq}}(sl_2)}

\newcommand{\suq}{\mathscr{U}_{\bfq}(sl_2)}
\newcommand{\suqt}{\mathscr{U}_{\tilde \bfq}(sl_2)}

\title{Representations of a quantum-deformed Lorentz algebra, Clebsch-Gordan map, and Fenchel-Nielsen representation of complex Chern-Simons theory at level-$\bm{N}$}

\author[1,2]{Muxin Han}

\affiliation[1]{Department of Physics, Florida Atlantic University, 777 Glades Road, Boca Raton, FL 33431-0991, USA}

\affiliation[2]{Department Physik, Institut f\"ur Quantengravitation, Theoretische Physik III, Friedrich-Alexander Universit\"at Erlangen-N\"urnberg, Staudtstr. 7/B2, 91058 Erlangen, Germany}

\emailAdd{hanm(At)fau.edu}

\abstract{
A family of infinite-dimensional irreducible $*$-representations on $\mathcal{H}\simeq L^2(\mathbb{R})\otimes\mathbb{C}^N$ is defined for a quantum-deformed Lorentz algebra $\mathscr{U}_{\bfq}(sl_2)\otimes \mathscr{U}_{\widetilde{\bfq}}(sl_2)$, where ${\bfq}=\exp[\frac{\pi i}{N}(1+b^2)]$ and $\tilde{{\bfq}}=\exp[\frac{\pi i}{N}(1+b^{-2})]$ with $N\in\mathbb{Z}_+$ and $|b|=1$. The representations are constructed with the irreducible representation of quantum torus algebra at level-$N$, which is developed from the quantization of $\mathrm{SL}(2,\mathbb{C})$ Chern-Simons theory. We study the Clebsch-Gordan decomposition of the tensor product representation, and we show that it reduces to the same problem as diagonalizing the complex Fenchel-Nielson length operators in quantizing $\mathrm{SL}(2,\mathbb{C})$ Chern-Simons theory on 4-holed sphere. Finally, we explicitly compute the spectral decomposition of the complex Fenchel-Nielson length operators and the corresponding direct-integral representation of the Hilbert space $\mathcal{H}$, which we call the Fenchel-Nielson representation.

}

\keywords{}

\makeatother

\usepackage{babel}
\providecommand{\lemmaname}{Lemma}
\providecommand{\theoremname}{Theorem}

\begin{document}

\maketitle

\section{Introduction}

This work is partly inspired by the early results on the relations between the modular double of $U_q(sl(2,\R))$ and quantum Teichm\"ueller theory  \cite{Kashaev2001,Kashaev:2000ku,Bytsko:2002br,Nidaiev:2013bda,Ponsot:2000mt,Teschner:2005bz,Teschner:2003em,Derkachov:2013cqa}. As has been shown in the literature, the representation of the modular double of $U_q(sl(2,\R))$ can be defined on $L^2(\R)$, where the representations of the generators relate to the representation of quantum torus algebra (composed by the generators of Weyl algebra). For the tensor product representation, the Clebsch-Gordan decomposition is equivalent to the spectral decomposition of certain Fenchel-Nielsen (FN) length operator in quantum Teichm\"ueller theory. These results find their generalizations in this paper. 

The quantum Teichm\"ueller theory closely relates to the $\mathrm{SL}(2,\R)$ Chern-Simons theory \cite{Killingback:1990hi,Teschner:2014nja}. There has been recent generalization in \cite{Andersen2014,andersen2016level,levelk} to the Teichm\"ueller TQFT of integer level, which relates to the quantization of Chern-Simons theory with complex gauge group $\Slc$. The quantum $\Slc$ Chern-Simons theory has the complex coupling constant $N+is$ and $N-is$ where $N\in\Z_+$ is called the level of Chern-Simons theory, and $s\in\R$ corresponds to one of the unitary branch \cite{Witten1991}. We use the parametrization 
\be    
is=N\frac{1-b^2}{1+b^2},\qquad |b|=1,\quad \mathrm{Re}(b)>0. 
\ee
and define $\hbar=\frac{\pi i}{N}(1+b^2),\ \tilde{\hbar}=\frac{\pi i}{N}(1+b^{-2})$. $N$ relates to the integer level of the Teichm\"ueller TQFT \cite{andersen2016level}. The quantization of $\Slc$ Chern-Simons theory results in the Weyl algebra and quantum torus algebra at level-$N$, motivated by quantizing the Chern-Simon symplectic structure \cite{levelk,Andersen2014}. The level-$N$ quantum torus algebra has $\bfq=\exp(\hbar)$ and $\tilde{\bfq}=\exp(\tilde{\hbar})$, and the Weyl algebra has $\bfq^2$ and $\tilde{\bfq}^2$. The Hilbert space carrying their representations is $\ch\simeq L^2(\R)\otimes \C^N$, and the representation reduces to the representation in quantum Teichm\"ueller theory \cite{Teschner:2005bz,goncharov2022quantum} when $N=1$. The representation of quantum torus algebra is reviewed in Section \ref{Quantum torus algebra and the representation}.

Based on the representation of the quantum torus algebra, we construct a family of irreducible $*$-representation of $\suquqt$ on $\ch$. The $*$-structure on $\suquqt$ is represented by the hermitian conjugate on $\ch$. As the tensor product of two Hopf algebras, $\suquqt$ has a well-defined Hopf algebra structure and can be understood as a quantum deformation of the Lorentz algebra with complex deformation parameter $\bfq,\tilde{\bfq}$ \footnote{For $N=1$, the deformation parameters reduces to $\bfq\to-e^{\pi i b^2}$ and $\tilde{\bfq}\to -e^{\pi i b^{-2}}$. In this case, $\suquqt$ may be compared with the modular double of $U_{\bfq}(sl(2,\R))$ \cite{Derkachov:2013cqa,Kashaev2001,Nidaiev:2013bda}, although the deformation parameter is different by a flip of sign.}. In particular, it is different from the quantum Lorentz group with real $q$ that is well-studied in the literature (see e.g. \cite{MR1059324,BR}). The irreducible representations of $\suquqt$ constructed in this paper are parametrized by a continuous parameter $\mu_L\in\R$ and a discrete parameter $m_L\in\Z/N\Z$. They may be viewed as analog with the principle-series unitary representation of $\Slc$.

We study the Clebsch-Gordan decomposition of the tensor product $*$-representation (of $\suquqt$) on $\ch\otimes\ch$, and we show that the result is a direct-integral of irreducible $*$-representations (see Section \ref{The Clebsch-Gordan decomposition}). The direct-integral is given by the spectral decomposition of the Casimir operator $Q_{21}$ of the tensor product representation, and we find the unitary transformation $\cu_{21}$ as the Clebsch-Gordan map representing the co-multiplication.

Interestingly $Q_{21}$ closely relates to the quantization of complex FN lengths for $\Slc$ flat connections on 4-holed sphere (see Section \ref{Quantum flat connections on 4-holed sphere}). As resulting from quantizing the $\Slc$ Chern-Simons theory, the quantization of $\Slc$ flat connections on 4-holed sphere can be constructed based on the level-$N$ representation of quantum torus algebra. The Hilbert space $\ch\simeq L^2(\R)\otimes\C^N$ carries the irreducible representation of the quantum algebra from quantizing the Fock-Goncharov (FG) coordinates of flat connections. The complex FN length that relates to $Q_{21}$ is given by the trace of holonomies around two holes. It turns out that the quantization of the complex FN length leads to the normal operators $\bm{L},\tilde{\bm{L}}=\bm{L}^\dagger$, and $ \mathrm{id}_{\ch}\otimes \bmL$ is unitary equivalent to $Q_{21}$ on $\ch\otimes\ch$. Moreover, we show that the traces of S-cycle (enclosing the 1st and 2nd holes) and T-cycle (enclosing the 2nd and 3rd holes) holonomies are related by a unitary transformation, which is a realization of the A-move in the Moore-Seiberg groupoid (as a generalizaton from \cite{Teschner:2003em}). In addition, we show that the quantization of holonomies' traces on $\ch$ gives a representation of the operator algebra from the skein quantization \cite{Coman:2015lna} of the flat connections.

We show in Section \ref{Eigenstates of the trace operators} that the spectral decomposition of $\bm{L},\tilde{\bm{L}}$ endows to $\ch $ the direct-integral representation, which we call the FN representation:
\be
\ch\simeq \bigoplus_{m_r\in\Z/N\Z}\int^\oplus_{\R_{\geq 0}}\rmd\mu_r\varrho(\mu_r,m_r)^{-1}\,\ch_{\mu_r,m_r},\label{DID0000}
\ee
where $\rmd\mu_r\varrho(\mu_r,m_r)^{-1}$ is the spectral measure and each $\ch_{\mu_r,m_r}$ is 1-dimensional. The spectra of $\bm{L},\tilde{\bm{L}}$ are respectively $\ell(r)=r+r^{-1}$ and $\ell(r)^*$, where $r=\exp[\frac{2\pi i}{N}(-ib\mu_r-m_r)]$. Due to the relation between $\bmL$ and $Q_{21}$, the direct-integral representation \eqref{DID0000} also gives the Clebsch-Gordan decomposition for the tensor product representation of $\suquqt$.

The results from this work should have impact on the complex Chern-Simons theory at level-$N$ and its relation to quantum group and quantum Teichm\"ueller theory. For instance, although the quantum Lorentz group with real deformation and the representation theory has been widely studied, the generalization to complex $\bfq$ has not been studied in the literature before. We show in this paper that this generalization closely relates to the $\Slc$ Chern-Simons theory with level-$N$ \footnote{The quantum Lorentz group with real deformation relates to the $\Slc$ Chern-Simons theory with $N=0$ \cite{BNR,Gaiotto:2024osr}.}. As another aspect, given the relation between quantum Teichm\"ueller theory and Liouville conformal field theory, our study might point toward certain generalization of Liouville conformal field theory relating to the level-k, and this generalization might also relate to the boundary field theory of the $\Slc$ Chern-Simons theory.

In addition to the above, another motivation of this work is the potential application to the spinfoam model with cosmological constant \cite{Han:2021tzw,HHKR,HHKRshort}. The spinfoam model is formulated with $\Slc$ Chern-Simons theory ($1/N$ is proportional to the absolute value of cosmological constant) with the special boundary condition called the simplicity constraint, which restricts the flat connections on 4-holed sphere to be SU(2) up to conjugation. Classically, the simplicity constraint is conveniently formulated in terms of the complex FN variables \cite{Han:2023hbe}. Therefore, it may be convenient to formulate the quantization of the simplicity constraint in the FN representation of quantum flat connections. The investigation on this perspective will be reported elsewhere. 

The structure of this paper is as follows: In Section \ref{Quantum torus algebra and the representation}, we review briefly the representation of quantum torus algebra at level-$k$ and set up some notations. In section \ref{Representation of a q-deformed Lorentz algebra}, we construct the representations of $\suquqt$ and discussion the Clebsch-Gordan decomposition of tensor product representation. In Section \ref{Quantum flat connections on 4-holed sphere}, we discuss the quantization of $\Slc$ flat connections on 4-holed sphere with a certain ideal triangulation, the relation with the skein quantization, the S-cycle and T-cycle trace operators and their unitary transformations. In Section \ref{Eigenstates of the trace operators}, we compute the eigenvalue and distributional eigenstates of the trace operators and prove the direct-integral decomposition of the Hilbert space. In Section \ref{Changing triangulation}, we discuss the unitary transformation induced by changing ideal triangulation of the 4-holed sphere.

\section{Quantum torus algebra and the representation at level $N$}\label{Quantum torus algebra and the representation}

The quantum torus algebra $\mathcal{O}_{\bfq}$ is spanned
by Laurent polynomials of the symbols $\bm{u}_{\alpha,\beta}$ with
$\alpha,\beta\in\mathbb{Z}$, satisfying the following relation
\begin{equation}
\bm{u}_{\alpha,\beta}\bm{u}_{\gamma,\delta}=\bfq^{\alpha\delta-\beta\gamma}\bm{u}_{\alpha+\gamma,\beta+\delta},\qquad \bfq=e^{\hbar}.\label{eq:quantumtorus1}
\end{equation}
where $\hbar\in\mathbb{C}$ is the quantum deformation parameter.
We associated to $\mathcal{O}_{\bfq}$ the ``anti-holomorphic
partner'' $\mathcal{O}_{\tilde{\bfq}}$ generated
by $\tilde{\bm{u}}_{\alpha,\beta}$ with $\alpha,\beta\in\mathbb{Z}$,
satisfying
\begin{equation}
\tilde{\bm{u}}_{\alpha,\beta}\tilde{\bm{u}}_{\gamma,\delta}=\tilde{\bfq}^{\alpha\delta-\beta\gamma}\tilde{\bm{u}}_{\alpha+\gamma,\beta+\delta},\qquad\tilde{\bfq}=e^{\tilde{\hbar}},\qquad\tilde{\hbar}=-\hbar^{*},\label{eq:quantumtorus2}
\end{equation}
and $\mathcal{O}_{\tilde{\bfq}}$ commutes with $\mathcal{O}_{\bfq}$.
The entire algebra is denoted by $\mathcal{A}_{\hbar}=\mathcal{O}_{\bfq}\otimes\mathcal{O}_{\tilde{\bfq}}$.
We can endow the algebra a $\star$-structure by 
\[
\star\left(\bm{u}_{\alpha,\beta}\right)=\tilde{\bm{u}}_{\alpha,\beta},\qquad\star\left(\tilde{\bm{u}}_{\alpha,\beta}\right)=\bm{u}_{\alpha,\beta}.
\]
which interchanges the holomorphic and antiholomorphic copies.

In this paper, we use the following parametrizations of $\hbar$ and $\tilde{\hbar}$
\[
\hbar=\frac{\pi i}{N}\left(1+b^{2}\right),\qquad\tilde{\hbar}=\frac{\pi i}{N}\left(1+b^{-2}\right),
\]
where $N,b$ satsifies 
\[
N\in\mathbb{Z}_+,\qquad|b|=1,\qquad\mathrm{Re}(b)>0,\qquad\mathrm{Im}(b)>0.
\]


It has been proposed in \cite{levelk,andersen2016level} an infinite-dimensional
unitary irreducible representation of $\mathcal{A}_{\hbar}$ as the quantization of the Chern-Simons theory with complex gauge group $\mathrm{SL}(2,\mathbb{C})$. The Hilbert space carrying the representation
is $\mathcal{H}\simeq L^{2}(\mathbb{R})\otimes\mathbb{C}^{N}$, where $N$ is identified to the integer level of the Chern-Simons theory. A state in $\ch$ is represented by the function $f(\mu,m),$ $\mu\in\mathbb{R},$ $m\in\mathbb{Z}/N\mathbb{Z}$. The following basic operators are defined on $\mathcal{H}$ 
\begin{align*}
\bm{\mu}f(\mu,m) & =\mu f(\mu,m),\qquad\bm{\nu}f(\mu,m)=-\frac{N}{2\pi i}\frac{\partial}{\partial\mu}f(\mu,m)\\
e^{\frac{2\pi i}{N}\bm{m}}f(\mu,m) & =e^{\frac{2\pi i}{N}m}f(\mu,m),\qquad e^{\frac{2\pi i}{N}\bm{n}}f(\mu,m)=f(\mu,m+1).
\end{align*}
They satisfy
\be 
[\bm{\mu},\bm{\nu}]=\frac{N}{2\pi i},\qquad e^{\frac{2\pi i}{N}\bm{n}}e^{\frac{2\pi i}{N}\bm{m}}=e^{\frac{2\pi i}{N}}e^{\frac{2\pi i}{N}\bm{m}}e^{\frac{2\pi i}{N}\bm{n}},\qquad [\bm{\nu},e^{\frac{2\pi i}{N}\bm{m}}]=[\bm{\mu},e^{\frac{2\pi i}{N}\bm{n}}]=0
\ee 
The representation of $\bm{u},\bm{y},\tilde{\bm{u}},\tilde{\bm{y}}$ in the quantum torus algebra are represented by
\be
\bm{y} & =&\exp\left[\frac{2\pi i}{N}(-ib\boldsymbol{\mu}-\bm{m})\right],\qquad\tilde{\bm{y}}=\exp\left[\frac{2\pi i}{N}\left(-ib^{-1}\boldsymbol{\mu}+\bm{m}\right)\right],\\
\bm{u} & =&\exp\left[\frac{2\pi i}{N}(-ib\boldsymbol{\nu}-\bm{n})\right],\qquad\tilde{\bm{u}}=\exp\left[\frac{2\pi i}{N}\left(-ib^{-1}\boldsymbol{\nu}+\bm{n}\right)\right].
\ee
Their actions on states $f(\mu,m)$ are given by
\be
\bm{y}f(\mu,m)&=&e^{\frac{2\pi i}{N}(-ib {\mu}-{m})}f(\mu,m),\qquad \bm{u}f(\mu,m)=f(\mu+ib,m-1),\label{repuandy}\\
\tilde{\bm{y}}f(\mu,m)&=&e^{\frac{2\pi i}{N}(-ib^{-1} {\mu}+{m})}f(\mu,m),\qquad  \tilde{\bm{u}}f(\mu,m)=f(\mu+ib^{-1},m+1).\label{repuandy1}
\ee
These operators are unbounded
operators. The common domain $\mathfrak{D}$ of their Laurent polynomials
contains $f(\mu,m)$ being entire functions in $\mu$ and satisfying
\be
e^{\alpha_1\frac{2\pi}{N}b\mu}e^{\alpha_2\frac{2\pi}{N}b^{-1}\mu}f(\mu+ib\b_1+ib^{-1}\b_2,m)\in L^{2}(\mathbb{R}),
\qquad\forall\,m_0\in\mathbb{Z}/N\mathbb{Z},\quad\alpha_i,\b_i\in\mathbb{Z}.\label{domainD0}
\ee
The Hermite functions $e^{-\mu^{2}/2}H_{n}(\mu)$ , $n=1,\cdots,\infty$
satisfy all the requirements and span a dense domain in $L^{2}(\mathbb{R})$,
so $\mathfrak{D}$ is dense in $\mathcal{H}$. The operators $\bm{u},\tilde{\bm{u}},\bm{y},\tilde{\bm{y}}$ form the $({\bfq}^2,\tilde{{\bfq}}^2)$-Weyl algebra on the domain $\Fd$:
\be
\bm{u}\bm{y}={\bfq}^2\bm{y}\bm{u},\qquad\tilde{\bm{u}}\tilde{\bm{y}}=\tilde{{\bfq}}^2\tilde{\bm{y}}\tilde{\bm{u}},\qquad \bm{u}\tilde{\bm{y}}=\tilde{\bm{y}}\bm{u},\qquad \tilde{\bm{u}}{\bm{y}}={\bm{y}}\tilde{\bm{u}}.
\ee
The tilded and untilded operators are related by the Hermitian conjugate
\[
\bm{u}^{\dagger}=\tilde{\bm{u}},\qquad\bm{y}^{\dagger}=\tilde{\bm{y}},
\]
and they are normal operators. 

In the following discussion, we sometimes use the following
formal notation 
\[
\bm{u}=e^{\bm{U}},\qquad\bm{y}=e^{\bm{Y}};\qquad\tilde{\bm{u}}=e^{\tilde{\bm{U}}},\qquad\tilde{\bm{y}}=e^{\tilde{\bm{Y}}}
\]
where 
\begin{align*}
\bm{Y} & =\frac{2\pi i}{N}(-ib\boldsymbol{\mu}-\bm{m}),\qquad\tilde{\bm{Y}}=\frac{2\pi i}{N}\left(-ib^{-1}\boldsymbol{\mu}+\bm{m}\right),\\
\bm{U} & =\frac{2\pi i}{N}(-ib\boldsymbol{\nu}-\bm{n}),\qquad\tilde{\bm{U}}=\frac{2\pi i}{N}\left(-ib^{-1}\boldsymbol{\nu}+\bm{n}\right),
\end{align*}
satisfy the canonical commutation relation
\[
[\bm{U},\bm{Y}]=2\hbar,\qquad[\tilde{\bm{U}},\tilde{\bm{Y}}]=2\tilde{\hbar},\qquad [\bm{U},\tilde{\bm{Y}}]=[\tilde{\bm{U}},{\bm{Y}}]=0.
\]

We denote by $\cl(\mathfrak{D})$ the space of linear operators on $\mathfrak{D}$. The representation $\rho$: $\mathcal{A}_{\hbar}\to\mathcal{L}(\mathfrak{D})$
is given by
\begin{align}
\rho:\  & \bm{u}_{\alpha,\beta}\mapsto e^{\alpha\bm{U}+\beta\bm{Y}}=\bfq^{-\alpha\beta}\bm{u}^{\alpha}\bm{y}^{\beta},\label{eq:reptor1}\\
 & \tilde{\bm{u}}_{\alpha,\beta}\mapsto e^{\alpha\tilde{\bm{U}}+\beta\tilde{\bm{Y}}}=\tilde{\bfq}^{-\alpha\beta}\tilde{\bm{u}}^{\alpha}\tilde{\bm{y}}^{\beta}.\label{eq:reptor2}
\end{align}
The relations (\ref{eq:quantumtorus1}) and (\ref{eq:quantumtorus2})
are obtained by applying the $({\bfq}^2,\tilde{{\bfq}}^2)$-Weyl
algebra. In the following, we often denote $\rho(\bm{u}_{\alpha,\beta}$)
by $\bm{u}_{\alpha,\beta}$ for simplifying notations. The $\star$-stucture
is represented by the Hermitian conjugate on $\Fd$:
\[
\bm{u}_{\alpha,\beta}^{\dagger}=\tilde{\bfq}^{\alpha\beta}\tilde{\bm{y}}^{\beta}\tilde{\bm{u}}^{\alpha}=\tilde{\bfq}^{-\alpha\beta}\tilde{\bm{u}}^{\alpha}\tilde{\bm{y}}^{\beta}=\tilde{\bm{u}}_{\alpha,\beta}.
\]

\begin{lemma}\label{irreducible1}

The representation $\rho$ is irreducible in the sense that any bounded operator $\bm\co\in \cl(\ch)$ commuting with all elements in $\ca_h$ (i.e. $\bm\co\bm{a}\psi=\bm{a}\bm\co\psi$ \footnote{The commutativity may be written as $\bm\co\bm{a}\subset\bm{a}\bm\co$ (the graph of $\bm{a}\bm\co$ contains the graph of $\bm\co\bm{a}$), namely, $\bm{a}\bm\co$ is an extension of $\bm\co\bm{a}$.} for all $\psi\in\mathfrak{D}$ and $\bm{a}\in\ca_h$) is a scalar multiple of identity operator.
	
\end{lemma}
	
\begin{proof}
	
It is sufficient to show that any bounded operator on $\ch$ commuting with $\bm{u}=\bm{u}_{1,0},\ \bm{y}=\bm{u}_{0,1},\ \widetilde{\bm{u}}=\widetilde{\bm{u}}_{1,0},\ \widetilde{\bm{y}}=\widetilde{\bm{u}}_{0,1}$ is a scalar multiple of the identity operator. The spectral decomposition of the normal operators $\bm{y},\widetilde{\bm{y}}$ gives the spectral projections $E(\Delta,m'): L^2(\R)\otimes \C^N \to L^2(\Delta)$ defined by $E(\Delta,m')f(\mu,m)=\chi_\Delta(\mu)\delta_{m'}(m)f(\mu,m)$ for any closed interval $\Delta\subset \R$ and $m'\in\Z/N\Z$ (Fixing $m'$, the function $E(\Delta,m')f(\mu,m)$ vanishes for $m\neq m'$ due to the projection $\delta_{m'}(m)$, so it is understood as an element in $L^2(\Delta)$). $\chi_\Delta$ is the characteristic function of $\Delta$. That $\bm\co$ commutes with $\bm{y},\widetilde{\bm{y}}$ implies $[\bm\co,E(\Delta,m')]=0$, so $\bm\co$ leaves the image $L^2(\Delta)$ of the projection invariant. We define 
	\be 
	\co_{\Delta,m'}(\mu,m):=\lt(\bm{\co}\lt(\chi_\Delta\otimes\delta_{m'}\rt)\rt)(\mu,m)\in L^2(\Delta).
	\ee
$\co_{\Delta,m'}(\mu,m)=0$ if $m\neq m'$. For any $\Delta'\subset \Delta$, $\bm\co E(\Delta',m')E(\Delta,m')= E(\Delta',m') \bm\co E(\Delta,m')$ implies $\bm\co(\chi_{\Delta'}\otimes\delta_{m'})=\chi_{\Delta'}\bm\co(\chi_{\Delta}\otimes\delta_{m'})$ and thus $\co_{\Delta',m'}(\mu,m)=\co_{\Delta,m'}(\mu,m)$ for $\mu\in\Delta'$, i.e. the value of $\co_{\Delta,m'}(\mu,m)$ at any $\mu\in\Delta$ is independent of the choice of $\Delta$. By using any cover of $\R$ with closed intervals, we obtain the function $\co_{m'}(\mu,m)$ such that $\co_{\Delta,m'}(\mu,m)=\co_{m'}(\mu,m)\chi_\Delta(\mu)=\co_{m'}(\mu,m)\chi_\Delta(\mu)\delta_{m'}(m)$ for any $\Delta$. We define $\co(\mu,m):=\co_{m}(\mu,m)$ and have $\co_{\Delta,m'}(\mu,m)=\co(\mu,m)\chi_\Delta(\mu)\delta_{m'}(m)$. For any simple function $f_\fp=\sum_{\Delta}\sum_{m}f_{\Delta,m}(\chi_\Delta\otimes\delta_{m})$ based on a set $\fp$ of disjoint $\Delta$, we obtain
	\be
	\lt(\bm\co f_\fp\rt)(\mu,m)=\sum_{\Delta\in\fp}\sum_{m'\in\Z/N\Z}f_{\Delta,m'}\co_{\Delta,m'}(\mu,m)=\co(\mu,m)f_\fp(\mu,m).
	\ee 
	For any function $f(\mu,m)\in \ch$, there exists a sequence of simple functions $f_\fp$ that converges to $f$ in the sense of the Hilbert space norm. Since $\bm\co$ is a bounded operator, $\bm\co f_\fp$ converges to $\bm\co f$. Therefore,
	\be 
	\lt(\bm\co f\rt)(\mu,m)=\co(\mu,m)f (\mu,m),
	\ee
	for any $f\in\ch$. The assumption that $\bm\co\psi\in\Fd$ for any $\psi\in\Fd$ (in order that $\bm{a}\bm\co\psi$ is well-defined) implies that $\co(\mu,m)$ is an entire function in $\mu$ after analytic continuation. $|\co(\mu,m)|$ is a bounded function, since $\bm\co$ is a bounded operator.
	
	The commutativity between $\bm\co$ and $\bm{u},\widetilde{\bm u}$ implies that $\co(\mu,m)$ is a constant: We consider $\co(\mu,m)$ as a tempered distribution: $\co\in \mathscr{S}'(\R)\otimes \C^k$ and for any $f\in  \mathfrak{D}$, 
	\be  
	\co[f]=\sum_{m\in\Z/N\Z}\int\rmd \mu \,\co(\mu,m) f(\mu,m)=\sum_{n\in\Z/N\Z}\int\rmd \nu \,\cf[\co](\nu,n) \tilde{f}(\nu,n), 
	\ee  
where the Fourier transformations of the tempered distribution $\co$ is given by
	\be
	\cf[\co](\nu,n)=\frac{1}{N}\sum_{m\in\mathbb{Z}/N\mathbb{Z}}\int d\mu\,e^{-\frac{2\pi i}{N}\left(\mu\nu-mn\right)}\co(\mu,m),
	\ee
That $\bm{u},\widetilde{\bm u}$ commute with $\bm\co$ implies that  $e^{\frac{2\pi i}{N}\left(-ib\nu-n\right)}\cf[\co](\nu,n)=e^{\frac{2\pi i}{N}\left(-ib^{-1}\nu+n\right)}\cf[\co](\nu,n)=\cf[\co](\nu,n)$, which implies $\mathcal{F}\left[\co\right](\nu,n)=\text{constant}\times \delta(\nu)\delta_{\exp\left(\frac{2\pi i}{N}n\right),1}$, then $\co(\mu,m)$ is constant.

	
	
	\end{proof}

\section{Representation of a quantum deformed Lorentz algebra}\label{Representation of a q-deformed Lorentz algebra}

\subsection{The representation}

Based on the representation $(\Fd,\rho)$ of the quantum torus
algebra, we obtain a representation of $\suq$ on $\Fd$: We define a family of operators $E_{\lambda},F_{\lambda},K_{\lambda},K_{\lambda}^{-1}$ labelled by a parameter $\l\in\C^\times$: 
\be   
K_{\lambda}&=&\bm{u}_{-1,0},\qquad K_{\lambda}^{-1}=\bm{u}_{1,0},\qquad F_\l=-\frac{iq}{\bm{q}-\bm{q}^{-1}}\bm{u}_{0,1}\\
E_\l&=&-\frac{iq^{-1}}{\bm{q}-\bm{q}^{-1}}\left[\left(\lambda+\lambda^{-1}\right)\bm{u}_{0,-1}+\bm{u}_{1,-1}+\bm{u}_{-1,-1}\right],\label{KKEF2}
\ee
where $q=\bfq^{1/2}=e^{\hbar/2}$. It is straight-forward to check that the operator algebra of $E_{\lambda},F_{\lambda},K_{\lambda},K_{\lambda}^{-1}$ represents the commutation relation of $\suq$ \cite{kassel2012quantum}:
\be
KE&=&\bfq^{2}EK,\qquad KF=\bfq^{-2}FK,\qquad\left[E,F\right]=\frac{K-K^{-1}}{\bfq-\bfq^{-1}}.\label{EFK}
\ee

Similarly, we define the tilded operators with $\tilde{\l}\in\C^\times$
\be
\tilde{K}_{\lambda}&=&\tilde{\bm{u}}_{-1,0},\qquad\tilde{K}_{\lambda}^{-1}=\tilde{\bm{u}}_{1,0},\qquad \tilde{F}_{\lambda}=-\frac{i\tilde{q}^{-1}}{\tilde{\bm{q}}-\tilde{\bm{q}}^{-1}}\tilde{\bm{u}}_{0,1},\\
\tilde{E}_\l&=&-\frac{i\tilde{q}}{\tilde{\bm{q}}-\tilde{\bm{q}}^{-1}}\left[\left(\overline{\lambda}+\overline{\lambda}^{-1}\right)\tilde{\bm{u}}_{0,-1}+\tilde{\bm{u}}_{1,-1}+\tilde{\bm{u}}_{-1,-1}\right],
\ee
where $\tilde{q}=\tilde{\bfq}^{1/2}=e^{\tilde \hbar/2}$.
Their operator algebra represents the commutation relation of $\suqt$:
\be
\tilde{K}\tilde{E}=\tilde{{\bfq}}^{2}\tilde{E} \tilde{K} ,\qquad\tilde{K}\tilde{F}=\tilde{{\bfq}}^{-2}\tilde{F}\tilde{K},\qquad\left[\tilde{E},\tilde{F}\right]=\frac{\tilde{K} -\tilde{K}^{-1}}{\tilde{{\bfq}}-\tilde{{\bfq}}^{-1}}.\label{EFKt}
\ee
The tensor product $\suquqt$ has the $*$-structure:
\be 
E^*=\tilde{E},\qquad F^*=\tilde{F},\qquad K^{\pm1}{}^*=\tilde{K}^{\pm1},\label{EFKstar}
\ee
which is represented by the Hermitian conjugate on $\Fd$ when $\lambda^{*}=\tilde{\lambda}$:
\be  
 E_{\lambda}^{\dagger}=\tilde{E}_{\lambda},\qquad F_{\lambda}^{\dagger}= \tilde{F}_{\lambda},\qquad K_{\lambda}^{\pm1\dagger}=\tilde{K}_{\lambda}^{\pm1}.
\ee 
It shows that $\Fd$ carries a $*$-representation of
$\suquqt$ labelled by the complex parameters $(\lambda,\tilde{\lambda})$ with $\lambda^{*}=\tilde{\lambda}$. We denote this representation by $\pi_{\lambda,\tilde{\lambda}}$. It is manifest that $\pi_{\lambda,\tilde{\lambda}}=\pi_{\lambda^{-1},\tilde{\lambda}^{-1}}$.

\begin{lemma}
	The representation $(\Fd,\pi_{\lambda,\tilde{\lambda}})$ is irreducible.
\end{lemma}
	
\begin{proof} The generators of the $\bfq^2,\tilde{\bfq}^2$-Weyl algebra can be recovered by $\suquqt$ generators in the following way:
	\be
	&&\bm{u}=K_\l^{-1},\qquad \bm{y}=i(\bfq-\bfq^{-1})q^{-1}F_\l,\\
	&&\tilde{\bm{u}}=\tilde K_\l^{-1},\qquad \tilde{\bm{y}}=i(\tilde{\bfq}-\tilde{\bfq}^{-1})\tilde{q}\tilde{F}_\l.
	\ee
Then the irreducibility is implied by Lemma \ref{irreducible1}.
	
\end{proof}

$\suquqt$ is a $*$-Hopf algebra, with the co-product $\Delta$, antipode $S$, and co-unit $\eps$ given by
\be 
&&\Delta E=E\otimes K+1\otimes E,\qquad\Delta\tilde{E}=\tilde{E}\otimes\tilde{K}+1\otimes\tilde{E},\\
&&\Delta F=F\otimes1+K^{-1}\otimes F,\qquad\Delta\tilde{F}=\tilde{F}\otimes1+\tilde{K}^{-1}\otimes\tilde{F}\\
&&\Delta K^{\pm1}=K^{\pm1}\otimes K^{\pm1},\qquad\Delta\tilde{K}^{\pm1}=\tilde{K}^{\pm1}\otimes\tilde{K}^{\pm1}\\
&&S\left(K^{\pm1}\right)=K^{\mp1},\qquad S\left(E\right)=-EK^{-1},\qquad S\left(F\right)=-KF,\\
&&S\left(\tilde{K}^{\pm1}\right)=\tilde{K}^{\mp1},\qquad S\left(\tilde{E}\right)=-\tilde{E}\tilde{K}^{-1},\qquad S\left(\tilde{F}\right)=-\tilde{K}\tilde{F}\\
&&\varepsilon\left({K}^{\pm1}\right)=\varepsilon\left(\tilde{{K}}^{\pm1}\right)=1,\qquad \varepsilon(E)=\varepsilon(F)=\varepsilon(\tilde{E})=\varepsilon(\tilde{F})=0.
\ee

The Casimir operators defined by
\be 
Q_{\lambda} & =&-\left({\bfq}-{\bfq}^{-1}\right)^{2}E_{\lambda}F_{\lambda}-{\bfq}^{-1}K_{\lambda}-{\bfq}K_{\lambda}^{-1}\\
\tilde{Q}_{\lambda} & =&-\left(\tilde{{\bfq}}-{\tilde{\bfq}}^{-1}\right)^{2}\tilde{E}_{\lambda}\tilde{F}_{\lambda}-\tilde{{\bfq}}^{-1}K_{\lambda}-\tilde{{\bfq}}K_{\lambda}^{-1}
\ee
are constant on $\ch$: $Q_\l=\left(\lambda+\lambda^{-1}\right)\mathrm{id}_{\ch}$ and $\tilde{Q}_{\lambda}=\left(\tilde{\lambda}+\tilde{\lambda}^{-1}\right)\mathrm{id}_{\ch}$. Then $Q_{\lambda}^{\dagger}=\tilde{Q}_{\lambda}$ when $\lambda^{*}=\tilde{\lambda}$.

$\suquqt$ with the $*$-structure described above is a quantum deformation of (the universal enveloping algebra of) the Lorentz Lie algebra, which is recovered by $\bfq\to 1$ or $N\to\infty$ \footnote{The modular double of $U_q(sl(2,\R))$ cannot be called as quantum deformation of Lorentz algebra, because it is at $N=1$, and there is no limit that can send the quantum group to classical Lorentz algebra.}: The Lorentz Lie algebra as a real Lie algebra is generated by rotation generators $J^i$ and boost generators $K^i$ ($i=1,2,3$), with the commutation relations $[J^i,J^j]=-\eps^{ijk}J^k$, $[K^i,K^j]=\eps^{ijk}J^k$, $[J^i,K^j]=-\eps^{ijk}K^k$ and $*$-structure $J^i{}^*=-J^i$, $K^i{}^*=-K^i$ \footnote{$J^i,K^i$ relate to $H_\pm,H_3,F_\pm,F_3$ in \cite{gelfand2018representations} by $H_+=-i J^1+J^2$, $H_-=-i J^1-J^2$, $H_3=-iJ^3$, $F_+=i K^1-K^2$, $F_-=iK^1+K^2$, $F_3=iK^3$, with $H_3^*=H_3,\ H^*_\pm=H_\mp,\ F_3^*=F_ 3,\ F^*_\pm=F_\mp$.}. The self-dual and anti-self-dual generators are defined by $L^i=-\frac{1}{2}(J^i+iK^i)$ and $\tilde{L}^i=-\frac{1}{2}(J^i-iK^i)$, satisfying $[L^i,L^j]=\eps^{ijk}L^k$, $[\tilde L^i,\tilde L^j]=\eps^{ijk}\tilde L^k$, $[L^i,\tilde L^j]=0$, and $L^i{}^*=-\tilde L^i$. Therefore, the Lorentz Lie algebra is isomorphic to $sl_2\oplus sl_2$. We define $E=iL_1-L_2$, $F=iL_1+L_2$, $H=iL_3$, $\tilde E=i\tilde L_1+\tilde L_2$, $\tilde F=i\tilde L_1-\tilde L_2$, $\tilde H=-i\tilde L_3$, satisfying 
\be    
&&[E,F]=2H,\qquad [H,E]=E,\qquad [H,F]=-F,\label{EFH}\\
&&[\tilde E,\tilde F]=2\tilde H,\qquad [\tilde H,\tilde E]=\tilde E,\qquad [\tilde H,\tilde F]=-\tilde F,\label{EFHt}
\ee
and 
\be 
E^*=\tilde{E},\qquad F^*=\tilde{F},\qquad H^*=-\tilde{H}.\label{EFHstar}
\ee
The relations \eqref{EFK} and \eqref{EFKt} of $\suquqt$ are deformations of \eqref{EFH} and \eqref{EFHt} by identifying $K=\bfq^{2H}$ and $\tilde K=\tilde{\bfq}^{2\tilde H}$. The $*$-structure of $\suquqt$ in \eqref{EFKstar} follows from \eqref{EFHstar} due to $\bfq^*=\tilde{\bfq}^{-1}$.

\subsection{The Clebsch-Gordan decomposition} \label{The Clebsch-Gordan decomposition}

The tensor product representation of $\suquqt$
on $\mathcal{H}\otimes\mathcal{H}$ is given by $(\pi_{\lambda_{2},\tilde{\lambda}_{2}}\otimes\pi_{\lambda_{1},\tilde{\lambda}_{1}})\circ\Delta$. In the following, we often use the notation e.g. $K_{1}\equiv K_{\lambda_{1}}\equiv\pi_{\lambda_{1},\tilde{\lambda}_{1}}(K)$
and similar for other generators. We only focus on the represesentation
of untilded operator, while the tilded operator can be analyzed in
the same way. The discussion in this section is a straight-forward generalization from the Clebsch-Gordan decomposition of the modular-doubled $U_q(sl(2,\R))$ in \cite{Nidaiev:2013bda} by incorporating the discrete degrees of freedom in $\ch\simeq L^2(\R)\otimes\C^N$.

The representation of Casimir $Q_{21}=(\pi_{\lambda_{2},\tilde{\lambda}_{2}}\otimes\pi_{\lambda_{1},\tilde{\lambda}_{1}})\circ\Delta Q$
is expressed as
\be
Q_{21}
&=&K_{2}^{-1}Q_{1}+Q_{2}K_{1}+\left(\bfq+\bfq^{-1}\right)K_{2}^{-1}K_{1}+\left(Q_{2}+{\bfq}^{-1}K_{2}+{\bfq}K_{2}^{-1}\right)F_{2}^{-1}K_{2}^{-1}K_{1}F_{1}\nonumber\\
&&+\,F_{2}\left(Q_{1}+{\bfq}^{-1}K_{1}+{\bfq}K_{1}^{-1}\right)F_{1}^{-1}
\ee
where $Q_{1}$ and $Q_{2}$ are proprtional to identity operator on
$\mathcal{H}$. Our task is to find the unitary transformation
to diagonalize $Q_{21}$. 

It turns out that we can use a few elementary unitary transformations to simplify $Q_{21}$:
\begin{itemize}
\item Firstly, given $f(\mu_{2},m_{2}\mid\mu_{1},m_{1})\in\mathcal{H}\otimes\mathcal{H}$,
the unitary transformation $\mathcal{S}_{2}$ shift $\mu_{1},m_{1}$ by $-\mu_2,-m_2$
\begin{align*}
\cs_{2}f(\mu_{2},m_{2}\mid\mu_{1},m_{1})= & e^{\frac{2\pi i}{N}\left(\mu_{2}\bm{\nu}_{1}-m_{2}\bm{n}_{1}\right)}f(\mu_{2},m_{2}\mid\mu_{1},m_{1})\\
= & f(\mu_{2},m_{2}\mid\mu_{1}-\mu_{2},m_{1}-m_{2}).
\end{align*}
\item The unitary transformation $t_{21}$ is defined by the quantum dilogarithm:
\[
t_{21}=\varphi\left(e^{\bm{Y}_{1}-\bm{Y}_{2}+\bm{U}_{2}},e^{\tilde{\bm{Y}}_{1}-\tilde{\bm{Y}}_{2}+\tilde{\bm{U}}_{2}}\right),
\]
where $\varphi(y,\tilde{y})$ is the quantum dilogarithm function (see Appendix \ref{Quantum dilogarithm} for details)
\[
\varphi(y,\tilde{y})=\left[\frac{\prod_{j=0}^{\infty}\left(1+{\bfq}^{2j+1}y\right)}{\prod_{j=0}^{\infty}\left(1+\tilde{{\bfq}}^{-2j-1}\tilde{y}\right)}\right]^{-1}.
\]
To understand the action of $t_{21}$ on $\mathcal{H}\otimes\mathcal{H}$,
we consider the Weil transformation $\mathcal{V}_{2}$ representing
the following symplectic transformation \cite{levelk}:
\be
&&\qquad \left(\begin{array}{c}
Y_{2}\\
U_{2}
\end{array}\right)  \mapsto\left(\begin{array}{c}
Y_{2}-U_{2}\\
X_{2}
\end{array}\right)=\begin{pmatrix}1 & -1\\
0 & 1
\end{pmatrix}\left(\begin{array}{c}
Y_{2}\\
U_{2}
\end{array}\right)=\begin{pmatrix}0 & 1\\
-1 & 0
\end{pmatrix}\begin{pmatrix}1 & 0\\
1 & 1
\end{pmatrix}\begin{pmatrix}0 & -1\\
1 & 0
\end{pmatrix}\left(\begin{array}{c}
Y_{2}\\
X_{2}
\end{array}\right),\nonumber\\
&&\mathcal{V}_{2}  f(\mu_{2},m_{2}\mid\mu_{1},m_{1})=\frac{1}{N^{2}}\sum_{n,m\in\mathbb{Z}/N\mathbb{Z}}\int d\nu d\mu\,e^{\frac{2\pi i}{N}\left(\mu_{2}\nu-m_{2}n\right)}(-1)^{n^2}e^{\frac{\pi i}{N}\left(\nu^{2}-n^{2}\right)}e^{-\frac{2\pi i}{N}\left(\mu\nu-mn\right)}f(\mu,m\mid\mu_{1},m_{1}),\nonumber
\ee
for all $f\in \Fd\otimes \Fd$. $\mathcal{V}_{2}$ diagonalizes $e^{-\bm{Y}_{2}+\bm{U}_{2}}$ and$e^{-\tilde{\bm{Y}}_{2}+\tilde{\bm{U}}_{2}}$
by 
\be
\mathcal{V}_{2}^{-1}e^{-\bm{Y}_{2}+\bm{U}_{2}}\mathcal{V}_{2}f=-e^{-\bm{Y}_{2}}f,\qquad \mathcal{V}_{2}^{-1}e^{-\tilde{\bm{Y}}_{2}+\tilde{\bm{U}}_{2}}\mathcal{V}_{2}f=-e^{-\tilde{\bm{Y}}_{2}}f,
\ee
for any $f\in\mathfrak{D}\otimes\mathcal{H}$. Therefore
\[
t_{21}=\mathcal{V}_{2}\varphi\left(-e^{\bm{Y}_{1}-\bm{Y}_{2}},-e^{\tilde{\bm{Y}}_{1}-\tilde{\bm{Y}}_{2}}\right)\mathcal{V}_{2}^{-1},
\]
where $\varphi(-e^{\bm{Y}_{1}-\bm{Y}_{2}},-e^{\tilde{\bm{Y}}_{1}-\tilde{\bm{Y}}_{2}})$
simply multiplies the quantum dilogarithm function $\varphi(-y_{1}y_{2}^{-1},-\tilde{y}_{1}\tilde{y}_{2}^{-1})$
to $f(\mu_{2},m_{2}\mid\mu_{1},m_{1})$, with $y_a=\exp[\frac{2\pi i}{N}(-ib\mu_a-m_a)]$ and $\tilde{y}_a=\exp[\frac{2\pi i}{N}(-ib^{-1}\mu_a+m_a)]$, $a=1,2$.
\item The unitary transformation $C_{1}$ is defined by
\be
C_{1}^{-1}=\varphi\left(e^{\bm{Y}_{1}-L_{2}},e^{\tilde{\bm{Y}}_{1}-\tilde{L}_{2}}\right)\frac{\varphi\left(e^{-\bm{U}_{1}+L_{1}},e^{-\tilde{\bm{U}}_{1}+\tilde{L}_{1}}\right)}{\varphi\left(e^{\bm{U}_{1}+L_{1}},e^{\tilde{\bm{U}}_{1}+\tilde{L}_{1}}\right)}
\ee
where $e^{L_{a}}=\lambda_{a}$, $a=1,2$, is a notation similar to e.g. $\bmy=e^{\bm Y}$ and $\bm{u}=e^{\bm U}$. 
\end{itemize}
\begin{lemma}\label{lemma31}
The unitary transformation $t_{21}\cs_{2}$ transforms $Q_{21}$
to act only on the second factor of $\mathcal{H}\otimes\mathcal{H}$.
Namely $Q_{21}=t_{21}\cs_{2}Q_{1}^{\prime}\cs_{2}^{-1}t_{21}^{-1}$
where
\begin{equation}
Q_{1}^{\prime}=e^{\bm{Y}_{1}-\bm{U}_{1}}+e^{-\bm{Y}_{1}+\bm{U}_{1}}+e^{-\bm{Y}_{1}}Q_{1}+e^{-\bm{U}_{1}}Q_{2}+e^{-\bm{U}_{1}-\bm{Y}_{1}}.\label{eq:Q1primes}
\end{equation}
\end{lemma}
\begin{proof}
The proofs of this Lemma, Lemma \ref{lemma32}, and Lemma \ref{CGirrep} only present the formal algebraic computations. More detailed discussions involving analyzing operator domains are given in Appendix \ref{Operator domains}.

Since $\cs_{2}e^{\bm{Y}_{1}}\cs_{2}^{-1}=e^{\bm{Y}_{1}-\bm{Y}_{2}}$ and $\cs_2 e^{\bm{U}_1}\cs_2^{-1}=e^{\bm{U}_1}$,
we have 
\[
\cs_{2}Q_{1}^{\prime}\cs_{2}^{-1}=e^{\bm{Y}_{1}-\bm{Y}_{2}-\bm{U}_{1}}+e^{-\bm{Y}_{1}+\bm{Y}_{2}+\bm{U}_{1}}+e^{-\bm{Y}_{1}+\bm{Y}_{2}}Q_{1}+e^{-\bm{U}_{1}}Q_{2}+e^{-\bm{U}_{1}-\bm{Y}_{1}+\bm{Y}_{2}}.
\]
Since $[\bm{Y}_{1}-\bm{Y}_{2}-\bm{U}_{1},\,\bm{Y}_{1}-\bm{Y}_{2}+\bm{U}_{2}]=[-\bm{Y}_{2},\bm{U}_{2}]+[-\bm{U}_{1},\,\bm{Y}_{1}]=0$, the first two terms commute with $t_{21}=\varphi(e^{\bm{Y}_{1}-\bm{Y}_{2}+\bm{U}_{2}},e^{\tilde{\bm{Y}}_{1}-\tilde{\bm{Y}}_{2}+\tilde{\bm{U}}_{2}})$.
In the following, we often suppress the tilded entry of $\varphi$
when it is not involved in the manipulation. We check the following relation by using the recursion relation of the quantum dilogarithm $\varphi$ in \eqref{recursion3} of Appendix \ref{Quantum dilogarithm}:
\begin{align*}
t_{21}e^{-\bm{Y}_{1}+\bm{Y}_{2}}t_{21}^{-1} & =\varphi\left(e^{\bm{Y}_{1}-\bm{Y}_{2}+\bm{U}_{2}}\right)e^{-\bm{Y}_{1}+\bm{Y}_{2}}\varphi\left(e^{\bm{Y}_{1}-\bm{Y}_{2}+\bm{U}_{2}}\right)^{-1}\\
 & =\varphi\left(e^{\bm{Y}_{1}-\bm{Y}_{2}+\bm{U}_{2}}\right)\varphi\left({\bfq}^{-2}e^{\bm{Y}_{1}-\bm{Y}_{2}+\bm{U}_{2}}\right)^{-1}e^{-\bm{Y}_{1}+\bm{Y}_{2}}\\
 & =\left(1+{\bfq}^{-1}e^{\bm{Y}_{1}-\bm{Y}_{2}+\bm{U}_{2}}\right)e^{-\bm{Y}_{1}+\bm{Y}_{2}}\\
 & =e^{-\bm{Y}_{1}+\bm{Y}_{2}}+e^{\bm{U}_{2}},
\end{align*}
\begin{align*}
t_{21}e^{-\bm{U}_{1}}t_{21}^{-1} & =\varphi\left(e^{\bm{Y}_{1}-\bm{Y}_{2}+\bm{U}_{2}}\right)e^{-\bm{U}_{1}}\varphi\left(e^{\bm{Y}_{1}-\bm{Y}_{2}+\bm{U}_{2}}\right)^{-1}\\
 & =\varphi\left(e^{\bm{Y}_{1}-\bm{Y}_{2}+\bm{U}_{2}}\right)\varphi\left({\bfq}^{-2}e^{\bm{Y}_{1}-\bm{Y}_{2}+\bm{U}_{2}}\right)^{-1}e^{-\bm{U}_{1}}\\
 & =\left(1+{\bfq}^{-1}e^{\bm{Y}_{1}-\bm{Y}_{2}+\bm{U}_{2}}\right)e^{-\bm{U}_{1}}\\
 & =e^{-\bm{U}_{1}}+e^{\bm{Y}_{1}-\bm{U}_{1}-\bm{Y}_{2}+\bm{U}_{2}},
\end{align*}
\begin{align*}
t_{21}e^{-\bm{U}_{1}-\bm{Y}_{1}+\bm{Y}_{2}}t_{21}^{-1} & =\varphi\left(e^{\bm{Y}_{1}-\bm{Y}_{2}+\bm{U}_{2}}\right)e^{-\bm{U}_{1}-\bm{Y}_{1}+\bm{Y}_{2}}\varphi\left(e^{\bm{Y}_{1}-\bm{Y}_{2}+\bm{U}_{2}}\right)^{-1}\\
 & =\varphi\left(e^{\bm{Y}_{1}-\bm{Y}_{2}+\bm{U}_{2}}\right)\varphi\left({\bfq}^{-4}e^{\bm{Y}_{1}-\bm{Y}_{2}+\bm{U}_{2}}\right)^{-1}e^{-\bm{U}_{1}-\bm{Y}_{1}+\bm{Y}_{2}}\\
 & =\left(1+{\bfq}^{-3}e^{\bm{Y}_{1}-\bm{Y}_{2}+\bm{U}_{2}}\right)\left(1+{\bfq}^{-1}e^{\bm{Y}_{1}-\bm{Y}_{2}+\bm{U}_{2}}\right)e^{-\bm{U}_{1}-\bm{Y}_{1}+\bm{Y}_{2}}\\
 & =\left[1+\left({\bfq}^{-1}+{\bfq}^{-3}\right)e^{\bm{Y}_{1}-\bm{Y}_{2}+\bm{U}_{2}}+{\bfq}^{-4}e^{2\bm{Y}_{1}-2\bm{Y}_{2}+2\bm{U}_{2}}\right]e^{-\bm{U}_{1}-\bm{Y}_{1}+\bm{Y}_{2}}\\
 & =e^{-\bm{U}_{1}-\bm{Y}_{1}+\bm{Y}_{2}}+\left({\bfq}+{\bfq}^{-1}\right)e^{-\bm{U}_{1}+\bm{U}_{2}}+e^{\bm{Y}_{1}-\bm{Y}_{2}-\bm{U}_{1}+2\bm{U}_{2}}.
\end{align*}
We obtain that 
\be 
t_{21}\cs_{2}Q_{1}^{\prime}\cs_{2}^{-1}t_{21}^{-1}&= & \bm{u}_{0,-1}^{(2)}\bm{u}_{-1,1}^{(1)}+\bm{u}_{0,1}^{(2)}\bm{u}_{1,-1}^{(1)}+Q_{1}\left(\bm{u}_{0,1}^{(2)}\bm{u}_{0,-1}^{(1)}+\bm{u}_{1,0}^{(2)}\right)+Q_{2}\left(\bm{u}_{-1,0}^{(1)}+\bm{u}_{1,-1}^{(2)}\bm{u}_{-1,1}^{(1)}\right)\nonumber\\
 &&+\,\bm{u}_{0,1}^{(2)}\bm{u}_{-1,-1}^{(1)}+\left({\bfq}+{\bfq}^{-1}\right)\bm{u}_{1,0}^{(2)}\bm{u}_{-1,0}^{(1)}+\bm{u}_{2,-1}^{(2)}\bm{u}_{-1,1}^{(1)}.
\ee 
It is the same as the expression of $Q_{21}$ in terms of $\bm{u}_{\alpha,\beta}$. 
\end{proof}
\begin{lemma}\label{lemma32}
The unitary transformation $C_{1}$ further simplifies $Q_{1}^{\prime}$:
\[
	{Q}_{1}^{\prime\prime}=C_{1}Q_{1}^{\prime}C_{1}^{-1}=e^{-\bm{U}_{1}+L_{2}}+e^{\bm{U}_{1}-L_{2}}+e^{-\bm{Y}_{1}-L_{1}}.
\]
\end{lemma}
\begin{proof}
For the 1st, 4th, and 5th terms of $Q_{1}^{\prime}$ in (\ref{eq:Q1primes}), we use the recursion relation of $\varphi$ in \eqref{recursion3} of Appendix \ref{Quantum dilogarithm},
\begin{align*}
 & C_{1}\left(e^{\bm{Y}_{1}-\bm{U}_{1}}+e^{-\bm{U}_{1}}Q_{2}+e^{-\bm{Y}_{1}-\bm{U}_{1}}\right)C_{1}^{-1}\\
= & \frac{\varphi\left(e^{\bm{U}_{1}+L_{1}}\right)}{\varphi\left(e^{-\bm{U}_{1}+L_{1}}\right)}\varphi\left(e^{\bm{Y}_{1}-L_{2}}\right)^{-1}{\bfq}e^{-\bm{Y}_{1}}\left(1+{\bfq}^{-1}e^{\bm{Y}_{1}-L_{2}}\right)\left(1+{\bfq}^{-1}e^{\bm{Y}_{1}+L_{2}}\right)\varphi\left({\bfq}^{-2}e^{\bm{Y}_{1}-L_{2}}\right)\frac{\varphi\left(e^{-\bm{U}_{1}+L_{1}}\right)}{\varphi\left(e^{\bm{U}_{1}+L_{1}}\right)}e^{-\bm{U}_{1}}\\
= & \frac{\varphi\left(e^{\bm{U}_{1}+L_{1}}\right)}{\varphi\left(e^{-\bm{U}_{1}+L_{1}}\right)}\varphi\left(e^{\bm{Y}_{1}-L_{2}}\right)^{-1}{\bfq}e^{-\bm{Y}_{1}}\left(1+{\bfq}^{-1}e^{\bm{Y}_{1}+L_{2}}\right)\varphi\left(e^{\bm{Y}_{1}-L_{2}}\right)\frac{\varphi\left(e^{-\bm{U}_{1}+L_{1}}\right)}{\varphi\left(e^{\bm{U}_{1}+L_{1}}\right)}e^{-\bm{U}_{1}}\\
= & \frac{\varphi\left(e^{\bm{U}_{1}+L_{1}}\right)}{\varphi\left(e^{-\bm{U}_{1}+L_{1}}\right)}\left({\bfq}^{-1}e^{-\bm{U}_{1}}e^{-\bm{Y}_{1}}\right)\frac{\varphi\left(e^{-\bm{U}_{1}+L_{1}}\right)}{\varphi\left(e^{\bm{U}_{1}+L_{1}}\right)}+e^{-\bm{U}_{1}+L_{2}}.
\end{align*}
For the 2nd term of $Q_{1}^{\prime}$, we use the recursion relation \eqref{recursion1}
\begin{align*}
 & C_{1}e^{-\bm{Y}_{1}+\bm{U}_{1}}C_{1}^{-1}\\
= & \frac{\varphi\left(e^{\bm{U}_{1}+L_{1}}\right)}{\varphi\left(e^{-\bm{U}_{1}+L_{1}}\right)}\varphi\left(e^{\bm{Y}_{1}-L_{2}}\right)^{-1}{\bfq}^{-1}e^{-\bm{Y}_{1}}\varphi\left({\bfq}^{2}e^{\bm{Y}_{1}-L_{2}}\right)\frac{\varphi\left(e^{-\bm{U}_{1}+L_{1}}\right)}{\varphi\left(e^{\bm{U}_{1}+L_{1}}\right)}e^{\bm{U}_{1}}\\
= & \frac{\varphi\left(e^{\bm{U}_{1}+L_{1}}\right)}{\varphi\left(e^{-\bm{U}_{1}+L_{1}}\right)}{\bfq}^{-1}e^{-\bm{Y}_{1}}\left(1+{\bfq}e^{\bm{Y}_{1}-L_{2}}\right)\frac{\varphi\left(e^{-\bm{U}_{1}+L_{1}}\right)}{\varphi\left(e^{\bm{U}_{1}+L_{1}}\right)}e^{\bm{U}_{1}}\\
= & \frac{\varphi\left(e^{\bm{U}_{1}+L_{1}}\right)}{\varphi\left(e^{-\bm{U}_{1}+L_{1}}\right)}\left({\bfq}e^{\bm{U}_{1}}e^{-\bm{Y}_{1}}\right)\frac{\varphi\left(e^{-\bm{U}_{1}+L_{1}}\right)}{\varphi\left(e^{\bm{U}_{1}+L_{1}}\right)}+e^{\bm{U}_{1}-L_{2}}
\end{align*}
For the 3rd term, 
\begin{align*}
C_{1}e^{-\bm{Y}_{1}}Q_{1}C_{1}^{-1}= & \frac{\varphi\left(e^{\bm{U}_{1}+L_{1}}\right)}{\varphi\left(e^{-\bm{U}_{1}+L_{1}}\right)}e^{-\bm{Y}_{1}}\left(e^{L_{1}}+e^{-L_{1}}\right)\frac{\varphi\left(e^{-\bm{U}_{1}+L_{1}}\right)}{\varphi\left(e^{\bm{U}_{1}+L_{1}}\right)}.
\end{align*}
Inserting these results, we obtain
\begin{align*}
 & C_{1}Q_{1}^{\prime}C_{1}^{-1}\\
= & e^{-\bm{U}_{1}+L_{2}}+e^{\bm{U}_{1}-L_{2}}+\frac{\varphi\left(e^{\bm{U}_{1}+L_{1}}\right)}{\varphi\left(e^{-\bm{U}_{1}+L_{1}}\right)}\left(1+{\bfq}e^{\bm{U}_{1}+L_{1}}\right)\left(1+{\bfq}^{-1}e^{-\bm{U}_{1}+L_{1}}\right)\frac{\varphi\left({\bfq}^{-2}e^{-\bm{U}_{1}+L_{1}}\right)}{\varphi\left({\bfq}^{2}e^{\bm{U}_{1}+L_{1}}\right)}e^{-\bm{Y}_{1}-L_{1}}\\
= & e^{-\bm{U}_{1}+L_{2}}+e^{\bm{U}_{1}-L_{2}}+e^{-\bm{Y}_{1}-L_{1}}.
\end{align*}
where the recursion relations \eqref{recursion1} and \eqref{recursion3} are used in the last step.
\end{proof}

If we define the unitary transformation $\mathcal{U}_{21}^{-1}=t_{21}\cs_{2}C_{1}^{-1}:\,\mathcal{H}\otimes\mathcal{H}\to\mathcal{H}\otimes\mathcal{H}$, the Casimir operator is simplified by 
\[
Q_{1}''=\mathcal{U}_{21}Q_{21}\mathcal{U}_{21}^{-1}=e^{-\bm{U}_{1}+L_{2}}+e^{\bm{U}_{1}-L_{2}}+e^{-\bm{Y}_{1}-L_{1}}.
\]
The Hermitian conjugate gives for the tilded operator
\[
\tilde{Q}_{1}^{\prime\prime}=\mathcal{U}_{21}\tilde{Q}_{21}\mathcal{U}_{21}^{-1}=e^{-\tilde{\bm{U}}_{1}+\tilde{L}_{2}}+e^{\tilde{\bm{U}}_{1}-\tilde{L}_{2}}+e^{-\tilde{\bm{Y}}_{1}-\tilde{L}_1}.
\]
$Q_{1}^{\prime\prime}$ and $\tilde{Q}_{1}^{\prime\prime}$
only acting on the second copy of $\mathcal{H}\otimes\mathcal{H}$. It turns out that $Q_1'',\tilde{Q}_{1}^{\prime\prime}$ can be simultaneously diagonalized (see Section \ref{Eigenstates of the trace operators}),
and their spectral decomposition gives a direct integral decomposition
\be    
\mathcal{H}\simeq\int_{\mathbb{C}}^{\oplus}d\mu\left(\l,\l^{*}\right)\,W\left(\l,\l^{*}\right),\qquad Q_2''=\int_\C\lt(\l+\l^{-1}\rt)\rmd P_{\l,\tilde{\l}},\qquad \tilde{Q}_2''=\int_\C\lt(\tilde{\l}+\tilde{\l}^{-1}\rt)\rmd P_{\l,\tilde{\l}},\nonumber
\ee  
where $\mu$ is the spectral measure, $P_{\l,\tilde{\l}}$ is the projective-valued measure for both $Q_1'',\tilde{Q}_{1}^{\prime\prime}$, and $\tilde{\l}=\l^*$. Therefore,
\[
\mathcal{U}_{21}:\ \mathcal{H}\otimes\mathcal{H}\to\int_{\mathbb{C}}^{\oplus}d\mu\left(\l,\l^{*}\right)\,\bm{\mathcal{H}}\left(\l,\l^{*}\right),\qquad\bm{\mathcal{H}}\left(\l,\l^{*}\right)=\mathcal{H}\otimes W\left(\l,\l^{*}\right).
\]
The unitary transformation $\mathcal{U}_{21}$ gives a Clebsch-Gordan
decomposition of the tensor product representation by the following
result:
\begin{lemma}\label{CGirrep}
If $\dim W\left(\l,\l^{*}\right)=1$, each $\bm{\mathcal{H}}\left(\l,\l^{*}\right)$
is a irreducible representation of $(\pi_{\lambda_{2},\tilde{\lambda}_{2}}\otimes\pi_{\lambda_{1},\tilde{\lambda}_{1}})\circ\Delta$. 
\end{lemma}
\begin{proof}
$F_{2},E_{2},K_{2},K_{2}^{-1}$ and their tilded relatives are represented
irreducibly on the first factor in $\bm{\mathcal{H}}\left(\l,\l^{*}\right)=\mathcal{H}\otimes W\left(\l,\l^{*}\right)$,
and they are explicitly given by
\begin{align*}
K_{2} & =e^{-\bm{U}_{2}},\qquad K_{2}^{-1}=e^{\bm{U}_{2}},\qquad F_{2}=-\frac{i}{{\bfq}-{\bfq}^{-1}}qe^{\bm{Y}_{2}},\\
E_{2} & =-\frac{1}{\left({\bfq}-{\bfq}^{-1}\right)^{2}}\left({\bfq}^{-1}K_{2}+{\bfq}K_{2}^{-1}+Q_{1}^{\prime\prime}\right)F_{2}^{-1},\qquad Q_{1}^{\prime\prime}=(\l+\l^{-1})\mathrm{id}_{\mathcal{H}}
\end{align*}
and similarly for the tilded operators. The representaion is unitarily
equivalent to the tensor production representation:
\begin{align*}
\mathcal{U}_{21}^{-1}K_{2}\mathcal{U}_{21} & =t_{21}\cs_{2}e^{-\bm{U}_{2}}\cs_{2}^{-1}t_{21}^{-1}=t_{21}e^{-\bm{U}_{2}-\bm{U}_{1}}t_{21}^{-1}\\
 & =\varphi\left(e^{\bm{Y}_{1}-\bm{Y}_{2}+\bm{U}_{2}}\right)e^{-\bm{U}_{2}-\bm{U}_{1}}\varphi\left(e^{\bm{Y}_{1}-\bm{Y}_{2}+\bm{U}_{2}}\right)^{-1}\\
 & =e^{-\bm{U}_{2}-\bm{U}_{1}}=K_{2}K_{1}=\lt(\Delta K\rt)_{21}
\end{align*}
\begin{align*}
\mathcal{U}_{21}^{-1}F_{2}\mathcal{U}_{21} & =-\frac{iq}{{\bfq}-{\bfq}^{-1}}t_{21}\cs_{2}e^{\bm{Y}_{2}}\cs_{2}^{-1}t_{21}^{-1}\\
 & =-\frac{iq}{{\bfq}-{\bfq}^{-1}}\varphi\left(e^{\bm{Y}_{1}-\bm{Y}_{2}+\bm{U}_{2}}\right)e^{\bm{Y}_{2}}\varphi\left(e^{\bm{Y}_{1}-\bm{Y}_{2}+\bm{U}_{2}}\right)^{-1}\\
 & =-\frac{iq}{{\bfq}-{\bfq}^{-1}}\varphi\left(e^{\bm{Y}_{1}-\bm{Y}_{2}+\bm{U}_{2}}\right)\varphi\left({\bfq}^{-2}e^{\bm{Y}_{1}-\bm{Y}_{2}+\bm{U}_{2}}\right)^{-1}e^{\bm{Y}_{2}}\\
 & =-\frac{iq}{{\bfq}-{\bfq}^{-1}}\left(1+{\bfq}^{-1}e^{\bm{Y}_{1}-\bm{Y}_{2}+\bm{U}_{2}}\right)e^{\bm{Y}_{2}}\\
 & =-\frac{iq}{{\bfq}-{\bfq}^{-1}}\left(e^{\bm{Y}_{2}}+e^{\bm{Y}_{1}+\bm{U}_{2}}\right)=F_{2}+K_{2}^{-1}F_{1}\\
 & =\lt(\Delta F\rt)_{21}
\end{align*}
Moreover we have $\mathcal{U}_{21}^{-1}F_{2}^{-1}\mathcal{U}_{12}=(\Delta F)_{21}^{-1}$. Therefore, at the level of the representation, $Q_{21}=-\left({\bfq}-{\bfq}^{-1}\right)^{2}(\Delta E)_{21}(\Delta F)_{21}-{\bfq}^{-1}(\Delta K)_{21}-{\bfq}(\Delta K)_{21}^{-1}$ implies
\be 
(\Delta E)_{21}&=&-\frac{\lt[{\bfq}^{-1}(\Delta K)_{21}+{\bfq}(\Delta K)_{21}^{-1}+Q_{21}\rt](\Delta F)_{21}^{-1}}{\left({\bfq}-{\bfq}^{-1}\right)^{2}}\nonumber\\
& =&\mathcal{U}_{21}^{-1}\frac{\left({\bfq}^{-1}K_{2}+{\bfq}K_{2}^{-1}+Q_{1}^{\prime\prime}\right)F_{2}^{-1}}{\left({\bfq}-{\bfq}^{-1}\right)^{2}}\mathcal{U}_{21}=\mathcal{U}_{21}^{-1}E_{2}\mathcal{U}_{21},
\ee   
by using $\mathcal{U}_{21}^{-1}Q_{1}^{\prime\prime}\mathcal{U}_{21}=Q_{21}$.
\end{proof}

When $\dim\mathcal{H}\left(\ell,\ell^{*}\right)=1$ holds, the unitary transformation $\cu_{21}$ is a Clebsch-Gordan map intertwining three representations $\l_1,\l_2,\l_3$:
\be
\cu_{21}\lt[\left(\pi_{\lambda_{2},\tilde{\lambda}_{2}}\otimes\pi_{\lambda_{1},\tilde{\lambda}_{1}}\right)\circ\Delta( X )\rt]\cu_{21}^{-1} =\int_{\mathbb{C}}^{\oplus}\pi_{\lambda_{3},\tilde{\lambda}_{3}}(X)\otimes \rmd P_{\lambda_{3},\tilde{\lambda}_{3}}.\label{CGequation}
\ee

In the following, we assume $\lambda_a=e^{L_a}$, $a=1,2$, is quantized
\be
\lambda_a=\exp\lt[\frac{2\pi i}{N}(-i b\mu_{a}-m_{a})\rt], \qquad \mu_{a}\in\R,\ m_{a}\in\Z/N\Z.
\ee
We define the following unitary transformations
\be 
\mathcal{S}_{\lambda_{1}}f\left(\mu,m\right)=f\left(\mu-\mu_{1},m-m_{1}\right),\qquad {\mathcal{D}}_{\lambda_{2}}f\left(\mu,m\right)=e^{-\frac{2\pi i}{N}(\mu_2\mu-m_2m)}f\left(\mu,m\right),
\ee
satisfying $\mathcal{S}_{\lambda_{1}}\mathcal{D}_{\lambda_{2}}=e^{\frac{2\pi i}{N}(\mu_{1}\mu_{2}-m_{1}m_{2})}\mathcal{D}_{\lambda_{2}}\mathcal{S}_{\lambda_{1}}$ and 
\be 
&&\mathcal{S}_{\lambda_{1}}\bm{y}^{-1}\mathcal{S}_{\lambda_{1}}^{-1}=\lambda_{1}\bm{y}^{-1},\quad \mathcal{S}_{\lambda_{1}}\bm{u}\mathcal{S}_{\lambda_{1}}^{-1}=\bm{u},\quad \mathcal{S}_{\lambda_{1}}\tilde{\bm{y}}^{-1}\mathcal{S}_{\lambda_{1}}^{-1}={\lambda}^*_{1}\tilde{\bm{y}}^{-1},\quad \mathcal{S}_{\lambda_{1}}\tilde{\bm{u}}\mathcal{S}_{\lambda_{1}}^{-1}=\tilde{\bm{u}},\\
&&\mathcal{D}_{\lambda_{2}}^{-1}\bm{y}\mathcal{D}_{\lambda_{2}}=\bm{y},\qquad \mathcal{D}_{\lambda_{2}}^{-1}\bm{u}\mathcal{D}_{\lambda_{2}}=\lambda_{2}\bm{u},\qquad \mathcal{D}_{\lambda_{2}}^{-1}\tilde{\bm{y}}\mathcal{D}_{\lambda_{2}}=\tilde{\bm{y}},\qquad \mathcal{D}_{\lambda_{2}}^{-1}\tilde{\bm{u}}\mathcal{D}_{\lambda_{2}}={\lambda}^*_{2}\tilde{\bm{u}}.
\ee
The Fourier transformation
\be  
\mathcal{F}f(\mu,m)=\frac{1}{N}\sum_{m^{\prime}\in\mathbb{Z}/N\mathbb{Z}}\int d\mu^{\prime}e^{\frac{2\pi i}{N}\left(\mu\mu^{\prime}-mm^{\prime}\right)}f\left(\mu^{\prime},m^{\prime}\right)
\ee
satisfies
\be 
\cf\bm{u}\cf^{-1} =\bmy,\quad \cf\bm{y}^{-1}\cf^{-1} =\bm{u},\quad \cf\tilde{\bm{u}}\cf^{-1} =\tilde{\bmy},\quad \cf\tilde{\bm{y}}^{-1}\cf^{-1} =\tilde{\bm{u}}.
\ee
Applying these unitary transformations to $Q''_2$ and $\tilde{Q}_2''$ leads to the simplification
\be 
\bm{L}\equiv{\cal F}\mathcal{D}_{\lambda_{2}}^{-1}\mathcal{S}_{\lambda_{1}}Q_{2}^{\prime\prime}\mathcal{S}_{\lambda_{1}}^{-1}\mathcal{D}_{\lambda_{2}}{\cal F}^{-1}&=&\bm{y}^{-1}+\bm{y}+\bm{u},\\
\tilde{\bm{L}}\equiv{\cal F}\mathcal{D}_{\lambda_{2}}^{-1}\mathcal{S}_{\lambda_{1}}\tilde Q_{2}^{\prime\prime}\mathcal{S}_{\lambda_{1}}^{-1}\mathcal{D}_{\lambda_{2}}{\cal F}^{-1}&=&\tilde{\bm{y}}^{-1}+\tilde{\bm{y}}+\tilde{\bm{u}}.
\ee
These pair of operators are generalizations to the Dehn-twist operators in quantum Teichm\"uller theory, which corresponds to $N=1$ \cite{Kashaev:2000ku,Derkachov:2013cqa}. The Dehn-twist operators are unitary equivalent to the Casimir operators of the tensor product representations for the modular double of $U_q(sl(2,\R))$ \cite{Ponsot:2000mt,Derkachov:2013cqa,Nidaiev:2013bda}.

The spectral decomposition of $\bmL,\tilde{\bmL}$ will be discussed in detail in Section \ref{Eigenstates of the trace operators}. The analysis shows that $\dim \ch(\ell,\ell^*)=1$ indeed holds.

\section{Quantization of flat connections on 4-holed sphere}\label{Quantum flat connections on 4-holed sphere}

\subsection{Fock-Goncharov coordinates and quantization}

In this section, we study the quantization of $\rm{SL}(2,\mathbb{C})$ flat connections on 4-holed sphere. The moduli space of framed $\rm{SL}(2,\mathbb{C})$ flat connections on $n$-holed sphere is a Poisson manifold. A set of useful coordinates are known as the Fock-Goncharov (FG) coordinates, each of which associate to an edge in an ideal triangulation of the $n$-holed sphere. See Appendix \ref{Fock-Goncharov coordinate and holonomies} for details. FIG.\ref{degtriangulation} is an example of the ideal triangulation of 4-holed sphere. Given any ideal triangulation on $n$-holed sphere, we denote the FG coordinate by $z_E= \exp({Z_E}),\ \tilde{z}_E= \exp({\tilde{Z}_E})$, for all edges $E$ of the triangulation. Here $Z_E,\tilde{Z}_E$ are lifts of $\log(z_E),\log(\tilde{z}_E)$. We define their classical Poisson bracket by $\{z_E,z_{E'}\}=2 \eps_{E,E'}z_Ez_{E'},\ \{\tilde{z}_E,\tilde{z}_{E'}\}=2 \eps_{E,E'}\tilde{z}_E\tilde{z}_{E'}$. Their quantization $[\bm{O},\bm{O}']=\hbar\{O,O'\}$ and $[\tilde{\bm{O}},\tilde{\bm{O}}']=\tilde{\hbar}\{\tilde{O},\tilde{O}'\}$ leads to the operator algebra 
\be
\bm{z}_E\bm{z}_{E'}={\bfq}^{2\eps_{E,E'}}\bm{z}_{E'}\bm{z}_E,\qquad \tilde{\bm{z}}_E\tilde{\bm{z}}_{E'}=\tilde{{\bfq}}^{2\eps_{E,E'}}\tilde{\bm{z}}_{E'}\tilde{\bm{z}}_E,\qquad \bm{z}_E\tilde{\bm{z}}_{E'}=\tilde{\bm{z}}_{E'}\bm{z}_E\label{opalgFG}
\ee
$\eps_{E,E'}=0,\pm1,\pm2$ counts the number of oriented triangles shared by edges $E,E'$, The contribution from each triangle is $+1$ ($-1$) if $E$ rotates to $E'$ counterclockwisely (clockwisely) in the triangle. $\eps_{E,E'}$ for the triangulation in FIG.\ref{degtriangulation} reads
\be
\eps=\left(
\begin{array}{cccccc}
 0 & -1 & 0 & 1 & 1 & -1 \\
 1 & 0 & -1 & 0 & 0 & 0 \\
 0 & 1 & 0 & -1 & -1 & 1 \\
 -1 & 0 & 1 & 0 & 0 & 0 \\
 -1 & 0 & 1 & 0 & 0 & 0 \\
 1 & 0 & -1 & 0 & 0 & 0 \\
\end{array}
\right).
\ee
Similar as the above, we often write $\bm{z}_E=\exp({\bm{Z}_E})$ and $\tilde{\bm z}_E=\exp({\tilde{\bm Z}_E})$. The algebra has centers given by 
\be
\prod_{E}\bm{z}_E,\qquad \prod_{E}\tilde{\bm{z}}_E
\ee 
where the products are over edges incident to a given hole. This quantization can be applied to all ideal triangulations, and different triangulation are associated with different data $(\{E\},\eps_{E,E'})$.

The Poisson brackets defined above can be derived from the $\Slc$ Chern-Simons theory \cite{GMN09,DGV}. Therefore, as the quantization of these poisson brackets, the operator algebra \eqref{opalgFG} relates to the quantization of the Chern-Simons theory on 4-holed sphere \cite{levelk,Andersen2014}. In the following, we consider the quantization of the theory at level-$N$ on a fixed triangulation FIG.\ref{degtriangulation}. The quantization on different triangulations should be related by unitary transformations representing the cluster transformations of the FG coordinates \cite{FG03,2008InMat.175..223F}, see e.g. \cite{Teschner:2003em} for the case of the quantum Teichm\"uller theory. Section \ref{Changing triangulation} gives a generalization of the unitary cluster transformation to include the level $N$ in the case of a 4-holed sphere.

\begin{figure}[h]
\centering
\includegraphics[width=0.7\textwidth]{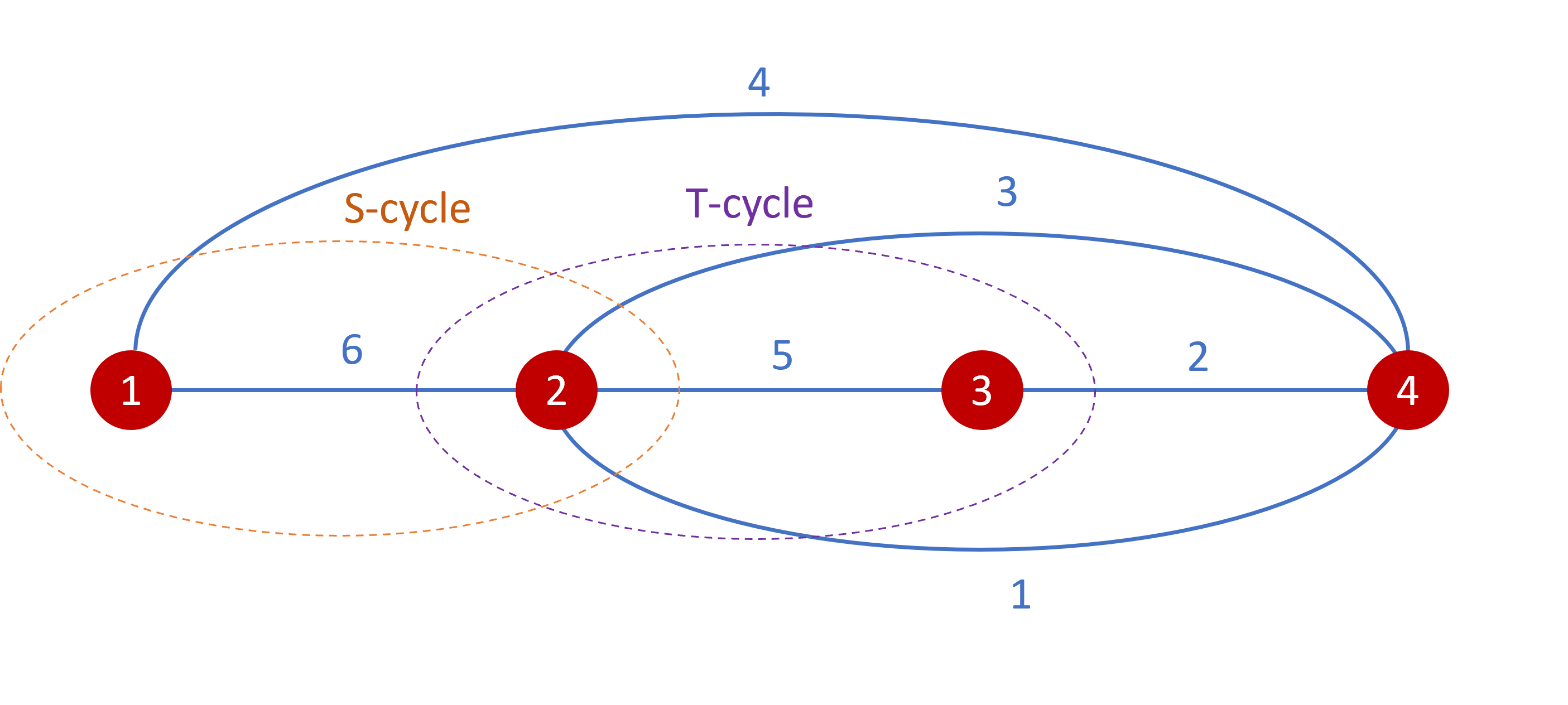}
\caption{An ideal triangulation of a 4-holed sphere, the S-cycle (orange) and T-cycle (purple).
}
\label{degtriangulation}
\end{figure}

Let us focus on the triangulation in FIG.\ref{degtriangulation}. The representation of \eqref{opalgFG} on $\mathcal{H}\simeq L^2(\R)\otimes \C^N$ can be constructed as the following: Firstly, we relate $\bm{X},\bm{Y}$ to $\bm{Z}_1,\bm{Z}_5$ by
\be
-\bm{U} =\bm{Z}_{1},\qquad -\bm{Y}+i\pi =\bm{Z}_{5},\quad\text{i.e.}\quad  \bm{z}_1=\bm{u}^{-1},\qquad \bm{z}_5=-\bm{y}^{-1}.\label{Z1andZ5}
\ee
where the representations of $\bm{u},\bm{y}$ are the same as the above in \eqref{repuandy}. $i\pi$ in $\bm{Z}_5$ or the minus sign in $\bm{z}_5$ is a choice of quantization and relates to the plus sign in the holonomy-trace operators of Sections \ref{S-cycle trace operator} and \ref{T-cycle trace operator}, so that the resulting holonomy-trace operators are unitary equivalent to the Casimirs $Q_{21},\tilde{Q}_{21}$ of $\suquqt$ and relates to the quantum Teichm\"uller theory when $N=1$.


For the triangulation in FIG.\ref{degtriangulation}, the center of the algebra is given
by $\l_a^2=e^{2L_a}$, where
\be
-L_{1}&=&\frac{1}{2}\left(\bm{Z}_{4}+\bm{Z}_{6}-2\pi i\right),\qquad -L_{2}=\frac{1}{2}\left(\bm{Z}_{1}+\bm{Z}_{3}+\bm{Z}_{5}+\bm{Z}_{6}-4\pi i\right),\nonumber\\
-L_{3}&=&\frac{1}{2}\left(\bm{Z}_{2}+\bm{Z}_{5}-2\pi i\right),\qquad -L_{4}=-\frac{1}{2}\left(\bm{Z}_{1}+\bm{Z}_{2}+\bm{Z}_{3}+\bm{Z}_{4}-4\pi i\right).\label{LZ1234}
\ee
The convention of the multiple of $\pi i$'s in these relations relates to the conventions of lifts of $L_a$ and $\bm{Z}_a$. Our convention follows from e.g. \cite{DGV,DGG11,Dimofte2011}, see also the discussion in Appendix \ref{Fock-Goncharov coordinate and holonomies}. Rigorously speaking, the FG coordinates are data for $\PSlc$ flat connections. Defining the lifts is useful to promote the coordinates to $\Slc$ flat connections, so that the traces of holonomies in the next subsection are defined without $\pm$ sign ambiguity. See Appendix \ref{Fock-Goncharov coordinate and holonomies} for some details.

We may express $\bm{Z}_{1},\cdots,\bm{Z}_{6}$ in terms of $\bm{U},\bm{Y},L_{1},\cdots,L_{4}$
\begin{align}\label{repZa}
\bm{Z}_{1} & =-\bm{U},\qquad 
\bm{Z}_{2}=-2L_{3}+\bm{Y}+i\pi,\nonumber\\
\bm{Z}_{3}&=L_{1}-L_{2}+L_{3}+L_{4}+\bm{U}+2i\pi,\nonumber\\
\bm{Z}_{4} & =-L_{1}+L_{2}+L_{3}+L_{4}-\bm{Y}+i\pi,\qquad \bm{Z}_{5}=-\bm{Y}+i\pi,\nonumber\\
\bm{Z}_{6}& =-L_{1}-L_{2}-L_{3}-L_{4}+\bm{Y}+i\pi
\end{align}
We find the representation of all $\bm{z}_e$ by
\begin{align}
\bm{z}_{1} & =\bm{u}^{-1},\qquad
\bm{z}_{2}=-\lambda_{3}^{-2}\bm{y},\qquad
\bm{z}_{3}=\frac{\lambda_{1}\lambda_{3}\lambda_{4}}{\lambda_{2}}\bm{u},\\
\bm{z}_{4} & =-\frac{\lambda_{2}\lambda_{3}\lambda_{4}}{\lambda_{1}}\bm{y}^{-1},\qquad
\bm{z}_{5}=-\bm{y}^{-1},\qquad
\bm{z}_{6}=-\frac{1}{\lambda_{1}\lambda_{2}\lambda_{3}\lambda_{4}}\bm{y}.
\end{align}
$\lambda_{1},\cdots,\lambda_{6} \in\mathbb{C}^\times$ are constants labelling the representation, and they relate to the centers of the operators algebra of $\bm{z}_E$, $E=1,\cdots,6$, by
\be
\bm{z}_4\bm{z}_6=\l_1^{-2},\qquad \bm{z}_1\bm{z}_3\bm{z}_5\bm{z}_6=\l_2^{-2},\qquad \bm{z}_2\bm{z}_5=\l_3^{-2},\qquad \bm{z}_1\bm{z}_3\bm{z}_2\bm{z}_4=\l_4^{2},
\ee

The tilded operators can be defined similarly
\begin{align}
	\tilde{\bm{z}}_{1} & =\tilde{\bm{u}}^{-1},\qquad
	\tilde{\bm{z}}_{2}=-\tilde{\lambda}_{3}^{-2}\tilde{\bm{y}},\qquad
	\tilde{\bm{z}}_{3}=\frac{\tilde{\lambda}_{1}\tilde{\lambda}_{3}\tilde{\lambda}_{4}}{\tilde{\lambda}_{2}}\tilde{\bm{u}},\\
	\tilde{\bm{z}}_{4} & =-\frac{\tilde{\lambda}_{2}\tilde{\lambda}_{3}\tilde{\lambda}_{4}}{\tilde{\lambda}_{1}}\tilde{\bm{y}}^{-1},\qquad
	\tilde{\bm{z}}_{5}=-\tilde{\bm{y}}^{-1},\qquad
	\tilde{\bm{z}}_{6}=-\frac{1}{\tilde{\lambda}_{1}\tilde{\lambda}_{2}\tilde{\lambda}_{3}\tilde{\lambda}_{4}}\tilde{\bm{y}}.
\end{align}
where $\tilde{\bm{u}},\tilde{\bm{y}}$ are represented as in \eqref{repuandy1}.
$\tilde{\lambda}_{1},\cdots,\tilde{\lambda}_{6} \in\mathbb{C}^\times$ are also constants and label the representation. We set $\tilde{\l}_a=\l_a^*$ so that 
\be
\bm{z}_E^\dagger = \tilde{\bm z}_E
\ee
representing the star structure. The above representation on $\ch$ of the quantum FG coordinates may be denoted by $\varrho_{\{\l,\tilde{\l}\}}$ labelled by $\{\l,\tilde{\l}\}\equiv\{\l_a,\tilde{\l}_a\}_{a=1}^4$.

\subsection{Trace operators and relation to skein quantization}

The flat connection is an $\Slc$-representation of the fundamental group on 4-holed sphere modulo conjugations. The fundamental group on 4-holed sphere $\pi_1$ is generated by the loops $\g_1,\g_2,\g_3,\g_4$, each traveling counter-clockwisely around a hole and satisfying $\g_1\g_2\g_3\g_4=1$. The flat connection defines the $\Slc$ holonomies $H(\g_1),\cdots,H(\g_4)$ modulo gauge transformations. On the moduli space of flat connections, the holonomies are $\Slc$-valued functions of logarithmic FG coordinates $Z_1,\cdots,Z_6$ by the snake rule, as shown in Appendix \ref{Fock-Goncharov coordinate and holonomies}. These functions depend on the choice of gauge. But the traces of the holonomies are gauge invariant:
\be
A_1&=&\tr\lt(H(\g_1)\rt)=e^{-\frac{Z_4}{2}-\frac{Z_6}{2}+i \pi }+e^{\frac{Z_4}{2}+\frac{Z_6}{2}-i \pi },\nonumber\\
B_1&=&\tr\lt(H(\g_2)\rt)=e^{-\frac{Z_1}{2}-\frac{Z_3}{2}-\frac{Z_5}{2}-\frac{Z_6}{2}+2 i \pi }+e^{\frac{Z_1}{2}+\frac{Z_3}{2}+\frac{Z_5}{2}+\frac{Z_6}{2}-2 i \pi },\\
C_1&=&\tr\lt(H(\g_3)\rt)=e^{-\frac{Z_2}{2}-\frac{Z_5}{2}+i \pi }+e^{\frac{Z_2}{2}+\frac{Z_5}{2}-i \pi },\\
D_1&=&\tr\lt(H(\g_4)\rt)=e^{-\frac{Z_1}{2}-\frac{Z_2}{2}-\frac{Z_3}{2}-\frac{Z_4}{2}+2 i \pi }+e^{\frac{Z_1}{2}+\frac{Z_2}{2}+\frac{Z_3}{2}+\frac{Z_4}{2}-2 i \pi }.
\ee
In order to obtain a complete set of coordinate functions of the moduli space, we need to consider the traces of the S-cycle, T-cycle, and U-cycle holonomies. S-cycle is the loop enclosing the 1st and 2nd holes, T-cycle encloses the 2nd and 3rd holes (see FIG. \ref{degtriangulation}), and U-cycle encloses the 1st and 3rd holes (here $L_s,L_t,L_u$ should be distinguished with $L_1,\cdots,L_4$). 
\be  
L_s=-\tr\lt(H(\g_1)H(\g_2)\rt)&=&e^{-\frac{Z_1}{2}-\frac{Z_3}{2}-\frac{Z_4}{2}-\frac{Z_5}{2}}-e^{-\frac{Z_1}{2}+\frac{Z_3}{2}-\frac{Z_4}{2}-\frac{Z_5}{2}}+e^{-\frac{Z_1}{2}+\frac{Z_3}{2}+\frac{Z_4}{2}-\frac{Z_5}{2}}\nonumber\\
&&+e^{-\frac{Z_1}{2}+\frac{Z_3}{2}+\frac{Z_5}{2}-\frac{Z_4}{2}}-e^{-\frac{Z_1}{2}+\frac{Z_3}{2}+\frac{Z_4}{2}+\frac{Z_5}{2}}+e^{\frac{Z_1}{2}+\frac{Z_3}{2}+\frac{Z_4}{2}+\frac{Z_5}{2}},\\
L_t=-\tr\lt(H(\g_2)H(\g_3)\rt)&=&e^{-\frac{Z_1}{2}-\frac{Z_2}{2}-\frac{Z_3}{2}-\frac{Z_6}{2}}-e^{\frac{Z_1}{2}-\frac{Z_2}{2}-\frac{Z_3}{2}-\frac{Z_6}{2}}+e^{\frac{Z_1}{2}+\frac{Z_2}{2}-\frac{Z_3}{2}-\frac{Z_6}{2}}\nonumber\\
&&+e^{\frac{Z_1}{2}+\frac{Z_6}{2}-\frac{Z_2}{2}-\frac{Z_3}{2}}-e^{\frac{Z_1}{2}+\frac{Z_2}{2}+\frac{Z_6}{2}-\frac{Z_3}{2}}+e^{\frac{Z_1}{2}+\frac{Z_2}{2}+\frac{Z_3}{2}+\frac{Z_6}{2}},\\
L_u=-\tr\lt(H(\g_1)H(\g_3)\rt)&=&-e^{-Z_1-\frac{Z_2}{2}-\frac{Z_4}{2}-\frac{Z_5}{2}-\frac{Z_6}{2}}-e^{-\frac{Z_2}{2}+\frac{Z_4}{2}-\frac{Z_5}{2}-\frac{Z_6}{2}}+e^{-Z_1+\frac{Z_4}{2}-\frac{Z_2}{2}-\frac{Z_5}{2}-\frac{Z_6}{2}}\nonumber\\
&&-e^{-\frac{Z_2}{2}+\frac{Z_5}{2}-\frac{Z_4}{2}-\frac{Z_6}{2}}+e^{-Z_1+\frac{Z_5}{2}-\frac{Z_2}{2}-\frac{Z_4}{2}-\frac{Z_6}{2}}+2 e^{-\frac{Z_2}{2}+\frac{Z_4}{2}+\frac{Z_5}{2}-\frac{Z_6}{2}}\nonumber\\
&&-e^{-Z_1+\frac{Z_4}{2}+\frac{Z_5}{2}-\frac{Z_2}{2}-\frac{Z_6}{2}}-e^{Z_1+\frac{Z_4}{2}+\frac{Z_5}{2}-\frac{Z_2}{2}-\frac{Z_6}{2}}-e^{\frac{Z_2}{2}+\frac{Z_4}{2}+\frac{Z_5}{2}-\frac{Z_6}{2}}\nonumber\\
&&+e^{Z_1+\frac{Z_2}{2}+\frac{Z_4}{2}+\frac{Z_5}{2}-\frac{Z_6}{2}}-e^{-\frac{Z_2}{2}+\frac{Z_4}{2}+\frac{Z_5}{2}+\frac{Z_6}{2}}+e^{Z_1+\frac{Z_4}{2}+\frac{Z_5}{2}+\frac{Z_6}{2}-\frac{Z_2}{2}}\nonumber\\
&&-e^{Z_1+\frac{Z_2}{2}+\frac{Z_4}{2}+\frac{Z_5}{2}+\frac{Z_6}{2}}
\ee   
They satisfy a constraint \cite{Coman:2015lna,Nekrasov:2011bc}
\be   
\cc&=&L_s L_t L_u - L_s^2 - L_t^2 - L_u^2 - A_1 B_1 C_1 D_1 - A_1^2 - B_1^2 - C_1^2 - D_1^2 \nonumber\\
&& - (A_1 B_1 + C_1 D_1) L_s - (B_1 C_1 + D_1 A_1) L_t - (C_1 A_1 + B_1 D_1) L_u + 4 =0.
\ee   
The polynomial ring of $A_1,B_1,C_1,D_1,L_s,L_t,L_u$ quotient by $\cc$ gives the algebra of functions on the moduli space of $\Slc$ flat connections.

The quantization is given by applying \eqref{repZa} to $A_1,B_1,C_1,D_1,L_s,L_t,L_u$ and defines the following operators on $\ch$
\be  
\bm{A}_1&=&\l_1+\l_1^{-1},\qquad \bm{B}_1=\l_2+\l_2^{-1},\qquad \bm{C}_1=\l_3+\l_3^{-1},\qquad \bm{D}_1=\l_4+\l_4^{-1}\nonumber\\
\bmL_s&=&e^{-L_3-L_4+\bmY}+e^{L_3+L_4-\bmY}-e^{L_1-L_2+\bm{U}+\bmY}-e^{L_3+L_4+\bm{U}-\bmY}-e^{L_1-L_2+\bm{U}}-e^{L_3+L_4+\bm{U}},\nonumber\\
\bmL_t&=&e^{-L_2-L_3+\bmY}+e^{L_2+L_3-\bmY}-e^{L_2+L_3-\bm{U}-\bmY}-e^{-L_1-2 L_3-L_4-\bm{U}+\bmY}-e^{L_2-L_3-\bm{U}}-e^{-L_1-L_4-\bm{U}}\nonumber\\
\bmL_u&=&-e^{-L_1+L_3-\bm{U}-\bm{Y}}-e^{L_1+L_3+\bm U-\bm Y}-e^{L_2+L_4-\bm U-\bm Y}-e^{L_2+2 L_3+L_4-\bm U-2 \bm Y}-e^{L_2+2
   L_3+L_4+\bm U-2 \bm Y}\nonumber\\
&&-e^{L_2+2L_3+L_4+\bm U-\bm Y}-e^{-L_1-L_3-\bm U}-e^{L_1+L_3+\bm U}+e^{-L_1+L_3-\bm Y}+e^{L_1+L_3-\bm Y}+e^{L_2+L_4-\bm Y}\nonumber\\
&&+\lt(\bfq +\bfq^{-1}\rt){e^{L_2+2 L_3+L_4-2\bm Y}}+e^{L_2+2 L_3+L_4-\bm Y}
\ee
where $\bm{A}_1,\bm{B}_1,\bm{C}_1,\bm{D}_1$ are proportional to identity operator on $\ch$. These operators satisfy the following commutation relations
\be 
\bfq^{-1} \bmL_s  \bmL_t - \bfq \bmL_t \bmL_s &=& (\bfq^{-2} - \bfq^2) \bmL_u - (\bfq - \bfq^{-1}) (\bm{A}_1 \bm{C}_1 + \bm{B}_1 \bm{D}_1),\label{skein1}\\
\bfq^{-1}\bmL_t \bmL_u - \bfq \bmL_u \bmL_t &=& (\bfq^{-2} - \bfq^2) \bmL_s - (\bfq - \bfq^{-1}) (\bm{A}_1 \bm{B}_1 + \bm{C}_1 \bm{D}_1),\\
\bfq^{-1}\bmL_u \bmL_s - \bfq \bmL_s \bmL_u &=& (\bfq^{-2} - \bfq^2) \bmL_t  - (\bfq - \bfq^{-1}) (\bm{B}_1 \bm{C}_1 + \bm{A}_1 \bm{D}_1),
\ee
and the quantized constraint
\be   
\bm{\cc}&=&\bfq^{-1} \bmL_s \bmL_t \bmL_u - \bfq^{-2} \bmL_s^2 - \bfq^2 \bmL_t^2 - \bfq^{-2}\bmL_u^2  \nonumber\\
&&- \bfq^{-1} (\bm{A}_1 \bm{B}_1 + \bm{C}_1 \bm{D}_1) \bmL_s - \bfq (\bm{B}_1 \bm{C}_1 + \bm{A}_1 \bm{D}_1) \bmL_t - \bfq^{-1} (\bm{A}_1 \bm{C}_1 + \bm{B}_1 \bm{D}_1) \bmL_u\nonumber\\
&&-\bm{A}_1 \bm{B}_1 \bm{C}_1 \bm{D}_1 - \bm{A}_1^2 - \bm{B}_1^2 - \bm{C}_1^2 -\bm{D}_1^2 +\lt( \frac{\bfq^2-\bfq^{-2}}{\bfq-\bfq^{-1}}\rt)^2=0.\label{skein4}
\ee   
The relations \eqref{skein1} - \eqref{skein4} coincide with the operator algebra $\ca_{0,4}^2$ from the skein quantization of $\Slc$ flat connections on 4-holed sphere \cite{Coman:2015lna}. Therefore, the operators $\bm{A}_1,\bm{B}_1,\bm{C}_1,\bm{D}_1,\bmL_s,\bmL_t,\bmL_u$ define the representation of $\ca_{0,4}^2$ carried by $\ch$. 

The Hermitian conjugates of $\bm{A}_1,\bm{B}_1,\bm{C}_1,\bm{D}_1,\bmL_s,\bmL_t,\bmL_u$ are denoted by the tilded operators $\tilde{\bm{A}}_1,\tilde{\bm{B}}_1,\tilde{\bm{C}}_1,\tilde{\bm{D}}_1,\tilde{\bmL}_s,\tilde{\bmL}_t,\tilde{\bmL}_u$. Their relations are given by \eqref{skein1} - \eqref{skein1} with $\bfq\to\tilde{\bfq}$ and the operators replaced by the tilded operators.

\subsection{S-cycle trace operator}\label{S-cycle trace operator}


The S-cycle trace operator can be written as
\be
\bm{L}_{s} 
&=&e^{-\bmY_{s}}+e^{\bmY_{s}}-e^{L_{3}+L_{4}}e^{-\bmY}{\bfq}^{-1}\left(e^{L_{1}-L_{2}-L_{3}-L_{4}}{\bfq}e^{\bmY}+1\right)\left(1+e^{\bmY}{\bfq}\right)e^{\bm U}, 
\ee
where $\bmY_{s}=\bmY-L_{3}-L_{4}$. Its tilded partner $\tilde{\bmL}_s=\bmL_s^\dagger$ is obtained analogously
\be 
\tilde{\bm{L}}_{s} 
&=&e^{-\tilde{\bmY}_{s}}+e^{\tilde{\bmY}_{s}}-e^{\tilde{L}_{3}+\tilde{L}_{4}}e^{-2\tilde{\bmY}}{\tilde{\bfq}}^{-1}\left(e^{\tilde{L}_{1}-\tilde{L}_{2}-\tilde{L}_{3}-\tilde{L}_{4}}{\tilde{\bfq}}e^{\tilde{\bmY}}+1\right)\left(1+e^{\tilde{\bmY}}{\tilde{\bfq}}\right)e^{\tilde{\bm U}}.
\ee

The trace operator $\bmL_s$ can be simplified by the unitary transformation
\be
\bm{U}'_{s}=\varphi\left(\bm{y}e^{L_{1}-L_{2}-L_{3}-L_{4}},\tilde{\bm{y}}e^{\tilde{L}_{1}-\tilde{L}_{2}-\tilde{L}_{3}-\tilde{L}_{4}}\right)\varphi\left(\bm{y},\tilde{\bm{y}}\right).
\ee
Indeed, the recursion relation of the quantum dilogarithm implies the following relation:
\be
\bm{U}_{s}^{\prime}\left[e^{L_{3}+L_{4}}{\bfq}^{-1}e^{-\bm{Y}}\left({\bfq}e^{L_{1}-L_{2}-L_{3}-L_{4}+\bm{Y}}+1\right)\left({\bfq}e^{\bm{Y}}+1\right)e^{\bm{U}}\right]\bm{U}_{s}^{\prime -1}=e^{-\bmY_{s}+\bm{U}}.
\ee
Therefore, we obtain
\be
\bm{U}'_{s}\bmL_s\bmU_s^{\prime -1}=e^{-\bmY_{s}}+e^{\bmY_{s}}-e^{-\bmY_{s}+\bm{U}}.
\ee

The operator can be further simplified by the Weil representation $\bm{T}$ of the T-type symplectic transformation \footnote{We refer to \cite{levelk} for a general discussion of T-type symplectic transformation and its Weil representation as unitary operator.}
\be
\left(\begin{array}{c}
	Y\\
	U
	\end{array}\right)\to \begin{pmatrix}1 & 0\\
	-1 & 1
	\end{pmatrix}\left(\begin{array}{c}
	Y\\
	U
	\end{array}\right),\qquad \bm{T}=(-1)^{\bm{m}^2}e^{\frac{\pi i}{N}\left(\bm{\mu}^{2}-\bm{m}^{2}\right)}.
\ee
The unitary operator $\bm{T}$ is well-defined on $\ch$ and satisfies
\be
\bm{T}e^{\bm{U}-\bm{Y}}\bm{T}^{-1}=-e^{\bm U},\label{Ttransform}
\ee
and commutes with $e^{\bm Y}$. As a result, we obtain
\be
\bmL_s'=\bmU_s\bmL_s\bmU_s^{-1}=e^{-\bmY +L_3+L_4}+e^{\bmY-L_3-L_4}+e^{\bm{U}+L_3+L_4},\qquad \bmU_s=\bmU_s'\bm{T}.
\ee
Note that the minus sign from the $\bm T$ transformation in \eqref{Ttransform} is a consequence from the integer level $N>1$. But this minus sign is cancelled due to the choice of quantization \eqref{Z1andZ5} so that $\bmL_s'$ has all plus signs.

A similar computation shows that $\bmU_s$ also simplify $\tilde{\bmL}_s$:
\be
\tilde{\bmL}_s'=\bmU_s\tilde{\bmL}_s\bmU_s^{-1}=e^{-\tilde{\bmY} +\tilde{L}_3+\tilde{L}_4}+e^{\tilde{\bmY}-\tilde{L}_3-\tilde{L}_4}+e^{\tilde{\bm U}+\tilde{L}_3+\tilde{L}_4}.
\ee


Here we assume both $e^{L_3},e^{L_4}$ can be parametrized by
\be
e^{L_3}=\exp\lt[\frac{2\pi i}{N}\lt(-ib\mu_{3}-m_3\rt)\rt],\qquad e^{L_4}=\exp\lt[\frac{2\pi i}{N}\lt(-ib\mu_{4}-m_4\rt)\rt]
\ee
where $\mu_3,\mu_4\in\R$ and $m_3,m_4\in\mathbb{Z}/N\Z$. For this paper, the purpose of this assumption is to further simplify $\bmL_s'$ and $\tilde{\bmL}_s'$ by unitary transformations and relate them to $\bmL,\tilde{\bmL}$ derived at the end of Section \ref{The Clebsch-Gordan decomposition}. However, when we consider the quantization of $\Slc$ Chern-Simons theory on 3-manifold whose boundary is a closed 2-surface containing the 4-holed sphere as a part of the geodesic boundary (see e.g. \cite{DGV,Han:2021tzw}), $\{L_a\}_{a=1}^4$ are part of the Darboux coordinate of the phase space, so $\{e^{L_a}\}_{a=1}^4$ should be quantized in the same way as $e^{\bm{Y}}$.

We apply the unitary shift operators 
\be
\cs_{34}f(\mu,m)=f\left(\mu+\mu_{3}+\mu_4,m+m_{3}+m_4\right),\qquad \cd_{34}=e^{\frac{2\pi i}{N}\lt[ ({\mu_3+\mu_4})\bm{\mu}-({m_3+m_4})\bm{m} \rt]}
\ee 
which gives 
\be
\cs_{34}e^{-{\bmY}+L_3+L_4}\cs_{34}^{-1}=e^{-{\bmY}},\qquad \cs_{34}e^{{\bm U}}\cs_{34}^{-1}=e^{{\bm U}}\\
\cd_{34}^{-1}e^{\bm{U}+L_3+L_4}\cd_{34}=e^{\bm U},\qquad \cd_{34}^{-1}e^{\bm{Y}}\cd_{34}=e^{\bm Y}
\ee
Therefore
\be
\cd_{34}^{-1}\cs_{34}\bmL_s'\cs_{34}^{-1}\cd_{34}&=&e^{-{\bmY}}+e^{{\bmY}}+e^{{\bm U}}=\bmL, \label{scyclefinal1}\\
\cd_{34}^{-1}\cs_{34}\tilde{\bmL}_s' \cs_{34}^{-1}\cd_{34}&=&e^{-{\tilde{\bmY}}}+e^{\tilde{\bmY}}+e^{\tilde{\bm U}}=\tilde{\bmL}.\label{scyclefinal2}
\ee
This shows that the sequence of unitary transformations on $\ch$
\be
\mathfrak{U}_s=\cd_{34}^{-1}\cs_{34}\bmU_s
\ee
relates the S-cycle trace operators $\bmL_s,\tilde{\bmL}_s$ to $\bmL,\tilde{\bmL}$ derived in Section \ref{The Clebsch-Gordan decomposition} from the Clebsch-Gordan decomposition of $\suquqt$ unitary representations:
\be
\bmL=\mathfrak{U}_s\bmL_s\mathfrak{U}_s^{-1},\qquad \tilde{\bmL}=\mathfrak{U}_s\tilde{\bmL}_s\mathfrak{U}_s^{-1}.
\ee

\subsection{T-cycle trace operator}\label{T-cycle trace operator}

The operator quantizing the trace of T-cycle enclosing the 2nd and 3rd holes is given by
\be
\bmL_t
&=&e^{-\bmY_{t}}+e^{\bmY_{t}}-e^{L_{2}+L_{3}}{\bfq}e^{-\bm{Y}}\left(1+{\bfq}^{-1}e^{-2L_{3}+\bm{Y}}\right)\left(1+{\bfq}^{-1}e^{-L_{1}-L_{2}-L_{3}-L_{4}+\bm{Y}}\right)e^{-\bm{U}}
\ee
where $\bmY_t=\bmY-L_{2}-L_{3}$. We obtain the tilded operator by $\tilde{\bmL}_t=\bmL_t^\dagger$. We apply the unitary transformation:
\be
\bmU'_{t}\bmL_t\bmU^{\prime -1}_{t}&=&e^{-\bmY+L_2+L_3}+e^{\bmY-L_2-L_3}-e^{-\bmY-\bm{U}+L_2+L_3},
\ee
where
\be
\bmU'_{t}&=&\varphi\left(e^{-\left(L_{1}+L_{2}+L_{3}+L_{4}\right)}\bm{y},e^{-\left(\tilde{L}_{1}+\tilde{L}_{2}+\tilde{L}_{3}+\tilde{L}_{4}\right)}\tilde{\bm{y}}\right)^{-1}\varphi\left(e^{-2L_{3}}\bm{y},e^{-2\tilde{L}_{3}}\tilde{\bm{y}}\right)^{-1}.
\ee

The further unitary transformation can be made by the Weil transformation representing the following symplectic transformation
\be
\left(\begin{array}{c}
	Y\\
	U
	\end{array}\right)\to \begin{pmatrix}-1 & 0\\
	-1 & -1
	\end{pmatrix}\left(\begin{array}{c}
	Y\\
	U
	\end{array}\right)=\begin{pmatrix}1 & 0\\
		1 & 1
		\end{pmatrix}\begin{pmatrix}-1 & 0\\
		0 & -1
		\end{pmatrix}\left(\begin{array}{c}
		Y\\
		U
		\end{array}\right).
\ee
We define the unitary operator $\bm{V}$ by $\bm{V}f(\mu,m)= f(-\mu,-m)$ and obtain
\be
\bmL_{t}'=\bmU_{t}\bmL_{t}\bmU_{t}^{-1}=e^{-\bm{Y}-L_2-L_3}+e^{\bm{Y}+L_2+L_3}+e^{\bm{U}+L_2+L_3},\qquad \bmU_{t}=\bm{T}^{-1}\bm{V}\bmU'_{t}.
\ee

Similar to $\bmL_s$, we parametrize $e^{L_2},e^{L_3}$ by
\be
e^{L_2}=\exp\lt[\frac{2\pi i}{N}\lt(-ib\mu_{2}-m_2\rt)\rt],\qquad e^{L_3}=\exp\lt[\frac{2\pi i}{N}\lt(-ib\mu_{3}-m_3\rt)\rt]
\ee
where $\mu_2,\mu_3\in\R$ and $m_2,m_3\in\mathbb{Z}/N\Z$. We apply the unitary shift operators 
\be    
\cs_{23}f(\mu,m)=f\lt(\mu+\mu_{2}+\mu_{3},m+m_2+m_3\rt),\qquad \cd_{23}=e^{\frac{2\pi i}{N}\lt[ ({\mu_2+\mu_3})\bm{\mu}-({m_2+m_3})\bm{m} \rt]},
\ee    
which gives
\be
\cs_{23}^{-1}e^{{\bmY}+L_2+L_3}\cs_{23}&=&e^{{\bmY}},\qquad \cs_{23}^{-1}e^{{\bm U}}\cs_{23}=e^{{\bm U}}\\
\cd_{23}^{-1}e^{\bm{U}+L_2+L_3}\cd_{23}&=&e^{\bm U},\qquad \cd_{23}^{-1}e^{\bm{Y}}\cd_{23}=e^{\bm Y}
\ee
As a result, the sequence of unitary transformations transforms the T-cycle trace operators to $\bmL,\tilde{\bmL}$
\be
\cd_{23}^{-1}\cs_{23}^{-1}\bmL_t'\cs_{23}\cd_{23}&=&e^{-{\bmY}}+e^{{\bmY}}+e^{{\bm U}}=\bmL,\\
\cd_{23}^{-1}\cs_{23}\tilde{\bmL}_t'\cs_{23}\cd_{23}&=&e^{-{\tilde{\bmY}}}+e^{\tilde{\bmY}}+e^{\tilde{\bm U}}=\tilde{\bmL}.
\ee
We denote the unitary operator $\Fu_t$ by
\be
\mathfrak{U}_t=\cd_{23}^{-1}\cs_{23}^{-1}\bm{U}_t,
\ee
so the above analysis shows that
\be
\bmL=\mathfrak{U}_t\bmL_t\mathfrak{U}_t^{-1},\qquad \tilde{\bmL}=\mathfrak{U}_t\tilde{\bmL}_t\mathfrak{U}_t^{-1},
\ee
and it also implies that there is a unitary transformation relating $(\bmL_t, \tilde{\bmL}_t)$ and $(\bmL_s,\tilde{\bmL}_s)$.
\be
\bmL_s=\mathfrak{U}_s^{-1}\mathfrak{U}_t\bmL_t\mathfrak{U}_t^{-1}\mathfrak{U}_s,\qquad \tilde{\bmL}_s=\mathfrak{U}_s^{-1}\mathfrak{U}_t\tilde{\bmL}_t\mathfrak{U}_t^{-1}\mathfrak{U}_s.
\ee
The operators $\bmL,\tilde{\bmL}$ are generalizations of the Dehn twist operator of quantum Teichm\"uller theory studied in \cite{Kashaev:2000ku,Nidaiev:2013bda,Teschner:2005bz}. The unitary transformation $\mathfrak{U}_s^{-1}\mathfrak{U}_t$ relating $(\bmL_s,\tilde{\bmL}_s)$ and $(\bmL_t,\tilde{\bmL}_t)$ can be seen as realizing the A-move of Moore-Seiberg groupoid on $\ch$, as a generalization from the result from quantum Teichm\"uller theory in \cite{Teschner:2003em}.


\section{Spectral decomposition}\label{Eigenstates of the trace operators}

The above discuss shows that both the S-cycle and T-cycle trace operators of quantized flat connection on 4-holed sphere, $(\bmL_s,\tilde{\bmL}_s)$ and $(\bmL_t,\tilde{\bmL}_t)$ are unitary equivalent to $(\bmL,\tilde{\bmL})$, which are derived from the unitary transformations acting on the Casimir operators $(Q_{21},\tilde{Q}_{21})$ of the tensor product representation of $\suquqt$ on $\ch\otimes\ch$. On one hand, the spectral decomposition of $(\bmL,\tilde{\bmL})$ is unitary equivalent to the Clebsch-Gordan decomposition for $\suquqt$ by Lemma \ref{CGirrep}, while on the other hand, diagonalizing  $(\bmL,\tilde{\bmL})$ gives the Fenchel-Nielsen (FN) representation of $\ch$, since classically $(L_s,\tilde{L}_s)$ and $(L_t,\tilde{L}_t)$ are the complexification of the Fenchel-Nielsen (FN) lengths on the 4-holed sphere, and $(\bmL_s,\tilde{\bmL}_s)$ and $(\bmL_t,\tilde{\bmL}_t)$ define their quantizations.

We define
\be
\g(-x,n)&=&\prod_{j=0}^{\infty}\frac{1-{\bfq}^{2j+1}\exp\left[\frac{2\pi i}{N}\left(-ibx\sqrt{N}+n\right)\right]}{1-\tilde{{\bfq}}^{-2j-1}\exp\left[\frac{2\pi i}{N}\left(-ib^{-1}x\sqrt{N}-n\right)\right]},
\ee
$\g(-x,n)=\mathrm{D}_b(x,n)$ is the quantum dilogarithm over $\R\times \Z/N\Z$ defined by Andersen and Kashaev in \cite{Andersen2014} (see also \cite{andersen2016level}). For $N=1$, $\g(x,n)=\g(x,0)\equiv\g(x)$ is Faddeev's quantum dilogarithm in e.g. \cite{Derkachov:2013cqa,Faddeev:1995nb}:
\be
\g(x)=\exp\lt[\frac{1}{4}\int_{\R+i0^+}\frac{\rmd w}{w}\frac{e^{2 i w x} }{\sinh \left(b^{-1}{w}\right) \sinh (b w)}\rt]
\ee
The relation with the quantum dilogarithm $\varphi$ used above is 
\be
\varphi\lt(-y,-\tilde{y}\rt)=\g\lt(-\frac{\mu}{\sqrt{N}},-m\rt)^{-1},\qquad y=e^{\frac{2\pi i}{N}(-ib\mu-m)},\quad \tilde{y}=e^{\frac{2\pi i}{N}(-ib^{-1}\mu+m)}
\ee
We introduce some short-hand notations that are useful below
\be
c_b=\frac{i}{2}(b+b^{-1}),\quad \omega = \frac{i}{2 b \sqrt{N}},\quad
\o'= \frac{i b}{2 \sqrt{N}},\quad
\o''= \frac{c_b}{\sqrt{N}},\quad
x = \frac{\mu }{\sqrt{N}},\quad
\lambda =\frac{\mu_r}{\sqrt{N}}.
\ee
Some useful properties of $\gamma(x,n)$ including the inverse relation, recursion relation, and integral identity can be found in Appendix \ref{Quantum dilogarithm}.

$\bmL,\widetilde{\bmL}$ has the following eigenstates
\be
\psi_r(\mu ,m)&=&\exp \left(\frac{i \pi  m^2}{N}+\frac{i \pi  m_r^2}{N}+i \pi  m_r^2-i \pi  m-i \pi  \left(x-\omega ''+i \epsilon \right)^2\right)\nonumber\\
&&\gamma \left(-\lambda +x-\omega ''+i \epsilon ,m-m_r\right) \gamma \left(\lambda +x-\omega ''+i \epsilon ,m_r+m\right),
\ee
It is straight-forward to check that $\psi_r$ satisfies the following eigen-equations
\be 
&&\bmL \psi_r(\mu ,m)=\lt(r+r^{-1}\rt)\psi_r(\mu ,m),\qquad \widetilde{\bmL}\psi_r(\mu ,m)=\lt({r}+{r}^{-1}\rt)^*\psi_r(\mu ,m),\\
&& \quad r=\exp\lt[\frac{2\pi i}{N}(-ib\mu_r-m_r)\rt],\qquad\quad {r}^*=\exp\lt[\frac{2\pi i}{N}(-ib^{-1}\mu_r+m_r)\rt].
\ee
$\psi_r(\mu ,m)$ is invariant manifestly under $(\mu_r,m_r)\to(-\mu_r,-m_r)$ and periodic under $m\to m+N$, $m_r\to m_r+N$. When $N=1$, $m=m_r=0$, $\psi_r$ reduces to Kashaev's eigenfunctions \cite{Kashaev:2000ku} ($\phi(x,\l)$ in \cite{Derkachov:2013cqa}) for the Dehn-twist operator.

The following two lemmas state that the set of $\psi_r$ forms a complete distributional orthogonal basis for $\ch$. The proofs are generalizations from \cite{Kashaev2001} for the quantum Teichm\"uller theory to include the discrete degrees of freedom labelled by $m,m_r$ due to the level $N$.

\begin{lemma}

The eigenstates $\psi_r$ satisfy the orthogonality
	\be
	\lag\psi_r\mid \psi_{r'}\rag=\sum_{m\in\Z/N\Z}\int\rmd\mu\,\psi_r(\mu,m)^*\psi_{r'}(\mu,m)=\frac{1}{4}\rho(\mu_r,m_r)\delta_{r,r'}.
	\ee
where
	\be
	\rho(\mu_r,m_r)&=&N^2\lt[\sin\lt(\frac{2\pi}{N}(ib\mu_r+m_r)\rt)\sin\lt(\frac{2\pi}{N}(-ib^{-1}\mu_r+m_r)\rt)\rt]^{-1},\\
	\delta_{r,r'}&=&\delta(\mu_r-\mu_{r'})\delta_{m_r,m_{r'}}+\delta(\mu_r+\mu_{r'})\delta_{m_r,-m_{r'}}.
	\ee


\end{lemma}

\begin{proof}


The integrand $\psi_r(\mu,m)^*\psi_{r'}(\mu,m)$ is given by
\be
(-1)^{m_r^2- m_{r'}^2}e^{-\frac{i \pi  (m_r^2- m_{r'}^2)}{N}+4 i \pi  x (\omega ''-i\epsilon)}\frac{\gamma \left(-\l'+x-\omega ''+i \epsilon ,m-{m_{r'}}\right) \gamma \left(\l'+x-\omega ''+i \epsilon ,m+{m_{r'}}\right)}{\gamma \left(-\lambda +x+\omega ''-i \epsilon ,m-{m_r}\right) \gamma \left(\lambda +x+\omega ''-i \epsilon ,m+{m_r}\right)},\nonumber\\
\label{integral111}
\ee
where $\lambda '={\mu_{r'}}/{\sqrt{N}}$. We use the integration identity \eqref{intIdgamma2} to transform the ratio of two $\gamma$'s \footnote{The variables involved in \eqref{intIdgamma2} are given by $t=x$, $\a=\l'-\o''+i\epsilon,\ \beta=\l'+\o''-i\epsilon$, $d=m$, $a=m_r'$, and $p=m_r$. We have $\im\left(\alpha +\frac{c_b}{\sqrt{N}}\right)=\epsilon>0,\ \im\left(\frac{c_b}{\sqrt{N}}-\beta \right)=\epsilon>0,\ \im(\alpha -\beta )=2\epsilon-\frac{b+b^{-1}}{\sqrt{N}}<0$ for sufficiently small $\epsilon>0$. Therefore there exist an integration contour of $s$ with $\im(\alpha -\beta )<\im(s)<0$ allowing \eqref{intIdgamma2} to hold.}
\be
&&\frac{\gamma \left(\l'+x-\omega ''+i \epsilon ,m+{m_{r'}}\right)}{\gamma \left(\lambda +x+\omega ''-i \epsilon ,m+{m_r}\right)}\nonumber\\
&=&\frac{\zeta _0 }{\sqrt{N}}\sum_{c\in\Z/N\Z}\int\rmd s\frac{e^{\frac{2 i \pi  c (m+m_{r'})}{N}+2 i \pi  s (x+\l')} \gamma \left(-\lambda '+\lambda +s+\omega ''-2 i \epsilon ,-c-{m_{r'}}+{m_r}\right)}{ \gamma \left(s+\omega '',-c\right) \gamma \left(-\lambda '+\lambda +\omega ''-2 i \epsilon ,{m_r}-{m_{r'}}\right)},\label{integral2}
\ee
where $\zeta_0=e^{-\frac{1}{12} i \pi  \left(N-4 c_b^2/N\right)}$.

Inserting \eqref{integral2} in \eqref{integral111}, we check that the integrand of $\int \rmd \mu \rmd s$ suppresses exponentially fast as $\mu,\re(s)\to\pm\infty$ for $\epsilon>0$ by the asymptotic behavior of the quantum dilogarithm (see Appendix \ref{Quantum dilogarithm}), so we can interchange the order of integration. Then we compute the integral and sum of $\mu=x\sqrt{N},m\in\Z/N\Z$ by applying \eqref{intIdgamma1} \footnote{
We replace the variables in \eqref{intIdgamma1} by $t\to x,\ \a\to-\l'-\o''+2i\epsilon,\ \b \to -\l  +\o''-2i\epsilon,\ s\to -s+2i\epsilon -2\o'',\ a\to -m_r',\ p\to-m_r,\ c\to -c$ (the left-hand sides are variables in \eqref{intIdgamma1}, while the right-hand sides are variables in \eqref{integral2}). We have $\im\left(\alpha +\frac{c_b}{\sqrt{N}}\right)\to\epsilon>0,\ \im\left(\frac{c_b}{\sqrt{N}}-\beta \right)\to\epsilon>0,\ \im(\alpha -\beta )\to 2\epsilon-\frac{b+b^{-1}}{\sqrt{N}}<0$ for sufficiently small $\epsilon>0$, and $\im(s)\to 2\epsilon-\frac{b+b^{-1}}{\sqrt{N}} -\im(s)$. Therefore for sufficiently small $\epsilon>0$, $\im(s)<0$ and small $|\im(s)|$ in \eqref{integral2} implies $\im(\alpha -\beta )<\im(s)<0$ in \eqref{intIdgamma1}, ensuring the validity of \eqref{intIdgamma1}.
} then using the inverse formula and recusion relation of $\g$. We obtain that as $\epsilon \to 0$, the result has the singularities at $(\mu_{r'},m_{r'})=\pm(\mu_{r},m_{r})$. Away from $(\mu_{r'},m_{r'})=(\mu_{r},m_{r})$, 
\be
\sum_{m\in\Z/N\Z}\int\rmd\mu\,\psi_r(\mu,m)^*\psi_{r'}(\mu,m)
&=&\frac{1}{4\sqrt{N}}\rho(\mu_r,m_r)\sum_{c\in\Z/N\Z}\int\rmd s\,e^{-\frac{2\pi i}{N}c \left({m_{r'}}+{m_r}\right)+2 i \pi  s \left(\lambda '+\lambda \right)}\nonumber\\
&=&\frac{1}{4}\rho(\mu_r,m_r)\delta(\mu_r+\mu_{r'})\delta_{m_r+m_{r'},0}\ .\label{deltadistribution}
\ee
This equation holds in the sense of tempered distribution\footnote{The result of $s$-integral is understood as a tempered distribution, which makes sense only when interring another integral, say, over $\mu_r=\sqrt{N}\l_r$. For finite $\epsilon$, the integrand of $\int \rmd s$ is $e^{-\frac{2\pi i}{N}c \left({m_{r'}}+{m_r}\right)+2 i \pi  s \left(\lambda '+\lambda \right)}$ multiplying a function $f_\epsilon (\l,\l',s)$ being Schwartz in $s$ and $\lim_{\epsilon\to0}f_\epsilon (\l,\l',s)=1$. Consider the following integral: $\int \rmd \l\rmd s\ e^{2\pi is(\l-\l')}f_\epsilon (\l,\l',s) F(\l)$ with smooth $F$ compact support on $\R_{\geq 0}$ and vanishing at zero. $\widetilde{F}_\epsilon (\l',s)=\int \rmd \l\ e^{2\pi is\l}f_\epsilon (\l,\l',s) F(\l)$ satisfies the bound $|\widetilde{F}_\epsilon (\l',s)|\leq C_{2,\l'}(1+|2\pi s|^2)^{-1}$ for any $\l'$ and any $N>0$. So we can use dominated convergence theorem: $\int\rmd s\,e^{-2\pi is \l'}\widetilde{F}_\epsilon(\l',s)\to \int\rmd s\,e^{-2\pi is \l'}\int \rmd \l\ e^{2\pi is\l} F(\l)=F(\l)$ as $\epsilon\to0$.}. The behavior at the other singularity follows from the reflection symmetry of $\psi_r$
\be
\sum_{m\in\Z/N\Z}\int\rmd\mu\,\psi_r(\mu,m)^*\psi_{r'}(\mu,m)
=\frac{1}{4}\rho(\mu_r,m_r)\lt[\delta(\mu_r-\mu_{r'})\delta_{m_r,m_{r'}}+\delta(\mu_r+\mu_{r'})\delta_{m_r,-m_{r'}}\rt]
\ee

\end{proof}

\begin{lemma}

The eigenstates satisfy the following resolution of identity on $\ch$
\be
4\sum_{m_r\in\Z/N\Z}\int_0^\infty\rmd\mu_r\, \rho(\mu_r,m_r)^{-1}\,\psi_r(\mu',m')\,\psi_r(\mu,m)^*=\delta(\mu,\mu')\,\delta_{m,m'}.
\ee
for $\mu,\mu'\in\R$ and $m,m'\in \Z/N\Z$.

\end{lemma}

\begin{proof}

Firstly we define $\sigma(\mu_r,m_r)$ such that $\rho(\mu_r,m_r)^{-1}=\sigma(\mu_r,m_r)+\sigma(-\mu_r,-m_r)$:
\be
\sigma(\mu_r,m_r)=\frac {e^{-4  i\pi\lambda\omega'' - \frac {2  i\pi   m_r} {N} + 8  i\pi   m_r \omega'{}^2 - 8  i\pi   m_r\omega'\omega'' } - e^{-4 i\pi\lambda\left (\omega'' - 2\omega' \right) + \frac {2  i\pi   m_r} {N} + 8  i\pi   m_r \omega'{}^2 - 8  i\pi   m_r\omega'\omega''}} {4  N^2},\nonumber
\ee
so that
\be
\sum_{m_r\in\Z/N\Z}\int_{-\infty}^\infty\rmd\mu_r\, \sig(\mu_r,m_r)\,\psi_r(\mu',m')\,\psi_r(\mu,m)^*  
&=&\sum_{m_r\in\Z/N\Z}\int_{0}^\infty\rmd\mu_r\, \rho(\mu_r,m_r)^{-1}\,\psi_r(\mu',m')\,\psi_r(\mu,m)^*.\nonumber
\ee
We consider the integral
\be
I&=&\sum_{m_r\in\Z/N\Z}\int_{-\infty}^\infty\rmd\mu_r\, \sig(\mu_r,m_r)\,\psi_r(\mu',m')\,\psi_r(\mu,m)^*\nonumber\\
&=&\sqrt{N}\sum_{m_r\in\Z/N\Z}\int_{-\infty}^\infty\rmd\l\,\sigma\, e^ {-\frac{i \pi  \left(m^2-{m'}^2\right)}{N}+i \pi  (m-{m'})+i \pi  \left(x^2-y^2\right)+2 i \pi  (x+y) (\omega '' -i\epsilon)}\nonumber\\
&&\frac{  \gamma \left(-\lambda -\omega ''+y+i \epsilon ,m'-{m_r}\right) \gamma \left(\lambda -\omega ''+y+i \epsilon ,{m_r}+m'\right)}{\gamma \left(-\lambda +x+\omega ''-i \epsilon ,m-{m_r}\right) \gamma \left(\lambda +x+\omega ''-i \epsilon ,{m_r}+m\right)},
\ee
and we make the following change of variables
\be
-t= \lambda ,\qquad \alpha = -\omega ''+y+i \epsilon ,\qquad\beta = x+\omega ''-i \epsilon ,\qquad a= m',\qquad p= m,\qquad -d= {m_r}.\nonumber
\ee
Then we apply \eqref{intIdgamma2} with $\mathrm{Im}(\a+\frac{c_b}{\sqrt{N}})=\mathrm{Im}(-\b+\frac{c_b}{\sqrt{N}})=\epsilon>0$ and $\mathrm{Im}(\a-\b)=2\epsilon-\frac{2\mathrm{Im}(c_b)}{\sqrt{N}}<\mathrm{Im}(s)<0$. The result is
\be
I=I_1+I_2
\ee
where $I_1,I_2$ correspond to two terms in $\sigma(\mu_r,m_r)$,
\be
I_1&=& \frac{\zeta _0}{4 N^2}\sum_{d,c\in\Z/N\Z}\int\rmd t\rmd s\, e^{-i \pi  \left(\alpha ^2-\beta ^2\right)+\frac{i \pi  \left(a^2+2 a c-p^2\right)}{N}-i \pi  (a-p)-\frac{2 i \pi}{N}  d c}\nonumber\\
&&e^{2 i \pi  s \left(\alpha +\omega ''\right)-2 i \pi  t (s-2\omega '')}\frac{ \gamma (\alpha +t,a+d) \gamma \left(-\alpha +\beta +s-\omega '',-a-c+p\right) }{ \gamma \left(s+\omega '',-c\right) \gamma (\beta +t,d+p) \gamma \left(-\alpha +\beta -\omega '',p-a\right)},
\ee
and 
\be
I_2&=&-\frac{\zeta _0}{4 N^2}\sum_{d,c\in\Z/N\Z}\int\rmd t\rmd s\,  e^{-i \pi  \left(\alpha ^2-\beta ^2\right)+\frac{i \pi  \left(a^2+2 a c-p^2\right)}{N}-i \pi  (a-p)-\frac{2 i \pi}{N}  d \left(c+2\right)}\nonumber\\
&&e^{2 i \pi  s \left(\alpha +\omega ''\right)-2 i \pi  t \left(s-2 \omega +2 \omega '\right)}\frac{ \gamma (\alpha +t,a+d) \gamma \left(-\alpha +\beta +s-\omega '',-a-c+p\right)}{ \gamma \left(s+\omega '',-c\right) \gamma (\beta +t,d+p) \gamma \left(-\alpha +\beta -\omega '',p-a\right)}.
\ee
We insert a regularization $\int\rmd t\to \int\rmd t\, e^{2\pi t\delta }$ in $I_1$. This may also be understood as a modification of the integration measure by inserting a factor $e^{-2\pi \l\delta }$ to the first term in $\sigma(\mu_r,m_r)$. We requiring $2\epsilon-\delta <\im(s)<0$. This condition implies $\mathrm{Im}(\a-\b)<\mathrm{Im}(s+i\delta-2\o'')<0$ and ensures that the integrand is a Schwarz function of $t,s$. Then we can interchange the order of integration and apply \eqref{intIdgamma1} to carry out the integration and sum of $t,d$
\be
I_1&=&\frac{\zeta _0^2}{4 N^{3/2}}\sum_{c\in\Z/N\Z}\int\rmd s\, e^{\frac{4 i \pi  c m'{}}{N} -\frac{i \pi  \left(m^2-m'{}^2\right)}{N}+i \pi  (m-m'{})+4 i \pi  s y +i \pi  \left(x^2-y^2\right)+2 i \pi  (x-y) \omega '' -2 \pi  \delta  y}\nonumber\\
&&\frac{ \gamma \left(s+x+\omega ''-y-2 i \epsilon ,-c+m-m'{}\right) \gamma \left(i \delta +s+x-\omega ''-y-2 i \epsilon ,-c+m-m'{}\right) }{ \gamma \left(s+\omega '',-c\right) \gamma \left(i \delta +s-\omega '',-c\right) \gamma \left(x+\omega ''-y-2 i \epsilon ,m-m'{}\right)^2}.\label{I1}
\ee
In $I_2$, the integrand is already a Schwarz function of $t,s$, and we have $\mathrm{Im}(s-2 \omega +2 \omega ')=\mathrm{Im}(s)<0$, so \eqref{intIdgamma1} can be applied without any regularization. 
\be
I_2&=&-\frac{\zeta _0^2}{4 N^{3/2}}\sum_{c\in\Z/N\Z}\int\rmd s\, e^{\frac{2 i \pi  (2 c m'{}+2 m'{})}{N}-\frac{i \pi  \left(m^2-m'{}^2\right)}{N}+i \pi  (m-m'{})+4 i \pi  s y+i \pi  \left(x^2-y^2\right)+2 i \pi  (x+3 y) \omega '+2 i \pi  \omega  (x-y)}\nonumber\\
&&\frac{ \gamma \left(s+x-y+\omega'' -2 i \epsilon ,-c+m-m'{}\right) \gamma \left(s+x+3 \omega '-y-\omega -2 i \epsilon ,-c+m-m'{}-2\right)}{\gamma \left(s+\omega '' ,-c\right) \gamma \left(s+3 \omega '-\omega ,-c-2\right) \gamma \left(x-y+\omega'' -2 i \epsilon ,m-m'{}\right)^2}.
\ee
The integrand of $I_1$ is suppressed asymptotically as $e^{\mp 4\pi \epsilon s}$ as $s\to \pm\infty$. 

The poles and zeros of $\mathrm{D}_b(x,n)=\gamma(-x,n)$ ($\im(b)>0$) is respectively at $x=\frac{1}{\sqrt{N}}({c_b}+ib^{-1}l+ibm)$ for $n=m-l+N\Z$ and $x=-\frac{1}{\sqrt{N}}({c_b}+ib^{-1}l+ibm)$ for $n=l-m+N\Z$, where $l,m\geq 0$ for both \cite{andersen2016level}. The integrand of $I_1$ has the following poles ($l,m\geq 0$ for below):

\begin{itemize}

\item $\gamma \left(s+x+\omega ''-y-2 i \epsilon ,-c+m-m'{}\right)=\rmD_b\left(-\frac{c_b}{\sqrt{N}}-s-x+y+2 i \epsilon ,-c+m-m'\right)$ has poles at
\be
s_{\rm pole} = -\frac{i l}{b \sqrt{N}}-\frac{i b m}{\sqrt{N}}-\frac{i}{b \sqrt{N}}-\frac{i b}{\sqrt{N}}-x+y+2 i \epsilon,\qquad n=m-l+N\Z.
\ee
These poles satisfy $\im(s_{\rm pole})\leq -2\re (b)/\sqrt{N}+2\epsilon$ and are below the integration contour of \eqref{I1}.

\item $\gamma \left(i \delta +s+x-\omega ''-y-2 i \epsilon ,-c+m-m'{}\right)=\rmD_b\left(\frac{c_b}{\sqrt{N}}-i \delta -s-x+y+2 i \epsilon ,-c+m-m'\right)$ has poles
\be
s_{\rm pole}= -\frac{i l}{b \sqrt{N}}-\frac{i b m}{\sqrt{N}}-i \delta -x+y+2 i \epsilon,\qquad n=m-l+N\Z.
\ee
These poles satisfy $\im(s_{\rm pole})\leq 2\epsilon-\delta$ and are below the integration contour of \eqref{I1}, since the contour of \eqref{I1} satisfy $2\epsilon-\delta\leq \im(s)<0$.

\item $ \gamma \left(s+\omega '',-c\right)^{-1}=\rmD_b\left(-\frac{c_b}{\sqrt{N}}-s,-c\right)^{-1}$ has poles
\be
s_{\rm pole }=\frac{i l}{b \sqrt{N}}+\frac{i b m}{\sqrt{N}},\qquad n=l-m+N\Z.
\ee
It contains the pole at the origin $s_{\rm pole}=0$ only when $c=0$, while other poles satisfy $\im(s_{\rm pole})\geq \re(b)/\sqrt{N}$.

\item $\gamma \left(i \delta +s-\omega '',-c\right)^{-1}=\rmD_b\left(\frac{c_b}{\sqrt{N}}-i \delta -s,-c\right)^{-1}$ has poles
\be
s_{\rm pole }=\frac{i l}{b \sqrt{N}}+\frac{i b m}{\sqrt{N}}+\frac{i}{b \sqrt{N}}+\frac{i b}{\sqrt{N}}-i \delta,\qquad n=l-m+N\Z.
\ee
These poles satisfy $\im(s_{\rm pole})\geq 2\re (b)/\sqrt{N}+\delta$

\end{itemize}

We deform the integration contour of $I_1$ to $s\to s+2\o'=s+\frac{ib}{\sqrt{N}}$ and a circle around $s_{\rm pole}=0$ only when $c=0$. The integration along $ s+2\o'$ cancels with $I_2$ as $\delta\to 0$. Therefore,  the nonvanishing contribution to $I$ is the residue at $s=0$ only coming from $c=0$ in $\sum_{c\in\Z/N\Z}$. The residue of ${\rm D}_{b}\left(\frac{c_{b}}{\sqrt{N}}+s,0\right)$ is 
\be
-\frac{\sqrt{N}}{2\pi b^{-1}}\frac{\left({\bfq}^{2},{\bfq}^{2}\right)_{\infty}}{\left(\tilde{{\bfq}}^{-2},\tilde{{\bfq}}^{-2}\right)_{\infty}}
\ee
and ${\rm D}_{\mathrm{b}}\left(-\frac{c_{b}}{\sqrt{N}}+s,0\right)\sim - \frac{2\pi b}{\sqrt{N}}  s (\bfq^2,\bfq^2)_\infty/(\tilde{\bfq}^{-2},\tilde{\bfq}^{-2})_{\infty}$ as $s \to 0$. The residue of $I_1$ at $s_{\rm pole}=0$ vanishes as $\delta\to0$ unless $m=m'$ and $x\to y$. As a result, for $\delta,\epsilon\to 0$ with $\delta>2\epsilon$, we obtain
\be
I=\delta_{m,m'}\frac{\delta }{8 \pi  \sqrt{N} (x-y-2 i \epsilon ) (i (\delta -2 \epsilon )+x-y)}\to\frac{1}{4}\delta(\mu-\mu')\delta_{m,m'}.
\ee

\end{proof}

For convenience, we denote by
\be
\varrho(\mu_r,m_r)=\frac{1}{4}\rho(\mu_r,m_r)=\frac{N^2}{4}\lt[\sin\lt(\frac{2\pi}{N}(ib\mu_r+m_r)\rt)\sin\lt(\frac{2\pi}{N}(-ib^{-1}\mu_r+m_r)\rt)\rt]^{-1}.
\ee
The above result indicates that for any $f\in\ch$, we can find the spectral representation $f(\mu_r,m_r)$ with $\mu_r\in\R_+$, $m_r\in\Z/N\Z$,
\be
\mid f\rangle &=&\sum_{m_r\in\Z/N\Z}\int_0^\infty\rmd\mu_r\,\varrho(\mu_r,m_r)^{-1}f(\mu_r,m_r)\mid \psi_r \rangle,\nonumber\\
\text{where}&&\quad f(\mu_r,m_r)=\langle\psi_r\mid f\rangle.
\ee
and 
\be
\langle f\mid f'\rangle=\sum_{m_r\in\Z/N\Z}\int_0^\infty\rmd\mu_r\,\varrho(\mu_r,m_r)^{-1}f(\mu_r,m_r)^*f'(\mu_r,m_r).
\ee
$\bmL,\tilde{\bmL}$ are represented as the multiplication operators
\be
\bmL f(\mu_r,m_r)=\ell(r) f(\mu_r,m_r),\qquad \tilde{\bmL} f(\mu_r,m_r)=\ell(r)^* f(\mu_r,m_r).\label{sprepplus}
\ee
where $\ell(r)=r+r^{-1}$.


\begin{theorem}\label{spectral}
	The spectral decomposition of $(\bmL,\tilde{\bmL})$ gives the following direct integral representation
	\be
	\ch\simeq \bigoplus_{m_r\in\Z/N\Z}\int_{\R_{\geq0}}^\oplus\rmd\mu_r\, \varrho(\mu_r,m_r)^{-1}\, \ch_{\mu_r,m_r}.\label{DID}
	\ee
	where each $\ch_{\mu_r,m_r}$ is 1-dimensional.

	
\end{theorem}

That $\ch_{\mu_r,m_r}$ is 1-dimensional implies $W(\l,\l^*)$ in Lemma \ref{CGirrep} is 1-dimensional. Moreover, the projective-valued measure $\rmd P_{\l,\tilde{\l}}$ in \eqref{CGequation} with $\l=\exp[\frac{2\pi i}{N}(-ib\mu_r-m_r)]$ and $\tilde{\l}=\l^*$ can be expressed as\footnote{$\psi_r$ with $\epsilon> 0$ is square-integrable and thus has well-defined Fourier transformation.}
\be   
\rmd P_{\l,\tilde{\l}}=\sum_{m_r'\in\Z/N\Z}\rmd \mu_r\rmd m_r\,\Theta(\mu_r)\delta(m_r-m'_r)\varrho(\mu_r,m_r)^{-1}\mathcal{S}_{\lambda_{1}}^{-1}\mathcal{D}_{\lambda_{2}}{\cal F}^{-1}|\psi_r\rangle\langle\psi_r|{\cal F}\mathcal{D}_{\lambda_{2}}^{-1}\mathcal{S}_{\lambda_{1}}.
\ee

Given that $|\psi_r\rangle$ diagonalizing $(\bmL,\tilde{\bmL})$, $\Fu_s^{-1}|\psi_r\rangle$ and $\Fu_t^{-1}|\psi_r\rangle$ diagonalize $(\bmL_s,\tilde{\bmL}_s)$ and $(\bmL_t,\tilde{\bmL}_t)$ respectively. Therefore, for any $f\in\ch$, 
\be    
f_s(\mu_r,m_r)=\langle \psi_r|\Fu_s|f\rangle, \qquad f_t(\mu_r,m_r)=\langle \psi_r|\Fu_t|f\rangle,
\ee
satisfy 
\be
\langle f\mid f'\rangle&=&\sum_{m_r\in\Z/N\Z}\int_0^\infty\rmd\mu_r\,\varrho(\mu_r,m_r)^{-1}f_{s}(\mu_r,m_r)^*f'_s(\mu_r,m_r) \\
&=&\sum_{m_r\in\Z/N\Z}\int_0^\infty\rmd\mu_r\,\varrho(\mu_r,m_r)^{-1}f_{t}(\mu_r,m_r)^*f'_t(\mu_r,m_r),
\ee
and
\be
\bmL_s f_s(\mu_r,m_r)&=&\ell(r) f_s(\mu_r,m_r),\qquad \tilde{\bmL}_s f_s(\mu_r,m_r)=\ell(r)^* f_s(\mu_r,m_r),\\
\bmL_t f_t(\mu_r,m_r)&=&\ell(r) f_t(\mu_r,m_r),\qquad \tilde{\bmL}_t f_t(\mu_r,m_r)=\ell(r)^* f_t(\mu_r,m_r).
\ee
Representing states in $\ch$ by functions $f_s(\mu_r,m_r)$ or $f_t(\mu_r,m_r)$ define the FN representation with respect to the S-cycle or T-cycle respectively.

\section{Changing triangulation}\label{Changing triangulation}

The above discussion is based on the ideal triangulation FIG.\ref{degtriangulation}. Changing ideal triangulation results in unitary equivalent quantization of FG coordinates. Let us denote the triangulation in FIG.\ref{degtriangulation}(a) by $\ct$ and consider the change of triangulation from $\ct$ to the tetrahedral triangulation $\ct'$ shown in FIG.\ref{changetriangulation}. The change is made by flipping the edge $3$ in the quadrilateral bounded by edges $4,6,5,2$. Other changes of triangulations have similar result. The edges in $\ct'$ are labelled by $1',\cdots,6'$, and their associated FG coordinates are denoted by $z_{E}'=\exp(Z_E')$, ($E=1,\cdots,6$). We may identify the triangulation and the associated data i.e. $\ct=(\{E\},\eps_{E_1,E_2})$ and $\ct'=(\{E'\},\eps'_{E'_1,E'_2})$. 

Classically, changing triangulation results in a change of FG coordinate. Taking $\ct\to\ct'$ as an example, $z_{E}'$ relates to $z_{E}$ by the following tranformation
\be
{z}_1'&=&{z}_1,\qquad {z}_3'={z}_3^{-1},\\
{z}_2'&=&{z}_2\lt(1-{z}_3\rt),\qquad {z}_4'{}^{-1}=(1-{z}_3^{-1}){z}_4^{-1},\\
{z}_6'&=&{z}_6\lt(1 - {z}_3\rt),\qquad {z}_5'{}^{-1}=(1 - {z}_3^{-1}){z}_5^{-1}.
\ee
where we denote $e'_a$ simply by $a$. The transformation preserves the Poisson bracket in a nontrivial manner, namely we have $\{z_a',z_{b}'\}={2} \eps_{a,b}'z_a'z_{b}',\ \{\tilde{z}_a',\tilde{z}_{b}'\}={2} \eps_{a,b}'\tilde{z}_a'\tilde{z}_{b}'$ from the Poisson bracket of $z_a,\tilde{z}_a$, whereas $\eps_{a,b}\neq \eps'_{a,b}$. $\eps'_{ab}$ is given by
\be
\eps'=\left(
\begin{array}{cccccc}
 0 & -1 & 0 & 1 & 1 & -1 \\
 1 & 0 & 1 & -1 & -1 & 0 \\
 0 & -1 & 0 & 1 & 1 & -1 \\
 -1 & 1 & -1 & 0 & 0 & 1 \\
 -1 & 1 & -1 & 0 & 0 & 1 \\
 1 & 0 & 1 & -1 & -1 & 0 \\
\end{array}
\right).
\ee

\begin{figure}[h]
\centering
\includegraphics[width=1\textwidth]{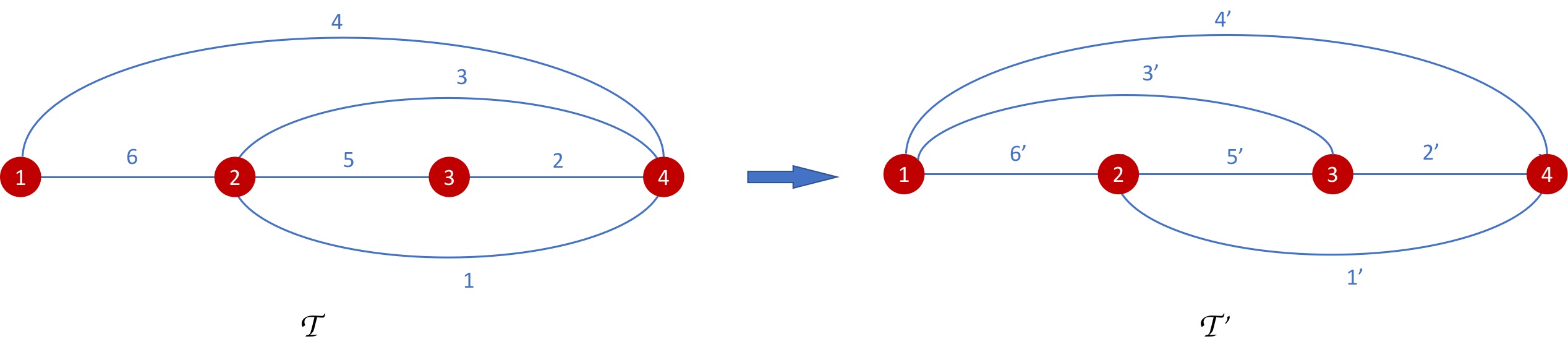}
\caption{Changing from the triangulation $\ct$ in FIG.\ref{degtriangulation}(a) to the tetrahedral triangulation $\ct'$}
\label{changetriangulation}
\end{figure}

The quantization on $\ct$ has been studied in the above. The quantization of the FG coordinates on $\ct'$ gives the quantum algebra
\be
\bm{z}'_{a}\bm{z}'_{b}={\bfq}^{2\eps'_{ab}}\bm{z}'_{b}\bm{z}'_a,\qquad \tilde{\bm{z}}'_a\tilde{\bm{z}}'_{b}=\tilde{{\bfq}}^{2\eps'_{ab}}\tilde{\bm{z}}'_{b}\tilde{\bm{z}}'_a,\qquad \bm{z}'_a\tilde{\bm{z}}'_{b}=\tilde{\bm{z}}'_{b}\bm{z}'_a\label{opalgFGprime}
\ee
The representation of this algebra relates to the representation of \eqref{opalgFG} (of $\ct$) by a unitary transformation known as the quantum cluster transformation, generalizing the results in \cite{2008InMat.175..223F,Teschner:2005bz}. Firstly, we define the map $i_3$ acting on the logarithmic coordinate
\be
i_3:\ \bm{Z}_a\mapsto\begin{cases}
-\bm{Z}_3+2\pi i, & a=3\\
\bm{Z}_a+\mathrm{Max}(\{\eps_{a,3},0\})(\bm{Z}_3-i\pi),& a\neq 3
\end{cases},
\ee
or explicitly,
\be
i_3\left(\bm{Z}_1\right)&=& -\bm{U} ,\qquad
i_3\left(\bm{Z}_2\right)\to -2 L_3+\bm{Y}+i\pi,\nonumber\\ 
i_3\left(\bm{Z}_3\right)&=& -L_1+L_2-L_3-L_4-\bm{U},\nonumber\\
i_3\left(\bm{Z}_4\right)&=& 2 L_3+2 L_4+\bm{U}- \bm{Y}+2i\pi,\qquad
i_3\left(\bm{Z}_5\right)= L_1-L_2+L_3+L_4+\bm{U}- \bm{Y}+2i\pi,\nonumber\\
i_3\left(\bm{Z}_6\right)&=& -L_1-L_2-L_3-L_4+ \bm{Y}+i\pi
\ee
and the same for $\tilde{\bm{Z}}_a$. They satisfy
\be
&&i_3\left(Z_3\right)+i_3\left(Z_4\right)+i_3\left(Z_6\right)-3\pi i=-2L_1,\qquad i_3\left(Z_1\right)+i_3\left(Z_5\right)+i_3\left(Z_6\right)-3\pi i=-2L_2,\\
&&i_3\left(Z_2\right)+i_3\left(Z_3\right)+i_3\left(Z_5\right)-3\pi i=-2L_3,\qquad i_3\left(Z_1\right)+i_3\left(Z_2\right)+i_3\left(Z_4\right)-3\pi i=2L_4
\ee
$i_3$ induces the monomial transformation of $\bm{z}_a$: $i_3(\bm{z}_a)=\exp[i_3(\bm{Z}_a)]$: 
\be
i_3\left(\bm{z}_1\right)&=&\bm{u}^{-1},\qquad 
i_3\left(\bm{z}_2\right)=-{\lambda _3^{-2}}\bm {y},\qquad 
i_3\left(\bm{z}_3\right)=\lambda _1^{-1} \lambda _2 \lambda _3^{-1} \lambda _4^{-1} \bm{u}^{-1},\nonumber\\
i_3\left(\bm{z}_4\right)&=&\bfq {\lambda _3^2 \lambda _4^2}\bm{u}\bm{y}^{-1},\qquad
i_3\left(\bm{z}_5\right)=\bfq \lambda _2^{-1}\lambda _1 \lambda _3 \lambda _4 \bm{u}\bm{y}^{-1},\qquad
i_3\left(\bm{z}_6\right)=-{(\lambda _1 \lambda _2 \lambda _3 \lambda _4)^{-1}}\bm{y}.\label{i3zi}
\ee
The image of $i_3$ gives operators on $\ch$ with the following commutation relation
\be
\lt[i_3(\bm{Z}_a),i_3(\bm{Z}_b)\rt]=2\hbar\eps'_{ab},\qquad \lt[i_3(\tilde{\bm{Z}}_a),i_3(\tilde{\bm{Z}}_b)\rt]=2\tilde{\hbar}\eps'_{ab}.
\ee
Therefore, the set of $i_3(\bm{z}_a)$, $i_3(\tilde{\bm{z}}_a)$ satisfies the same operator algebra as \eqref{opalgFGprime}. Then we define the unitary transformation 
\be
\varphi(-\bm{z}_3,-\tilde{\bm{z}}_3):\ \ch\to\ch'
\ee
where $\ch'$ carries the representation of \eqref{opalgFGprime}. Both $\ch$ and $\ch'$ as Hilbert spaces are isomorphic to $L^2(\R)\otimes \C^{N}$. $\varphi$ is the quantum dilogarithm and $\bm{z}_3=\frac{\lambda_{1}\lambda_{3}\lambda_{4}}{\lambda_{2}}\bm{u}$. Explicitly, for any state $f\in \ch$, the Fourier transform\footnote{The Fourier transformation is given by $\widetilde{f}(\nu, n)=\frac{1}{N} \sum_{m \in \mathbb{Z} / k \mathbb{Z}} \int d \mu\, e^{\frac{2 \pi i}{N}(\mu \nu-m n)} f(\mu, m)$.} of $f(\mu,m)$ gives $\widetilde{f}(\nu,n)$. $\bm{u}$ acts on $\widetilde{f}(\nu,n)$ as the multiplication of $u=\exp[\frac{2\pi i}{N}(-ib\nu-n)]$. Then $\varphi(-\bm{z}_3,-\tilde{\bm{z}}_3){f}(\mu,m)$ is given by the inverse Fourier transform of 
\be
\varphi\lt(-\frac{\lambda_{1}\lambda_{3}\lambda_{4}}{\lambda_{2}}u,-\frac{\tilde{\lambda}_{1}\tilde{\lambda}_{3}\tilde{\lambda}_{4}}{\tilde{\lambda}_{2}}\tilde{u}\rt)\widetilde{f}(\nu,n)\ .
\ee
The representation of $\bm{z}'_a,\tilde{\bm{z}}_a'$ is given by
\be
\bm{z}'_a=\varphi(-\bm{z}_3,-\tilde{\bm{z}}_3)i_3(\bm{z}_a)\varphi(-\bm{z}_3,-\tilde{\bm{z}}_3)^{-1},\qquad \tilde{\bm{z}}'_a=\varphi(-\bm{z}_3,-\tilde{\bm{z}}_3)i_3(\tilde{\bm{z}}_a)\varphi(-\bm{z}_3,-\tilde{\bm{z}}_3)^{-1}.\label{phii3phi}
\ee
We may compute \eqref{phii3phi} explicitly
\be
\bm{z}_1'&=&\bm{z}_1,\qquad \bm{z}_3'=\bm{z}_3^{-1},\nonumber\\
\bm{z}_2'&=&\bm{z}_2\lt(1-\bfq \bm{z}_3\rt),\qquad \bm{z}_4'{}^{-1}=(1-\bfq \bm{z}_3^{-1})\bm{z}_4^{-1},\nonumber\\
\bm{z}_6'&=&\bm{z}_6\lt(1-\bfq \bm{z}_3\rt),\qquad \bm{z}_5'{}^{-1}=(1-\bfq \bm{z}_3^{-1})\bm{z}_5^{-1}, \label{zaprime}
\ee
and similar for $\tilde{\bm{z}}'_a$. By the classical limit $\bfq=\exp[\frac{\pi i}{N}(b^2+1)]\to 1$ as $N\to\infty$, these transformations reduce to the classical transformation of FG coordinate under the flip. 

If we denote by $V_s$ the unitary transformation from $\ch$ to the direct integral representation for the S-cycle trace, the composition $V_s\circ \varphi(\bm{z}_3,\tilde{\bm{z}}_3)^{-1}$ maps from $\ch'$ to the direct integral representation.

\section*{Acknowledgements}

The author acknowledges an anonymous referee for the helpful comments on the earlier version of the manuscript and acknowledges Chen-Hung Hsiao and Qiaoyin Pan for some useful discussions. This work receives supports from the National Science Foundation through grant PHY-2207763, the College of Science Research Fellowship at Florida Atlantic University, and the Blaumann Foundation. The author also receives support from the visiting professorship at FAU Erlangen-N\"urnberg at the early stage of this work.

\appendix

\section{Quantum dilogarithm}\label{Quantum dilogarithm}

The quantum dilogarithm function is defined by
\be
\varphi\left(y,\tilde{y}\right) & \equiv\varphi_{{\bfq},\tilde{\bfq}}\left(y,\tilde{y}\right)=\left[\prod_{j=0}^{\infty}\frac{1+\bfq^{2j+1}y}{1+\tilde{\bfq}^{-2j-1}\tilde{y}}\right]^{-1}
\equiv\phi(\mu,m)
\ee
where $y=\exp[\frac{2\pi i}{N}(-i b\mu-m)]$, $\tilde{y}=\exp[\frac{2\pi i}{N}(-i b^{-1}\mu+m)]$, $\bfq=\exp[\frac{\pi i}{N}(b^2+1)]$, and $\tilde{\bfq}=\exp[\frac{\pi i}{N}(b^{-2}+1)]$. The quantum dilogarithm function satisfies the following recursion relations:
\begin{gather}
\varphi\left(\bfq^{2}y,\tilde{y}\right)=\left(1+\bfq y\right)\varphi\left(y,\tilde{y}\right),\qquad\varphi\left(y,\tilde{\bfq}^{2}\tilde{y}\right)=\left(1+\tilde{\bfq}\tilde{y}\right)\varphi\left(y,\tilde{y}\right),\label{recursion1}\\
\varphi\left(\bfq y,\tilde{y}\right)=\left(1+y\right)\varphi\left(\bfq^{-1}y,\tilde{y}\right),\qquad\varphi\left(y,\tilde{\bfq}\tilde{y}\right)=\left(1+\tilde{y}\right)\varphi\left(y,\tilde{\bfq}^{-1}\tilde{y}\right),\label{recursion2}\\
\varphi\left(y,\tilde{y}\right)=\left(1+\bfq^{-1}y\right)\varphi\left(\bfq^{-2}y,\tilde{y}\right),\qquad\varphi\left(y,\tilde{y}\right)=\left(1+\tilde{\bfq}^{-1}\tilde{y}\right)\varphi\left(y,\tilde{\bfq}^{-2}\tilde{y}\right).\label{recursion3}
\end{gather}
In the proofs of Lemmas \ref{lemma31}, \ref{lemma32} and \ref{CGirrep} and in Appendix \ref{Operator domains}, we often suppress the tilded entry of $\varphi$ and write $\varphi(y,\tilde{y})\equiv \varphi(y)$. 



$\varphi$ relates to the quantum dilogarithm ${\rm D}_b(x,n)$ in \cite{Andersen2014,andersen2016level} by
\be
\varphi\left(-y,-\tilde{y}\right)={\rm D}_b\lt(\frac{\mu}{\sqrt{N}},-m\rt)^{-1}.
\ee
We also introduce \footnote{This notation is inspired by \cite{Derkachov:2013cqa}.}
\be
\gamma(x,n)=\rmD_b(-x,n).\label{gammaD}
\ee 
For $N=1$, $\g(x,n)=\g(x,0)\equiv\g(x)$ is Faddeev's quantum dilogarithm in e.g. \cite{Derkachov:2013cqa,Faddeev:1995nb}:
\be
\g(x)=\exp\lt[\frac{1}{4}\int_{\R+i0^+}\frac{\rmd w}{w}\frac{e^{2 i w x} }{\sinh \left(b^{-1}{w}\right) \sinh (b w)}\rt]
\ee

The quantum dilogarithm functions satisfy the unitarity
\be
\phi\left(\mu,m\right)^{*}  =\phi\left(\mu^{*},m\right)^{-1},\qquad \mathrm{D}_b(x,n)^*=\mathrm{D}_b(x^*,n)^{-1},\qquad \g(x,n)^*=\g(x^*,n)^{-1}.
\ee
We introduce some notations $\o,\o',\o''$ by
\be 
\omega = \frac{i}{2 b \sqrt{N}},\quad
\o'= \frac{i b}{2 \sqrt{N}},\quad
\o''= \frac{c_b}{\sqrt{N}}\ .
\ee 
The following summarizes some useful properties of $\gamma(x,n)$:

\begin{itemize}

\item The inverse relation:
\be
\g(x,n)\g(-x,-n)=\exp \left(-\frac{i \pi  n^2}{N}-i \pi  n-\frac{i \pi  N}{6}+i \pi  x^2-\frac{1}{3} i \pi  \omega ''{}^2\right).
\ee

\item The recursion relation:
\be
\gamma (x\pm 2 \o',n\mp 1)&=&\lt(1+ e^{\frac{2 i \pi  n}{N}\pm\frac{i \pi  (N-1)}{N}+4 i \pi  x \omega '\pm4 i \pi  \omega '{}^2}\rt)^{\pm1}\gamma (x,n),\\
\gamma (x\pm 2 \omega ,n\pm 1)&=&\lt(1+e^{-\frac{2 i \pi  n}{N}\pm\frac{i \pi  (N-1)}{N}+4 i \pi  x \omega \pm 4 i \pi  \omega ^2}\rt)^{\pm1}\gamma (x,n).
\ee

\item The integration identity:
\be
&&\frac{1}{\sqrt{N}}\sum_{d\in \Z/N\Z}\int_\R\rmd t\frac{\gamma (t+\alpha ,a+d)}{\gamma (t+\beta,p+d)}e^{-2 i \pi  s t}e^{-\frac{2 i \pi cd}{N}} \nonumber\\
&=&\zeta_0\,e^{\frac{2 i \pi  a c}{N}+2 i \pi  s \left(\alpha +\omega ''\right)}\frac{ \gamma \left(-\alpha +\beta +s-\omega '',-a-c+p\right)}{\gamma \left(s+\omega '',-c\right) \gamma \left(-\alpha +\beta -\omega '',p-a\right)},\label{intIdgamma1}
\ee
where $\zeta_0=e^{-\frac{1}{12} i \pi  \left(N-\frac{4 c_b^2}{N}\right)}$ and $\a,\b,s$ satisfy
\be
\mathrm{Im}(\alpha+\frac{c_b}{\sqrt{N}})>0,\qquad \mathrm{Im}(-\beta+\frac{c_b}{\sqrt{N}})>0,\qquad \mathrm{Im}(\alpha-\beta)<\mathrm{Im}(s)<0.
\ee
The identity \eqref{intIdgamma1} translates to Theorem 3.10 in \cite{andersen2016level} by $t\to-x$, $\a\to-u$, $\b\to -v$, $s\to w$, and \eqref{gammaD}. The inverse of the identity \eqref{intIdgamma1} is given by
\be
&&\frac{\gamma (\alpha +t,a+d)}{\gamma (\beta +t,d+p)}\nonumber\\
&=&\frac{\zeta _0}{\sqrt{N}} \sum_{c\in \Z/N\Z}\int\rmd s\frac{\gamma \left(-\alpha +\beta +s-\omega '',-a-c+p\right)e^{\frac{2 i \pi  (a+d) c}{N}+2 i \pi  s \left(\alpha +\omega ''\right)+2 i \pi  s t} }{ \gamma \left(s+\omega '',-c\right) \gamma \left(-\alpha +\beta -\omega '',p-a\right)}.\label{intIdgamma2}
\ee

\item Asymptotic behavior: $\g(-\frac{c_{b}}{\sqrt{N}}-x,n)\sim \exp\left[i\pi\left(x+\frac{c_{b}}{\sqrt{N}}\right)^{2}+O(1)\right]$ as $\re(x)\to\infty$ and $\g(-\frac{c_{b}}{\sqrt{N}}-x,n)\sim O(1)$ as $\re(x)\to-\infty$ \cite{levelk}. 

\end{itemize}

\section{Operator domains}\label{Operator domains}

We denote by $\Fw\subset L^2(\R)$ the space of functions
\be
e^{-\alpha \mu^2+\b\mu}\, \mathrm{Pol}(\mu),\quad\text{where},\quad \re(\alpha)>0,\ \b\in\C,\ \text{and $\mathrm{Pol}(\mu)$ is a polynomial in $\mu$}.
\ee
$\Fw$ is dense because all Hermite functions are inside $\Fw$. We define $\sw\simeq \Fw\otimes\C^N\subset\ch$ and $\sw_2=\sw\otimes\sw$. Obviously, $\sw_2$ is a subset of the domain $\Fd_2$ of all Laurent polynomials of $\bm{u}_1,\bm{u}_2,\bmy_1,\bmy_2$ and their tilded partners. For the notation, the actions of $\bm{u}_1,\bmy_1$ and $ \bm{u}_2,\bmy_2$ are given by \eqref{repuandy} and \eqref{repuandy1} with the replacement $(\mu,m)\to(\mu_1,m_1)$ and $(\mu,m)\to(\mu_2,m_2)$. In the following, $y_a=\exp[\frac{2\pi i}{N}(-ib\mu_a-m_a)]$, $a=1,2$, and $y=\exp[\frac{2\pi i}{N}(-ib\mu-m)]$.


First, we would like to understand the image of the unitary operator $\cs_2^{-1}t^{-1}_{21}$ acting on $\sw_2$: for any $w\in\sw_2$, we denote by $f=\varphi\left(-y_{1}y_{2}^{-1}\right)^{-1}\mathcal{V}_{2}^{-1}w\left(\mu_{2},m_{2},\mu_{1},m_{1}\right)$ and
\be
&&t_{21}^{-1 }w=\mathcal{V}_{2}f\\
&=&\sum_{n,m}\int_\R d\nu e^{\frac{2\pi i}{N}\left(\mu_{2}\nu-m_{2}n\right)}(-1)^{n^{2}}e^{\frac{\pi i}{N}\left(\nu^{2}-n^{2}\right)}\int_\R d\mu e^{-\frac{2\pi i}{N}(\mu\nu-mn)}\varphi\left(-y_{1}y^{-1}\right)^{-1}\mathcal{V}_{2}^{-1}w\left(\mu,m,\mu_{1},m_{1}\right)\nonumber
\ee
First, $\mathcal{V}_{2}^{-1}w\in \sw_2$ since both Fourier transformation and multiplying $e^{-\frac{\pi i}{N}\left(\nu^{2}-n^{2}\right)}$ leave $\sw_2$ invariant. $f$ does not belong to $\Fd_2$ since $\varphi\left(-y_{1}y_{2}^{-1}\right)^{-1}$ introduces poles. $\varphi\left(-y,-\widetilde{y}\right)$ with $y=e^{\frac{2\pi i}{N}(-ib \mu -m)}$ has poles in lower half plane and zeros in upper half plane.
\be
&&\text{Poles: \ensuremath{\qquad}}\mu=-\frac{i}{2}\left(b+b^{-1}\right)-ib\alpha_{p}-ib^{-1}\beta_{p},\qquad\alpha_{p}-\beta_{p}-m=N\mathbb{Z},\qquad\alpha_{p},\beta_{p}\in\mathbb{Z},\qquad\alpha_{p},\beta_{p}\geq0,\nonumber\\
&&\text{Zeros: \ensuremath{\qquad}}\mu=\frac{i}{2}\left(b+b^{-1}\right)+ib\alpha_{0}+ib^{-1}\beta_{0},\qquad\alpha_{0}-\beta_{0}+m=N\mathbb{Z},\qquad\alpha_{0},\beta_{0}\in\mathbb{Z},\qquad\alpha_{0},\beta_{0}\geq0.\nonumber
\ee
$|\varphi\left(-y,-\widetilde{y}\right)|=1$ for real $\mu$. For complex $\mu$, $\varphi\left(-y,-\widetilde{y}\right)$ has the following asymptotic behavior:
\be
\varphi\left(-y,-\widetilde{y}\right)=
\begin{cases}
	1+O\left(e^{-\frac{2\pi\mathrm{Re}\left(b\right)}{N}|\mathrm{Re}(\mu)|}\right) & \mathrm{Re}(\mu)\to-\infty,\\
	\zeta_{\rm inv}e^{\pi in(n+N)/N}e^{-\frac{\pi i}{N}\mu^{2}}\left[1+O\left(e^{-\frac{2\pi\mathrm{Re}\left(b\right)}{N}\mathrm{Re}(\mu)}\right)\right] & \mathrm{Re}(\mu)\to-\infty.
	\end{cases}
\ee
where ${\zeta}_{\rm inv}=e^{\frac{i\pi}{6}\left(N+2c_{b}^{2}/N\right)}$. The function $\varphi(-y,-\widetilde{y})$ at most grows exponentially as $\re(\mu)\to\pm\infty$. $f\left(\mu,m,\mu_{1},m_{1}\right)=\varphi\left(-y_{1}y^{-1}\right)^{-1}\mathcal{V}_{2}^{-1}w\left(\mu,m,\mu_{1},m_{1}\right)$ has following properties on the analyticity and asymptotic behavior:

\begin{itemize}

\item $f$ is analytic in the regime $\mathrm{Im}\left(\mu_{1}-\mu\right)<\mathrm{Re}(b)$. In particular, along the integration contour $\mu\in\R$, $f$ is analytic for $\mathrm{Im}\left(\mu_{1}\right)<\mathrm{Re}(b)$.

\item In the regime where $f$ is analytic, $f$ is a Schwartz function of $\re(\mu),\re(\mu_1)$ and decays as Gaussian as $\re(\mu),\re(\mu_1)\to\pm\infty$.
\end{itemize}
Due to these properties, the Fourier transform of $f$: 
\be
\mathcal{F}[f](\nu,n;\mu_1,m_1)=\frac{1}{N}\sum_{m\in\Z/N\Z}\int_{\R}d\mu \, e^{-\frac{2\pi i}{N}(\mu\nu-mn)}f(\mu,m;\mu_1,m_1)
\ee
satisfies the following analyticity and asymptotic behaviors (in the following computations, we often neglect the unimportant overall factor $1/N^2$):
\begin{itemize}
\item Analyticity in $\nu$: $\mathcal{F}[f]$ is an entire function in $\nu$, and $\mathcal{F}[f](\nu,n;\mu_1,m_1)$ is a Schwartz function in $\re(\nu)$.

\item Asymptotics in $\nu$: $\mathcal{F}[f](\nu,n;\mu_1,m_1)\sim e^{-\frac{2\pi}{N}\left(\mathrm{Re}(b)-\mathrm{Im}\left(\mu_{1}\right)\right)\re(\nu)}$ as $\re(\nu)\to \infty$ and decays faster than $e^{a\re(\nu)}$ for any $a>0$ as $\re(\nu)\to-\infty$. The exponential decay of $\mathcal{F}[f]$ requires $\mathrm{Im}\left(\mu_{1}\right)<\mathrm{Re}(b)$

\item Asymptotics in $\mu_{1}$: for any $p\in\mathbb{Z}$, 
\be
y_{1}^{p}\mathcal{F}[f]\left(\nu,n;\mu_{1},m_{1}\right)=\sum_{m\in\mathbb{Z}/N\mathbb{Z}}\int_{\mathbb{R}}d\mu e^{-\frac{2\pi i}{N}(\mu\nu-mn)}y_{1}^{p}\varphi\left(-y_{1}y^{-1}\right)^{-1}\mathcal{V}_{2}^{-1}w\left(\mu,m,\mu_{1},m_{1}\right).
\ee 
The integrand decays as Gaussian at infinities and the Fourier transformation is unitary. So $y_{1}^{p}F\left(\nu,n;\mu_{1},m_{1}\right)\in{\cal H}\otimes{\cal H}$.

\end{itemize}
These properties imply that
\be
t_{21}^{-1}w\left(\mu_{2},m_{2},\mu_{1},m_{1}\right)=\sum_{n\in\mathbb{Z}/N\mathbb{Z}}\int d\nu e^{\frac{2\pi i}{N}\left(\mu_{2}\nu-m_{2}n\right)}(-1)^{n^{2}}e^{\frac{\pi i}{N}\left(\nu^{2}-n^{2}\right)}\cf[f]\left(\nu,n;\mu_{1},m_{1}\right)
\ee
satisfies the following analyticity and asymptotic behaviors:
\begin{itemize}

\item Analyticity: $t^{-1}_{21}w\left(\mu_{2},m_{2},\mu_{1},m_{1}\right)$ is analytic for $\mathrm{Im}\left(\mu_{2}\right)>-\mathrm{Re}(b)+\mathrm{Im}\left(\mu_{1}\right)$. 

\item Asymptotics in $\mu_{1}$: $y_{1}^{p}t^{-1}_{21}w\left(\mu_{2},m_{2},\mu_{1},m_{1}\right)\in{\cal H}\otimes{\cal H}$ for any $p\in\Z$.

\item Analyticity in $\mu_1$:  For $\mathrm{Im}(\mu_{1})<\mathrm{Re}(b)$, we can make small analytic rotation of the $\nu$-contour: 
\be
&&\sum_{n,m}\int d\nu e^{\frac{2\pi i}{N}\left(\mu_{2}\nu-m_{2}n\right)}e^{\frac{\pi i}{N}\left(\nu^{2}-n^{2}\right)}\int d\mu e^{-\frac{2\pi i}{N}(\mu\nu-mn)}f\left(\mu,m,\mu_{1},m_{1}\right)\nonumber\\
&=&\sum_{n,m}\int d\mu\sqrt{iN}e^{\frac{i\pi\left(-(\mu-\mu_{2})^{2}+2mn-n(2m_{2}+n)\right)}{N}}f\left(\mu,m,\mu_{1},m_{1}\right),
\ee
and there exists a Schwartz function $s\left(\mu,m\right)$ such that  $|\partial_{\mu_{1}}f\left(\mu,m,\mu_{1},m_{1}\right)|\leq s\left(\mu,m\right)$. We can interchange the derivative of $\mu_1$ with the integral, so the integral is analytic in $\mu_1$.

\end{itemize}
The analyticity of $t_{21}w\left(\mu_{2},m_{2},\mu_{1},m_{1}\right)$ implies that 
\be
\cs^{-1}_2 t^{-1}_{21}w\left(\mu_{2},m_{2},\mu_{1},m_{1}\right)=t^{-1}_{21}w\left(\mu_{2},m_{2},\mu_{1}+\mu_2,m_{1}+m_2\right)
\ee 
is analytic for $\mathrm{Im}\left(\mu_{1}\right)<\mathrm{Re}(b)$ and is an entire function in $\mu_2$.

Let us revisit the derivation in Lemma \ref{lemma31}: First, let us consider 
\be
t_{21}{\cal S}_2\bm{y}_{1}^{-1}{\cal S}_2^{-1}t_{21}^{-1}w, \qquad \text{where}\qquad t_{21}=\cv_2\varphi\left(-y_{1}y_{2}^{-1}\right)\mathcal{V}_{2}^{-1}
\ee
Recall that for any $\psi\in\ch$ satisfying $\psi\left(\nu,n\right)=\frac{1}{N^2}\sum_{m\in\Z/N\Z}\int d\mu e^{-\frac{2\pi i}{N}(\mu\nu-mn)}\psi\left(\mu,m\right)$ analytic in the strip $\im(\nu)\in[0,\re(b)]$ and $\psi \left(\nu+ib,n-1\right)\in{\cal H}$:
\be
\psi\left(\nu+ib,n-1\right)&=&\frac{1}{N}\sum_{m\in\Z/N\Z}\int d\mu e^{-\frac{2\pi i}{N}(\mu\nu-mn)}y\psi\left(\mu,m\right),\nonumber
\ee
and the inverse Fourier transformation of this equation gives
\be
y\psi\left(\mu,m\right)&=&\frac{1}{N}\sum_{n\in\Z/N\Z}\int d\nu e^{\frac{2\pi i}{N}(\mu\nu-mn)}\psi\left(\nu+ib,n-1\right) ,
\ee
while $y\psi\left(\mu,m\right)=\frac{1}{N}\sum_n\int d\nu e^{\frac{2\pi i}{N}(\mu(\nu-ib)-m(n+1))}\psi\left(\nu,n\right)$, so the shift of integration contour does not change the result. 

$\psi\left(\nu,n;\mu_{1},m_{1}\right)\equiv(-1)^{n^{2}}e^{\frac{\pi i}{N}\left(\nu^{2}-n^{2}\right)}\mathcal{F}\left[f\right]\left(\nu,n;\mu_{1},m_{1}\right)$ is entire in $\nu$ and 
\be 
&&\psi\left(\nu+ib,n-1;\mu_{1},m_{1}\right)=(-1)^{\left(n-1\right)^{2}}e^{\frac{\pi i}{N}\left(\left(\nu+ib\right)^{2}-\left(n-1\right)^{2}\right)}\mathcal{F}\left[f\right]\left(\nu+ib,n-1;\mu_{1},m_{1}\right)\nonumber\\
&=&-(-1)^{n^{2}}e^{\frac{\pi i}{N}\left(\nu^{2}-n^{2}\right)}e^{\frac{\pi i}{N}\left(-b^{2}-1\right)}e^{-\frac{2\pi i}{N}\left(-ib\nu-n\right)}\mathcal{F}\left[f\right]\left(\nu+ib,n-1;\mu_{1},m_{1}\right)\in{\cal H}\otimes{\cal H}
\ee
where 
\be
&&e^{-\frac{2\pi i}{N}\left(-ib\nu-n\right)}\mathcal{F}\left[f\right]\left(\nu+ib,n-1;\mu_{1},m_{1}\right)\nonumber\\
&=&\sum_{m\in\Z/N\Z}\int d\mu e^{-\frac{2\pi i}{N}(\mu\nu-mn)}\bm{q}^{2}e^{\frac{2\pi i}{N}(-ib\mu-m)}\varphi\left(-\bm{q}^{-2}y_{1}y^{-1}\right)^{-1}\mathcal{V}_{2}^{-1}w\left(\mu+ib,m-1;\mu_{1},m_{1}\right)\nonumber\\
&=&\sum_{m\in\Z/N\Z}\int d\mu e^{-\frac{2\pi i}{N}(\mu\nu-mn)}\bm{q}^{2}\left(y-\bm{q}^{-1}y_{1}\right)\varphi\left(-y_{1}y^{-1}\right)^{-1}\lt(\bm{u}\mathcal{V}_{2}^{-1}w\rt)\left(\mu,m;\mu_{1},m_{1}\right).\nonumber
\ee
We have shifted the $\mu$-contour by using the analyticity and asymptotic behavior of the integrand. Then we have
\be
&&\bm{y}_{1}^{-1}{\cal S}_2^{-1}\mathcal{V}_{2}f
={\cal S}_{2}^{-1}\lt[y_{1}^{-1}y_{2}\sum_{n\in\Z/N\Z}\int d\nu e^{\frac{2\pi i}{N}\left(\mu_{2}\nu-m_{2}n\right)}\psi\left(\nu,n;\mu_{1},m_{1}\right)\rt]\nonumber\\
&=&{\cal S}_{2}^{-1}\lt[y_{1}^{-1}\sum_{n\in\Z/N\Z}\int d\nu e^{\frac{2\pi i}{N}\left(\mu_{2}\nu-m_{2}n\right)}\psi\left(\nu+ib,n-1;\mu_{1},m_{1}\right)\rt]\nonumber\\
&=&\lt[{\cal S}_{2}^{-1}\mathcal{V}_{2}\left(1-\bm{q} y_{1}^{-1} y_{2}\right)\varphi\left(-y_{1} y_{2}^{-1}\right)^{-1}\bm{u}_{2}\mathcal{V}_{2}^{-1}w\rt]\left(\mu_{1},m_{1};\mu_{2},m_{2}\right)\in\ch\otimes\ch.
\ee
This shows that ${\cal S}_2^{-1}\mathcal{V}_{2}f$ is in the domain of $\bm{y}_{1}^{-1}$. Moreover, we obtain the result
\be
t_{21}{\cal S}_2\bm{y}_{1}^{-1}{\cal S}^{-1}_2t_{21}^{-1}w&=&\mathcal{V}_{2}\left(\bm{u}_{2}-\bm{q}^{-1}\bm{y}_{1}^{-1}\bm{u}_{2}\bm{y}_{2}\right)\mathcal{V}_{2}^{-1}w\nonumber\\
&=&\left(\bm{u}_{2}+\bm{y}_{1}^{-1}\bm{y}_{2}\right)w.
\ee

Next, we compute
\be
t_{21}{\cal S}_2\bm{u}_{1}^{-1}{\cal S}_2^{-1}t_{21}^{-1}w, \qquad \text{where}\qquad t_{21}=\cv_2\varphi\left(-y_{1}y_{2}^{-1}\right)\mathcal{V}_{2}^{-1}.
\ee
$\cs^{-1}_2 t^{-1}_{21}w\left(\mu_{2},m_{2},\mu_{1},m_{1}\right)$ is analytic for $\mathrm{Im}\left(\mu_{1}\right)<\mathrm{Re}(b)$. We have
\be
&&\cf[f]\left(\nu,n;\mu_{1}+\mu_2-ib,m_{1}+m_2+1\right)\nonumber\\
&=&\sum_{m\in\Z/N\Z}\int d\mu\, e^{-\frac{2\pi i}{N}(\mu\nu-mn)}\varphi\left(-\bm{q}^{-2}y_{1}y_{2}y^{-1}\right)^{-1}\mathcal{V}_{2}^{-1}w\left(\mu,m,\mu_{1}-ib+\mu_{2},m_{1}+1+m_{2}\right)
\ee
and
\be
&&\bm{u}_{1}^{-1}{\cal S}_2^{-1}\mathcal{V}_{2}f\nonumber\\
&=&\sum_{n,m\in\Z/N\Z}\int d\nu e^{\frac{2\pi i}{N}\left(\mu_{2}\nu-m_{2}n\right)}(-1)^{n^{2}}e^{\frac{\pi i}{N}\left(\nu^{2}-n^{2}\right)}\cf[f]\left(\nu,n;\mu_{1}+\mu_2-ib,m_{1}+m_2+1\right)\nonumber\\
&=&\cs_2^{-1}\lt[\mathcal{V}_{2}\left(1-\bm{q}^{-1}y_{1}y_{2}^{-1}\right)\varphi\left(-y_{1}y_{2}^{-1}\right)^{-1}\mathcal{V}_{2}^{-1}w\left(\mu_{1}-ib,m_{1}+1;\mu_{2},m_{2}\right)\rt]
\ee
The result belongs to $\ch\otimes\ch$, so ${\cal S}_2^{-1}\mathcal{V}_{2}f$ is in the domain of $\bm{u}_{1}^{-1}$. From this result, we obtain
\be	
t_{21}{\cal S}_2\bm{u}_{1}^{-1}{\cal S}^{-1}_2t_{21}^{-1}w=\left(\bm{u}_{1}^{-1}+e^{\bm{U}_{2}-\bm{Y}_{2}}e^{-\bm{U}_{1}+\bm{Y}_{1}}\right)w
\ee

A computation similar to the above gives
\be
t_{21}{\cal S}_2e^{-\bm{U}_{1}-\bm{Y}_{1}}{\cal S}^{-1}_2t_{21}^{-1}w&=&\left[\left(\bm{q}^{-1}+\bm{q}\right)e^{-\bm{U}_{1}+\bm{U}_{2}}+e^{-\bm{U}_{1}-\bm{Y}_{1}+\bm{Y}_{2}}+e^{-\bm{U}_{1}+\bm{Y}_{1}+2\bm{U}_{2}-\bm{Y}_{2}}\right]w,\\
t_{21}{\cal S}_2e^{-\bm{U}_{1}+\bm{Y}_{1}}{\cal S}_2^{-1}t_{21}^{-1}w&=&e^{-\bm{U}_{1}+\bm{Y}_{1}-\bm{Y}_{2}}w.
\ee

In order to compute 
\be
t_{21}{\cal S}_2e^{\bm{U}_{1}-\bm{Y}_{1}}{\cal S}_2^{-1}t_{21}^{-1}w,
\ee
A regularization is needed, since the earlier argument only ensures ${\cal S}_2^{-1}t_{21}^{-1}w$ be analytic for $\im(\mu_1)<\re(b)$, whereas $e^{\bm{U}_1}$ sends $\phi(\mu_1,m_1)$ to $\phi(\mu_1+ib,m-1)$. We introduce a regularization parameter $\epsilon>0$ and consider 
\be
f_{\epsilon}=\varphi\left(-y_{1}(\mu_{1}-i\epsilon)y_{2}^{{-1}}\right)^{{-1}}\mathcal{V}_{2}^{-1}w\left(\mu,m,\mu_{1}-i\epsilon,m_{1}\right)
\ee
which is analytic for ${\rm Im}\left(\mu_{1}\right)<\mathrm{Re}(b)+\epsilon$ if $\mu\in \R$.
\be
&&\bm{u}_{1}\bm{y}_{1}^{-1}{\cal S}_{2}^{-1}\mathcal{V}_{2}f_{\epsilon}\nonumber\\
&=&\sum_{n,m\in\Z/N\Z}\int d\nu e^{\frac{2\pi i}{N}\left(\mu_{2}\nu-m_{2}n\right)}(-1)^{n^{2}}e^{\frac{\pi i}{N}\left(\nu^{2}-n^{2}\right)}\int_{\mathbb{R}}d\mu\, e^{-\frac{2\pi i}{N}(\mu\nu-mn)}\nonumber\\
&&\bm{q}^{-2}y_{1}^{-1}\varphi\left(-\bm{q}^{2}y_{1}(\mu_{1}-i\epsilon)y_{2}y^{-1}\right)^{-1}\mathcal{V}_{2}^{-1}w\left(\mu,m,\mu_{1}+ib+\mu_{2}-i\epsilon,m_{1}-1+m_{2}\right)\nonumber\\
&=&{\cal S}_{2}^{-1}\lt[{y}_{2}\sum_{n,m\in\Z/N\Z}\int d\nu e^{\frac{2\pi i}{N}\left(\mu_{2}\nu-m_{2}n\right)}(-1)^{n^{2}}e^{\frac{\pi i}{N}\left(\nu^{2}-n^{2}\right)}e^{\frac{2\pi i}{N}(-ib\nu-n)}\int_{\mathbb{R}}d\mu e^{-\frac{2\pi i}{N}(\mu\nu-mn)}\rt.\nonumber\\
&&\lt.\bm{q}^{-2}y_{1}^{-1}\varphi\left(-y_{1}(\mu_{1}-i\epsilon)y^{-1}\right)^{-1}\lt(\bm{u}\bm{u}_{1}\mathcal{V}_{2}^{-1}w\rt)\left(\mu,m,\mu_{1}-i\epsilon,m_{1}\right)\rt]
\ee
where we have shift the $\mu$-contour from $\R$ to $\R+ib$ in the second step (recall \eqref{repuandy} for the action of $\bm u$, while replacing $(\mu,m)$ by $(\mu_a,m_a)$ for $\bm u_a$, $a=1,2$). The poles in $\varphi(\cdot)^{-1}$ are avoided by the $i\epsilon$ regularization.
\be
\psi(\nu,n;\mu_1,m_1)&=&\sum_{m\in\Z/N\Z}(-1)^{n^{2}}e^{\frac{\pi i}{N}\left(\nu^{2}-n^{2}\right)}e^{\frac{2\pi i}{N}(-ib\nu-n)}\int_{\mathbb{R}}d\mu e^{-\frac{2\pi i}{N}(\mu\nu-mn)}\nonumber\\
&&\bm{q}^{-2}y_{1}^{-1}\varphi\left(-y_{1}(\mu_{1}-i\epsilon)y^{-1}\right)^{-1}\lt(\bm{u}\bm{u}_{1}\mathcal{V}_{2}^{-1}w\rt)\left(\mu,m,\mu_{1}-i\epsilon,m_{1}\right)
\ee
belongs to $\ch\otimes\ch$ for $\nu\in\R$ and is entire in $\nu$, and $\psi(\nu+ib,n-1;\mu_1,m_1)$ also belongs to $\ch\otimes\ch$. We have
\be
\bm{y}_{2}\sum_{n\in\Z/N\Z}\int d\nu e^{\frac{2\pi i}{N}\left(\mu_{2}\nu-m_{2}n\right)}\psi(\nu,n;\mu_1,m_1)=\sum_{n\in\Z/N\Z}\int d\nu e^{\frac{2\pi i}{N}\left(\mu_{2}\nu-m_{2}n\right)}\psi(\nu+ib,n-1;\mu_1,m_1)\nonumber
\ee
Then we obtain
\be
\bm{u}_{1}\bm{y}_{1}^{-1}{\cal S}_{2}^{-1}\mathcal{V}_{2}f_{\epsilon}=-\cs_2^{-1}\mathcal{V}_{2}y_{1}^{-1}y_{2}\varphi\left(-y_{1}(\mu_{1}-i\epsilon)y_{2}^{-1}\right)^{-1}\bm{q}^{-1}\lt(\bm{u}_{1}\bm{u}_{2}\mathcal{V}_{2}^{-1}w\rt)\left(\mu,m,\mu_{1}-i\epsilon,m_{1}\right).\nonumber
\ee
We denote by
\be
\phi_{\epsilon}\equiv \mathcal{V}_{2}^{-1}{\cal S}_2\bm{u}_{1}\bm{y}_{1}^{-1}{\cal S}_2^{-1}\mathcal{V}_{2}f_{\epsilon}=-y_{1}^{-1}y_{2}\varphi\left(-y_{1}(\mu_{1}-i\epsilon)y_{2}^{-1}\right)^{-1}\bm{q}^{-1}\lt(\bm{u}_{1}\bm{u}_{2}\mathcal{V}_{2}^{-1}w\rt)\left(\mu,m,\mu_{1}-i\epsilon,m_{1}\right).\nonumber
\ee
We can remove the regulator $\epsilon$: $f_\epsilon\to f\equiv f_0$ and $\phi_{\epsilon}\to \phi\equiv \phi_0$ converge respectively in the sense of Hilbert space norm, so ${\cal S}^{-1}_2\mathcal{V}_{2}f_{\epsilon}\to{\cal S}_2^{-1}\mathcal{V}_{2}f$ and ${\cal S}_2^{-1}\mathcal{V}_{2}\phi_{\epsilon}\to\psi\equiv{\cal S}_2^{-1}\mathcal{V}_{2}\phi$ converge. $\bm{u}_{1}\bm{y}_{1}^{-1}$ is closed operator and $\bm{u}_{1}\bm{y}_{1}^{-1}{\cal S}_2^{-1}\mathcal{V}_{2}f_{\epsilon}={\cal S}_2^{-1}\mathcal{V}_{2}\phi_{\epsilon}$. It implies that ${\cal S}_2^{-1}\mathcal{V}_{2}f$ is in the domain of $\bm{u}_{1}\bm{y}_{1}^{-1}$ and $\bm{u}_{1}\bm{y}_{1}^{-1}{\cal S}_2^{-1}\mathcal{V}_{2}f={\cal S}_2^{-1}\mathcal{V}_{2}\phi$. Then we have the result
\be
t_{21}{\cal S}_2e^{\bm{U}_{1}-\bm{Y}_{1}}{\cal S}_2^{-1}t_{21}^{-1}w=e^{\bm{U}_{1}-\bm{Y}_{1}+\bm{Y}_{2}}w
\ee

It is straight-forward to check that 
\be
t_{21}{\cal S}_{2}e^{\bm{Y}_{2}}{\cal S}_{2}^{-1}t_{21}^{-1}w&=&\left(\bm{y}_{2}+\bm{u}_{2}\bm{y}_{1}\right)w,\\
t_{21}{\cal S}_2e^{-\bm{U}_{2}}{\cal S}_2^{-1}t_{21}^{-1}w&=&\bm{u}_{1}^{-1}\bm{u}_{2}^{-1}w
\ee

In order to compute
\be
t_{21}{\cal S}_2e^{\bm{U}_{2}}{\cal S}_2^{-1}t_{21}^{-1}w,
\ee
we express $\mathcal{V}_{2}f$ by shifting the $\mu$-contour:
\be
{\cal S}_2^{-1}\mathcal{V}_{2}f&=&\sum_{n,m\in\Z/N\Z}\int d\nu e^{\frac{2\pi i}{N}\left(\mu_{2}\nu-m_{2}n\right)}(-1)^{n^{2}}e^{\frac{\pi i}{N}\left(\nu^{2}-n^{2}\right)}\int_{\mathbb{R}}d\mu\, e^{-\frac{2\pi i}{N}((\mu+ib)\nu-(m-1)n)}\nonumber\\
&&\varphi\left(-\bm{q}^{-2}y_{1}y_{2}y^{-1}\right)^{-1}\lt(\bm{u}\mathcal{V}_{2}^{-1}w\rt)\left(\mu,m,\mu_{1}+\mu_{2},m_{1}+m_{2}\right)
\ee
${\cal S}_2^{-1}\mathcal{V}_{2}f $ is entire in $\mu_2$, so 
\be
\bm{u}_{2}{\cal S}_2^{-1}\mathcal{V}_{2}f&=&{\cal S}_2^{-1}\lt[\sum_{n,m\in\Z/N\Z}\int d\nu e^{\frac{2\pi i}{N}\left(\mu_{2}\nu-m_{2}n\right)}(-1)^{n^{2}}e^{\frac{\pi i}{N}\left(\nu^{2}-n^{2}\right)}\int_{\mathbb{R}}d\mu\, e^{-\frac{2\pi i}{N}(\mu\nu-mn)}\rt.\nonumber\\
&&\varphi\left(-y_{1}y^{-1}\right)^{-1}\lt(\bm{u}\bm{u}_{1}\mathcal{V}_{2}^{-1}w\rt)\left(\mu,m,\mu_{1},m_{1}\right)\Bigg]\nonumber\\
&=&{\cal S}_2^{-1}\mathcal{V}_{2}\varphi\left(-y_{1}y_{2}^{-1}\right)^{-1}\bm{u}_{1}\bm{u}_{2}\mathcal{V}_{2}^{-1}w 
\ee
Therefore, ${\cal S}_2^{-1}\mathcal{V}_{2}f$ is in the domain of $\bm u_2$ and 
\be
t_{21}{\cal S}_2e^{\bm{U}_{2}}{\cal S}_2^{-1}t_{21}^{-1}w=\bm{u}_{1}\bm{u}_{2}w.
\ee

The above computation implies
\be
t_{21}{\cal S}_2Q'_1{\cal S}^{-1}_2t_{21}^{-1}w=Q_{21}w,\qquad t_{21}{\cal S}_2K_2^{\pm1}{\cal S}_2^{-1}t_{21}^{-1}w=(\Delta K^{\pm 1})_{21}w,\qquad t_{21}{\cal S}_2F_2{\cal S}_2^{-1}t_{21}^{-1}w=(\Delta F)_{21}w,\nonumber
\ee
for any $w\in\sw_2$. Since $Q_{21},(\Delta K^{\pm 1})_{21},(\Delta F)_{21}$ leave $\sw_2$ invariant, these relations can generalize to polynomials of $Q_{21},(\Delta K^{\pm 1})_{21},(\Delta F)_{21}$: For any $\xi$ polynomial of $Q,K^{\pm1},F$, we have 
\be
t_{21}{\cal S}_2\L(\xi){\cal S}^{-1}_2t_{21}^{-1}w=(\Delta\xi)_{21}w,
\ee
where $\L(\xi)$ is the polynomial $\xi$ with the substitution: $Q\mapsto Q'_1,K^{\pm1}\mapsto K_2^{\pm 1},F\mapsto F_2$. 

Define the operator $*$-algebra $\bf{L}$ generated by polynomials $\Delta\xi$ and the tilded partners on $\ch\otimes\ch$, and define the Schwartz space $\css_{\bf L}\subset \ch \otimes\ch$ as the domain of $\bf{L}$: $f\in\cs_{\bf L}$ if and only if the functional $w\mapsto \langle f\mid(\Delta\xi)_{21} w\rangle$ is continuous in $w\in\sw_2$ for all $\Delta\xi\in\mathbf{L}$ (equivalently, $\sw_2$ can be replaced by any dense domain in $\ch\otimes\ch$). The Schwartz space $\css_{\bf L}$ is invariant under the action of $\bf{L}$ and is dense in $\ch\otimes\ch$. $\css_{\bf L}$ is a Fr\'echet space with a family of seminorms $\Vert f \Vert_{n}=\Vert (\Delta\xi_n)f\Vert$ by choosing a basis $\Delta\xi_n\in\bf{L}$. By the same method of proving Proposition 5.8 in \cite{2008InMat.175..223F}, we can show that $\sw_2$ is dense in $\css_{\bf L}$ with respect to the Fr\'echet topology.

In a similar way, we define the operator $*$-algebra $\L({\bf L})$ generated by $\L(\xi)$ and the tilded partners and define the corresponding Schwartz space $\css_{\L(\mathbf L)}\subset \ch \otimes\ch$ as the common domain of $\L(\xi)$. 

Since $\sw_2$ is dense in $\css_{\bf L}$, there is a sequence $w_i\in\sw_2$ converging to $f\in\css_{\bf L}$ in the sense of the Fr\'chet topology. This means that $w_i\to f$ and $(\Delta\xi)_{21}w_i\to (\Delta\xi)_{21}f $ in the sense of Hilbert space norm. The following quantity is clearly a continuous functional of $w'\in\sw_2$ \cite{2008InMat.175..223F}:
\be
&&\lag {\cal S}_2^{-1}t_{21}^{-1}(\Delta\xi)_{21} f\mid w' \rag=\lag (\Delta\xi)_{21} f\mid t_{21}{\cal S}_2 w' \rag=\lim_{i\to\infty}\lag (\Delta\xi)_{21} w_i\mid t_{21}{\cal S}_2 w' \rag\nonumber\\
&=&\lim_{i\to\infty}\lag \L(\xi) {\cal S}_2^{-1}t_{21}^{-1} w_i\mid  w' \rag=\lim_{i\to\infty}\lag {\cal S}_2^{-1}t_{21}^{-1} w_i\mid \L(\xi) ^\dagger w' \rag=\lag {\cal S}_2^{-1}t_{21}^{-1} f\mid \L(\xi) ^\dagger w' \rag
\ee
This implies ${\cal S}_2^{-1}t_{21}^{-1} f\in \css_{\L(\mathbf L)}$, i.e. ${\cal S}_2^{-1}t_{21}^{-1}: \css_{\mathbf L}\to \css_{\L(\mathbf L)}$, and 
\be
t_{21}{\cal S}_2\L(\xi){\cal S}^{-1}_2t_{21}^{-1}f=(\Delta\xi)_{21}f.
\ee
Conversely, for any $F\in  \css_{\L(\mathbf L)}$, the following quantity is a continuous functional of $w'\in\sw_2$:
\be
\langle t_{21}{\cal S}_2 \L(\xi)^\dagger F \mid w'\rangle&=&\langle F \mid \L(\xi) {\cal S}^{-1}_2t_{21}^{-1}w'\rangle=\langle F \mid  {\cal S}^{-1}_2t_{21}^{-1}(\Delta\xi)_{12} w'\rangle\nonumber\\
&=&\langle t_{21}{\cal S}_2 F \mid (\Delta\xi)_{12} w'\rangle
\ee
so $t_{21}{\cal S}_2 F\in \css_{\mathbf L}$. Therefore, ${\cal S}_2^{-1}t_{21}^{-1}: \css_{\mathbf L}\to \css_{\L(\mathbf L)}$ is bijective.

Now we enlarge the operator algebra $\L(\mathbf{L})$ by taking into account $F_2^{-1}\propto \bmy_2^{-1}$: We define $\L(\mathbf{L}')$ to be the algebra of polynomials of $Q_1',K_2^{\pm 1},F_2^{\pm1}$ and the tilded partners. The corresponding Schwartz space $\css_{\L(\mathbf{L}')}$ is the common domain of $\L(\mathbf{L}')$,  $\css_{\L(\mathbf{L}')}$ is dense in $\ch\otimes\ch$, and $\css_{\L(\mathbf{L}')}\subset \css_{\L(\mathbf{L})}$. For any $\phi\in\css_{\L(\mathbf{L}')} $, $F_2^{-1}\phi \in\css_{\L(\mathbf{L}')}\subset \css_{\L(\mathbf{L})} $, and there exist $f\in \css_{\bf L}$ such that ${\cal S}^{-1}_2t_{21}^{-1}f=F_2^{-1}\phi$. Let $Q_{21}=-\left(\bm{q}-\bm{q}^{-1}\right)^{2}(\Delta E)_{21}(\Delta F)_{21}-\bm{q}^{-1}(\Delta K)_{21}-\bm{q}(\Delta K)_{21}^{-1}$ act on $f$, and we obtain
\be
{\cal S}_{2}^{-1}t_{21}^{-1}(\Delta E)_{21}t_{21}{\cal S}_{2}\phi=-\frac{\left[Q_{1}^{\prime}+\bm{q}^{-1}K_{2}+\bm{q}K_{2}^{-1}\right]F_{2}^{-1}\phi}{\left(\bm{q}-\bm{q}^{-1}\right)^{2}}\equiv E_2 \phi
\ee
The operator algebra of polynomials $E_2,F_2,K_2^{\pm1}$ and their tilded partners form a representation of $\suquqt$ on $\css_{\L(\mathbf{L}')}$. $t_{21}{\cal S}_{2}$ maps $\css_{\L(\mathbf{L}')}$ to the domain of the polynomials of $(\Delta E)_{21}$, $(\Delta F)_{21}$, and $(\Delta K^{\pm1})_{21}$. 

We denote by $\mathbf{L}'$ the algebra of polynomials of $(\Delta E)_{21}$, $(\Delta F)_{21}$, and $(\Delta K^{\pm1})_{21}$ and define the corresponding Schwartz space $\css_{\mathbf{L}'}$ as the dense domain of $\mathbf{L}'$. We have shown above that ${\cal S}_{2}^{-1}t_{21}^{-1}\css_{\L(\mathbf{L}')}\subset \css_{\mathbf{L}'}$. For any $f\in\css_{\mathbf{L}'} $, $\phi\in\css_{\L(\mathbf{L}')}$ and any $\zeta$ a polynomial of $E,F,K^{\pm1}$ and the tilded partners,
\be
\langle f \mid (\Delta\zeta)_{21} t_{21} {\cal S}_{2}\phi \rangle=\langle f \mid t_{21}{\cal S}_2 \L(\zeta)\phi \rangle=\langle {\cal S}_2^{-1}t_{21}^{-1} f \mid \L(\zeta)\phi \rangle
\ee
is a continuous linear functional of $\phi\in \css_{\L(\mathbf{L}')}$. Therefore ${\cal S}_2^{-1}t_{21}^{-1} f \in \css_{\L(\mathbf{L}')}$, and thus ${\cal S}_2^{-1}t_{21}^{-1}$ is bijective from $\css_{\mathbf{L}'}$ to $\css_{\L(\mathbf{L}')}$.

The operator domains in Lemma \ref{lemma32} can be understood in a similar way: The unitary operator $C_{1}^{-1}$ can be expressed as
\be
C_{1}^{-1}=\mathcal{F}\varphi\left(\lambda_{2}^{-1}\bm{u}^{-1}\right){\cal F}^{-1}\frac{\varphi\left(\lambda_{1}\bm{u}^{-1}\right)}{\varphi\left(\lambda_{1}\bm{u}\right)}
\ee
and it is convenient to use the representation $\bm{u}\psi\left(\nu,n\right)=e^{-\frac{2\pi i}{N}\left[-ib\nu-n\right]}\psi\left(\nu,n\right)$, $\bm{y}\psi\left(\nu,n\right)=\psi\left(\nu+ib,n-1\right)$. For any $w\in\sw$, 
\be
F&\equiv&\mathcal{F}\varphi\left(\lambda_{2}^{-1}\bm{u}^{-1}\right){\cal F}^{-1}\frac{\varphi\left(\lambda_{1}\bm{u}^{-1}\right)}{\varphi\left(\lambda_{1}\bm{u}\right)}w\\
&=&\sum_{n_{1},n_{2}\in\Z/N\Z}\int d\nu_{2}e^{-\frac{2\pi i}{N}\left(\nu\nu_{2}-nn_{2}\right)}\varphi\left(\lambda_{2}^{-1}\bm{u}_{2}^{-1}\right)\int d\nu_{1}e^{\frac{2\pi i}{N}\left(\nu_{2}\nu_{1}-n_{2}n_{1}\right)}\frac{\varphi\left(\lambda_{1}\bm{u}_{1}^{-1}\right)}{\varphi\left(\lambda_{1}\bm{u}_{1}\right)}w\left(\nu_{1},n_{1}\right)\nonumber
\ee
and we define $F_{\epsilon>0}$ by the following substitution in the integrand, which regularizing the quantum dilogarithm
\be
\frac{\varphi\left(\lambda_{1}\bm{u}_{1}^{-1}\right)}{\varphi\left(\lambda_{1}\bm{u}_{1}\right)}\mapsto \frac{\varphi\left(\lambda_{1}\bm{u}_{1}^{-1}\right)}{\varphi\left(\lambda_{1}\bm{u}_{1}\right)}\Big|_{\nu_{1}\to\nu_{1}+i\epsilon},\qquad \varphi\left(\lambda_{2}^{-1}\bm{u}_{2}^{-1}\right)\mapsto\varphi\left(\lambda_{2}^{-1}\bm{u}_{2}^{-1}\right)\Big|_{\nu_{2}\to\nu_{2}+i\epsilon}, 
\ee
so that they are analytic for $\mathrm{Im}(\nu_{1})>-\mathrm{Re}(b)-\epsilon$ and $\mathrm{Im}(\nu_{2})>-\mathrm{Re}(b)-\epsilon$ respectively. As $\epsilon\to0$, $F_\epsilon\to F$ in the sense of Hilbert space norm. As a result, for $\mathrm{Im}(\nu)>-\mathrm{Re}(b)-\epsilon$, $F_\epsilon(\nu,n)$ is analytic and is a Schwartz function of $\re(\nu)$ with the asymptotic behavior $e^{-\frac{2\pi}{N}\left(\mathrm{Re}(b)+\epsilon\right)\mathrm{Re}(\nu)}$ as $\mathrm{Re}(\nu)\to\infty$ and decaying faster than $e^{a\re(\nu)}$ for any $a>0$ as $\mathrm{Re}(\nu)\to-\infty$. Therefore, $F_\epsilon$ is in the domain of $\bm{u}^\alpha\bmy^\beta$ with $\a,\b=\pm 1$. Then following the formal proof of Lemma \ref{lemma32}, one can show that $Q_1'F_\epsilon$ converges to $F'=C_1^{-1}Q_1'' w\in \ch$ in the Hilbert space norm as $\epsilon\to 0$. Therefore, $F=C^{-1}_1 w$ is in the domain of $Q_1'$ and 
\be
C_1 Q_1'C^{-1}_1 w=Q_1''w.
\ee
We denote by $\css_{Q_1'},\css_{Q_1''}\in\ch$ the Schwartz spaces associated to the polynomial algebras of $Q_1'$ and $Q_1''$. By the same argument as the above, $C_1^{-1}$ maps $Q_1''$ to $\css_{Q_1'}$ and is bijective.

\section{Fock-Goncharov coordinate and holonomies}\label{Fock-Goncharov coordinate and holonomies}

\begin{figure}[h]
	\centering
	\includegraphics[width=0.5\textwidth]{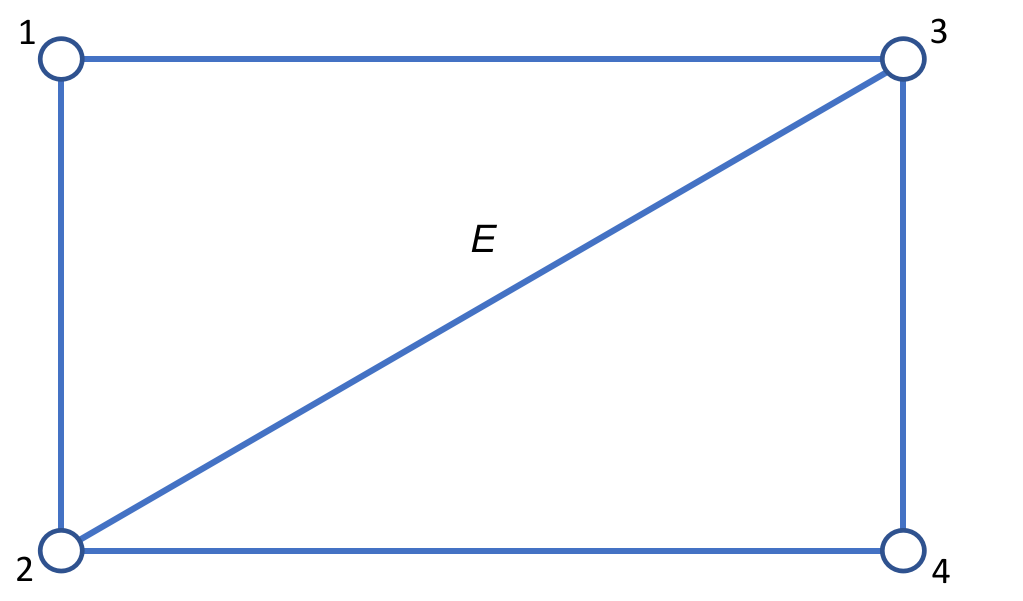}
	\caption{}
	\label{FGcoord}
\end{figure}

A 2-sphere in which $n$ discs are removed is a $n$-holed sphere. We make a 2d ideal triangulation of the $n$-holed sphere such that edges in the triangulation end at the boundary of the holes. For example, the boundary of the ideal tetrahedron is an ideal triangulation of the 4-holed sphere. The 2d ideal triangulation has $3(n-2)$ edges on the $n$-holed sphere. Each edge $E$ associates to a coordinate $z_E$ of the moduli space of framed $\PSlc$ flat connections. A framed flat connection on is a standard flat connection with a choice of $\mathbb{CP}^1$ flat section $s_i$ for each hole $i$. The section $s_i$ obeying the condition $\nabla s_i=0$ ($\nabla$ is the flat connection) and is the eigenvector of monodromy around the hole $i$. $s_i$ associates to the eigenvalue $\l_i=e^{L_i}$ of the monodromy matrix. Given a framed flat connection, $z_E$ is a cross-ratio of 4 flat section $s_1,s_2,s_3,s_4$ associated to the vertices of the quadrilateral containing $E$ as the diagonal (see FIG.\ref{FGcoord}), 
\be
z_E=\frac{\lag s_1\wedge s_2\rag\lag s_3\wedge s_4\rag}{\lag s_1\wedge s_3\rag\lag s_2\wedge s_4\rag}
\ee
where $\lag s_i\wedge s_j\rag$ is an $\Slc$ invariant volume on $\C^2$, and is computed by parallel transporting $s_1,\cdots,s_4$ to a common point inside the quadrilateral by the flat connection. The coordinates $z_E$ for all $E$ are the Fock-Goncharov (FG) coordinates. We often consider a lift of $z_E$ to the logarithmic coordinate $Z_E$ such that $z_E=e^{Z_E}$. At any hole $i$, $z_E$ on the adjacent edges $E$ satisfy
\be
\prod_{E\ \text{at}\ i}(-z_E)=\l_i^2.\label{monodromy000}
\ee
The lifts of the coordinates are chosen such that the logarithm of \eqref{monodromy000} is given by \cite{DGV}
\begin{eqnarray}
	\sum_{E\ \text{at}\ i}(Z_E-\pi i)=2L_i.\label{monodromy001}
\end{eqnarray}
Note that $L_{1,\cdots,3}$ in \eqref{LZ1234} are defined with flipped sign.

$\Slc$ holonomy along any closed path on the $n$-holed sphere can be expressed as $2\times 2$ matrices whose entries are functions of $\{Z_E\}$ by using the ``snake rule'' \cite{DGV}: There are three rules for transporting a {\it snake} -- an arrow pointing from one vertex of the triangle to another with a {\it fin} facing inside the triangle, each corresponds to a matrix as follows. (The inverse transportation of each type corresponds to the inverse of the relevant matrix.
\be
\includegraphics[width=0.25\textwidth]{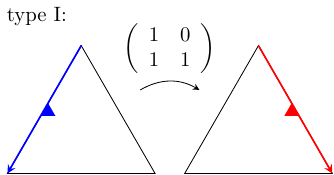}\quad\qquad
\includegraphics[width=0.25\textwidth]{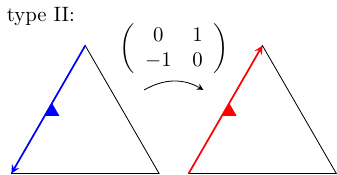}\quad\qquad
\includegraphics[width=0.3\textwidth]{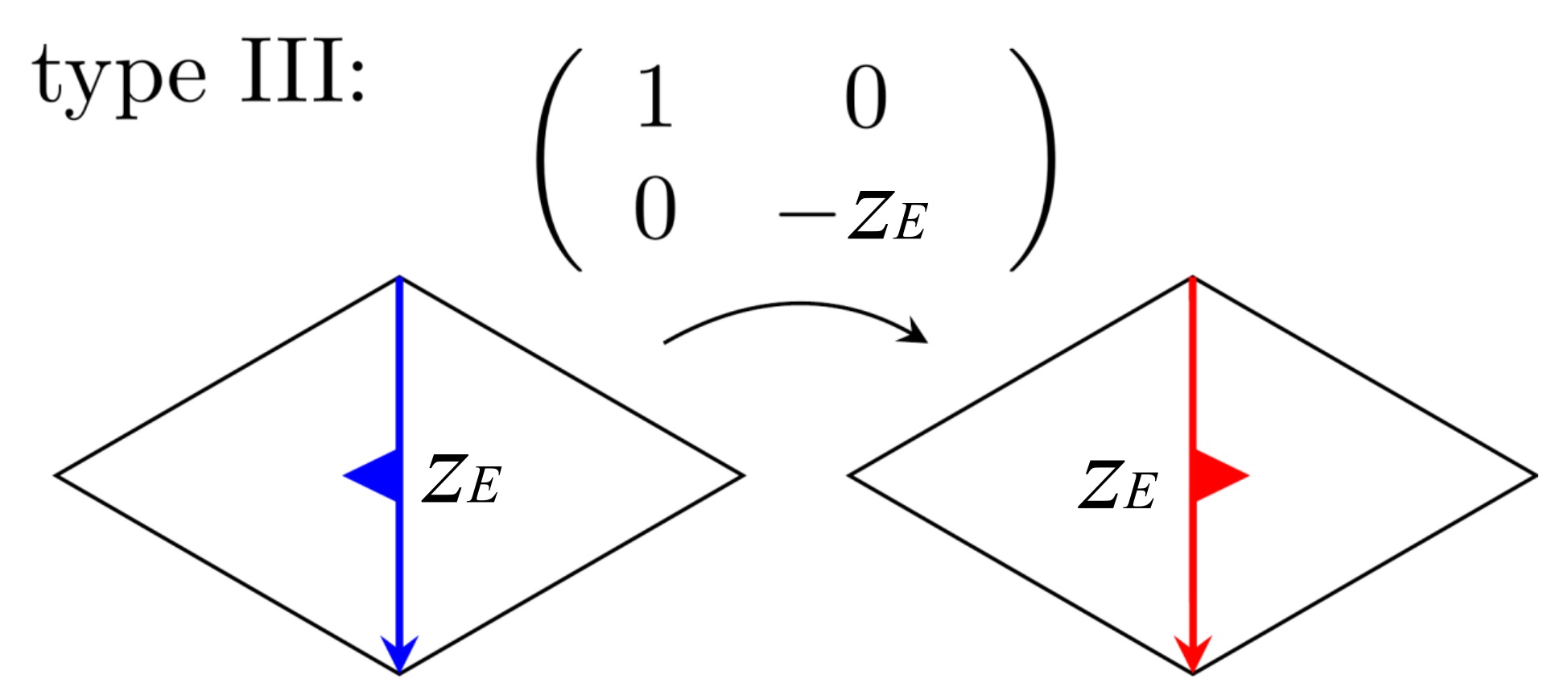}\,.
\label{eq:snake_rule}
\ee
A snake represents a projective basis $(v_1,v_2)$ given by $v_1\in\C$ at the tail of the snake
and $v_2\in\C$ at the head of the snake, such that either $v_1 + v_2$ (type I: blue) or
$v_1-v_2$ (type I: red) at the third vertex of the triangle. Type I and II correspond to transporting a snake within a triangle and III correspond to moving a snake from one triangle to its adjacent triangle. The transformation matrix acts on the projective basis $(v_1,v_2)^T$ by left multiplication. Any holonomy $H(\g)$ of a closed loop $\g$ can be calculated by multiplying the matrices from right to left corresponding to moving a snake along the loop. The holonomy matrix resulting from the snake rule is not immediately $\Slc$ but rather understood as $\mathrm{PGL}(2,\C)$. The $\Slc$ holonomy is obtained up to $\pm$ sign by a lift that can conveniently chosen by normalizing the Type III matrix
\be
\text{Type III:}\qquad  \left(\begin{array}{cc}
	e^{-\frac{Z_E-\pi i}{2}} & 0 \\
	0 & e^{\frac{Z_E-\pi i}{2}}
	\end{array}\right)\equiv E(Z_E).
\ee 
For any closed path $\g_i$ around a single hole $i$, the $\pm$ sign of $H(\g_i)$ can be determined by requiring the trace of $H(\g_i)$ to be
\be
\tr(H(\g_i))=e^{\frac{1}{2}\sum_{E\ \text{at}\ i}(Z_E-i\pi)}+e^{-\frac{1}{2}\sum_{E\ \text{at}\ i}(Z_E-i\pi)}=e^{L_i}+e^{-L_i}.
\ee
This requirements is consistent with \eqref{monodromy001} and the eigenvalue of the monodromy matrix when discussing $s_i$. The fundamental group of the $n$-holed sphere is generated by $\{\g_i\}_{i=1,\dots,n}$, so $\{H(\g_i)\}_{i=1,\dots,n}$ determine all holonomies of closed paths.



\begin{figure}[h]
	\centering
	\includegraphics[width=0.5\textwidth]{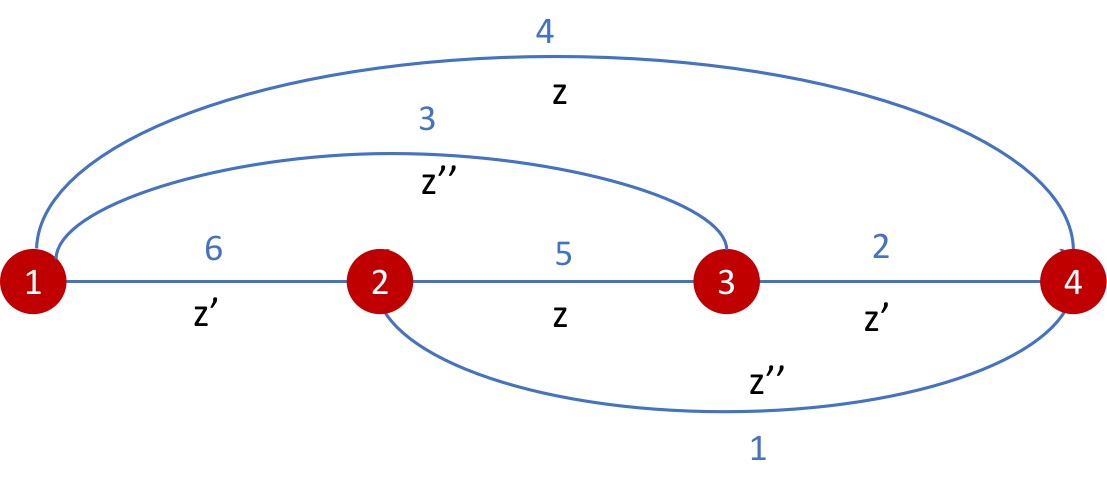}
	\caption{The tetrahedral ideal triangulation of 4-holed sphere and the FG coordinates $z,z',z''$.}
	\label{zzpzpp}
\end{figure}

We consider the 4-holed sphere and the ideal triangulation in FIG.\ref{zzpzpp} as an example. The ideal triangulation is tetrahedral since it is the boundary of an ideal tetrahedron. We denote by $\g_i$ a loop around the hole $i$ oriented counter-clockwisely. All $\g_i$ share the same base point represented by a snake pointing from the 4th hole to the 2nd hole along the edge 1 with the fin inside the triangle with vertices $2,3,4$.
\be
H(\gamma_4)&=&\left(
\begin{array}{cc}
 1 & 0 \\
 1 & 1 \\
\end{array}
\right)E\left(Z_2\right)\left(
\begin{array}{cc}
 1 & 0 \\
 1 & 1 \\
\end{array}
\right)E\left(Z_4\right)\left(
\begin{array}{cc}
 1 & 0 \\
 1 & 1 \\
\end{array}
\right)E\left(Z_1\right),\nonumber\\
H(\g_3)&=&-\left(
\begin{array}{cc}
 0 & 1 \\
 -1 & 0 \\
\end{array}
\right)\left(
\begin{array}{cc}
 1 & 0 \\
 1 & 1 \\
\end{array}
\right)^{-1}\left(
\begin{array}{cc}
 0 & 1 \\
 -1 & 0 \\
\end{array}
\right)E\left(Z_5\right)\left(
\begin{array}{cc}
 1 & 0 \\
 1 & 1 \\
\end{array}
\right)E\left(Z_3\right)\left(
\begin{array}{cc}
 1 & 0 \\
 1 & 1 \\
\end{array}
\right)\left(
\begin{array}{cc}
 0 & 1 \\
 -1 & 0 \\
\end{array}
\right)^{-1}E\left(Z_2\right)^{-1}\left(
\begin{array}{cc}
 1 & 0 \\
 1 & 1 \\
\end{array}
\right)^{-1},\nonumber\\
H(\g_2)&=&\left(
\begin{array}{cc}
 0 & 1 \\
 -1 & 0 \\
\end{array}
\right)E\left(Z_1\right)\left(
\begin{array}{cc}
 1 & 0 \\
 1 & 1 \\
\end{array}
\right)E\left(Z_6\right)\left(
\begin{array}{cc}
 1 & 0 \\
 1 & 1 \\
\end{array}
\right)E\left(Z_5\right)\left(
\begin{array}{cc}
 1 & 0 \\
 1 & 1 \\
\end{array}
\right)\left(
\begin{array}{cc}
 0 & 1 \\
 -1 & 0 \\
\end{array}
\right)^{-1},\nonumber\\
H(\g_1)&=&\left(
	\begin{array}{cc}
	 0 & 1 \\
	 -1 & 0 \\
	\end{array}
	\right)E\left(Z_1\right)\left(
	\begin{array}{cc}
	 1 & 0 \\
	 1 & 1 \\
	\end{array}
	\right)\left(
	\begin{array}{cc}
	 0 & 1 \\
	 -1 & 0 \\
	\end{array}
	\right)^{-1}\left(
	\begin{array}{cc}
	 1 & 0 \\
	 1 & 1 \\
	\end{array}
	\right)E\left(Z_4\right)\left(
	\begin{array}{cc}
	 1 & 0 \\
	 1 & 1 \\
	\end{array}
	\right)E\left(Z_3\right)\left(
	\begin{array}{cc}
	 1 & 0 \\
	 1 & 1 \\
	\end{array}
	\right)E\left(Z_6\right)\nonumber\\
	&& \left(
	\begin{array}{cc}
	 0 & 1 \\
	 -1 & 0 \\
	\end{array}
	\right)\left(
	\begin{array}{cc}
	 1 & 0 \\
	 1 & 1 \\
	\end{array}
	\right)^{-1}E\left(Z_1\right){}^{-1}\left(
	\begin{array}{cc}
	 0 & 1 \\
	 -1 & 0 \\
	\end{array}
	\right)^{-1}.
\ee
They satisfy $\tr(H(\g_i))=e^{L_i}+e^{-L_i}$ and
\be
H(\g_1)H(\g_2)H(\g_3)H(\g_4)=1.
\ee


When the 4-holed sphere is the boundary of an ideal tetrahedron, all $H(\g_i)$ are constrained to the identity matrix, since the connection is flat inside the tetrahedron. In this case, all $L_i$ vanishes, i.e.
\be
Z_1+Z_2+Z_4=Z_3+Z_4+Z_6=Z_1+Z_5+Z_6=Z_2+Z_3+Z_5=3\pi i
\ee
so we can parametrize $z_e$ by calling the FG coordinates $z,z',z''$ occurring in the same counter-clockwise order around any hole, equal on opposite edges, and satisfying \footnote{We relabel $z_1=z_3= z'',\ z_2=z_6= z',\ z_4=z_5= z$.}
\be
zz'z''=-1.
\ee
The off-diagonal vanishes implies
\be
z''+z^{-1}-1=0.\label{Ltetra}
\ee
$z,z''$ are the symplectic coordinates of the phase space (of flat connections) on the boundary of ideal tetrahedron, with the poisson bracket given by $\{Z'',Z\}=\{Z,Z'\}=\{Z',Z''\}=2$. Eq.\eqref{Ltetra} defines the Lagrangian submanifold $\cl_\Delta$ of the flat connections that can be extended to interior of the tetrahedron.

Let us consider the ideal triangulation in FIG.\ref{degtriangulation}. $\g_i$ is the loop travels around the $i$-th hole counter-clockwisely. All $\g_i$ share the same base point represented by a snake pointing from the 1st hole to the 2nd hole along the edge 6 with the fin inside the triangle with vertices 1,2,4. We obtain 
\be    
H(\g_1)&=&E\left(Z_6\right).\left(
	\begin{array}{cc}
	 1 & 0 \\
	 1 & 1 \\
	\end{array}
	\right)E\left(Z_4\right).\left(
	\begin{array}{cc}
	 1 & 0 \\
	 1 & 1 \\
	\end{array}
	\right),\nonumber\\
H(\g_2)&=&\left(
	\begin{array}{cc}
	 0 & 1 \\
	 -1 & 0 \\
	\end{array}
	\right)^{-1} \left(
	\begin{array}{cc}
	 1 & 0 \\
	 1 & 1 \\
	\end{array}
	\right) E\left(Z_3\right) \left(
	\begin{array}{cc}
	 1 & 0 \\
	 1 & 1 \\
	\end{array}
	\right) E\left(Z_5\right) \left(
	\begin{array}{cc}
	 1 & 0 \\
	 1 & 1 \\
	\end{array}
	\right) E\left(Z_1\right) \left(
	\begin{array}{cc}
	 1 & 0 \\
	 1 & 1 \\
	\end{array}
	\right) E\left(Z_6\right) \left(
	\begin{array}{cc}
	 0 & 1 \\
	 -1 & 0 \\
	\end{array}
	\right),\nonumber\\
H(\g_3)&=&\mathbf{M}_3^{-1}E\left(Z_5\right)\left(
	\begin{array}{cc}
	 1 & 0 \\
	 1 & 1 \\
	\end{array}
	\right) E\left(Z_2\right)\left(
	\begin{array}{cc}
	 1 & 0 \\
	 1 & 1 \\
	\end{array}
	\right)\mathbf{M}_3,\nonumber\\
H(\g_4)&=&\mathbf{M}_4^{-1}E\left(Z_2\right) \left(
	\begin{array}{cc}
	 1 & 0 \\
	 1 & 1 \\
	\end{array}
	\right) E\left(Z_3\right) \left(
	\begin{array}{cc}
	 1 & 0 \\
	 1 & 1 \\
	\end{array}
	\right) E\left(Z_4\right) \left(
	\begin{array}{cc}
	 1 & 0 \\
	 1 & 1 \\
	\end{array}
	\right) E\left(Z_1\right) \left(
	\begin{array}{cc}
	 1 & 0 \\
	 1 & 1 \\
	\end{array}
	\right)\mathbf{M}_4,
\ee    
where 
\be   
\mathbf{M}_3&=&\left(
	\begin{array}{cc}
	 0 & 1 \\
	 -1 & 0 \\
	\end{array}
	\right)^{-1} \left(
	\begin{array}{cc}
	 1 & 0 \\
	 1 & 1 \\
	\end{array}
	\right) E\left(Z_1\right) \left(
	\begin{array}{cc}
	 1 & 0 \\
	 1 & 1 \\
	\end{array}
	\right) E\left(Z_6\right) \left(
	\begin{array}{cc}
	 0 & 1 \\
	 -1 & 0 \\
	\end{array}
	\right),\nonumber\\
\mathbf{M}_4&=&\left(
	\begin{array}{cc}
	 0 & 1 \\
	 -1 & 0 \\
	\end{array}
	\right)^{-1} \left(
	\begin{array}{cc}
	 1 & 0 \\
	 1 & 1 \\
	\end{array}
	\right)\mathbf{M}_3.
\ee   
The set of $H(\g_i)$ satisfies 
\be    
H(\g_1)H(\g_2)H(\g_3)H(\g_4)=1.
\ee

\bibliographystyle{jhep}
\bibliography{muxin.bib}

\end{document}